\newif\ifshort\shortfalse

\ifshort

\documentclass[acmsmall]{acmart}

\else

\documentclass[acmsmall,screen,nonacm,natbib=false]{acmart}
\usepackage[style=acmnumeric,backend=biber]{biblatex}
\newcommand{\citet}{\textcite}
\fi

\newcommand{\AppRef}[2]{%
\ifshort%
Appendix #1~\cite{extended_version}%
\else%
\Cref{#2}%
\fi%
}

\usepackage{amsmath}

\usepackage{amssymb}
\usepackage{semantic}
\usepackage{mathtools}
\usepackage{stmaryrd}
\usepackage{quantikz}
\usepackage{color}
\usepackage{tensor}
\usepackage{graphicx}
\usepackage{circuitikz}
\usepackage{quantikz}
\usepackage{adjustbox}
\usepackage{subcaption}
\usepackage{cleveref}
\usepackage{afterpage}
\usepackage{longtable}
\usepackage{array}
\usepackage{multirow}
\usepackage{listings}
\usepackage{wrapfig}
\usepackage[nounderscore]{syntax}
\usepackage{enumitem}
\usepackage[group-separator={,}]{siunitx}

\makeatletter
\renewcommand{\qwbundle}[2][]{%
  \pgfkeys{/quantikz/gates/.cd,style=,Strike Width=0.08cm,Strike Height=0.12cm,#1}%
  \pgfkeysgetvalue{/quantikz/gates/style}{\qz@style}%
  \pgfkeysgetvalue{/quantikz/gates/Strike Width}{\qz@sw}%
  \pgfkeysgetvalue{/quantikz/gates/Strike Height}{\qz@sh}%
  \expanded{%
    \noexpand\arrow[strike arrow={\qz@sw}{\qz@sh}{\unexpanded{#2}},\qz@style,phantom]{l}%
  }%
}
\makeatother

\newcommand{\CC}{\mathbb C}

\newcommand{\II}{\mathbb I}
\newcommand{\NN}{\mathbb N}

\newcommand{\VV}{\mathbb V}
\newcommand{\WW}{\mathbb W}

\newcommand{\cE}{\mathcal{E}}
\newcommand{\cH}{\mathcal{H}}
\newcommand{\cL}{\mathcal{L}}
\newcommand{\cR}{\mathcal{R}}

\newcommand{\rarr}{\rightarrow}

\newcommand{\parens}[1]{\left(#1\right)}

\newcommand{\braces}[1]{\left\{#1\right\}}

\newcommand{\abs}[1]{\left\lvert #1 \right\rvert}
\newcommand{\norm}[1]{\left\lVert #1 \right\rVert}

\newcommand{\adj}{^\dagger}

\newcommand\doubleplus{\mathbin{+\!+}}
\newcommand\partition{\mathbin{\|}}
\newcommand{\bm}[1]{\begin{bmatrix}#1\end{bmatrix}}

\newcommand{\zero}{\texttt{0}}
\newcommand{\one}{\texttt{1}}

\DeclarePairedDelimiter{\msem}{\llbracket}{\rrbracket}

\newcommand*{\defeq}{\mathrel{\vcenter{\baselineskip0.5ex \lineskiplimit0pt
                     \hbox{\scriptsize.}\hbox{\scriptsize.}}}%
                     =}
\newcommand*{\defeqq}{\mathrel{\vcenter{\baselineskip0.5ex \lineskiplimit0pt
                     \hbox{\scriptsize.}\hbox{\scriptsize.}}}%
                     \mathrel{\vcenter{\baselineskip0.5ex \lineskiplimit0pt
                     \hbox{\scriptsize.}\hbox{\scriptsize.}}}%
                     =}

\newcommand{\subcap}[1]{_{\textsc{#1}}}

\newcommand{\Bit}{\mathrm{Bit}}

\newcommand{\lef}[2]{\texttt{left}_{\color{gray}{#1 \oplus #2}}}
\newcommand{\rit}[2]{\texttt{right}_{\color{gray}{#1 \oplus #2}}}
\newcommand{\pair}[2]{\texttt{(} #1 \texttt{,} #2 \texttt{)}}
\newcommand{\spanning}[2]{\textup{spanning}_{\color{gray} #1}\left( \begin{aligned} #2 \end{aligned} \right)}
\newcommand{\ortho}[2]{\textup{ortho}_{\color{gray} #1}\left( \begin{aligned} #2 \end{aligned} \right)}

\newcommand{\erases}[1]{\textup{erases}_{\color{gray} #1}}
\newcommand{\trycatch}[2]{\texttt{try}\; #1 \; \texttt{catch}\; #2}
\newcommand{\cntrl}[4]{\texttt{ctrl } #1 \tensor*[_{\color{gray}{#2}}]{\left\{ \begin{aligned} #3 \end{aligned} \right\}}{_{\color{gray}{#4}}}}
\newcommand{\pmatch}[3]{\texttt{pmatch } \tensor*[_{\color{gray}{#1}}]{\left\{ \begin{aligned} #2 \end{aligned} \right\}}{_{\color{gray}{#3}}}}
\newcommand{\match}[4]{\texttt{match } #1 \tensor*[_{\color{gray}{#2}}]{\left\{ \begin{aligned} #3 \end{aligned} \right\}}{_{\color{gray}{#4}}}}
\newcommand{\gphase}[2]{\texttt {gphase}_{\color{gray}{#1}}\texttt{(} #2  \texttt{)}}
\newcommand{\rphase}[4]{\texttt {rphase}_{\color{gray}{#1}} \left\{ \begin{aligned} #2 &\mapsto #3 \\ \texttt{else} &\mapsto #4 \end{aligned} \right\}}
\newcommand{\uthree}[3]{\texttt u_{\texttt 3} \texttt{(} #1 \texttt{,} #2 \texttt{,} #3 \texttt{)}}

\newcommand{\had}{\texttt{had}}

\newcommand{\Void}{\texttt{Void}}
\newcommand{\Unit}{\texttt{Unit}}
\renewcommand{\unit}{\texttt{()}}
\newcommand{\Ccptp}{\mathcal{C}\subcap{cptp}}

\newcommand{\Leaf}{\mathrm{Leaf}}

\DeclareMathOperator{\tr}{tr}
\DeclareMathOperator{\size}{size}
\DeclareMathOperator{\height}{height}
\DeclareMathOperator{\enc}{enc}
\DeclareMathOperator{\Span}{span}
\DeclareMathOperator{\un}{un}
\DeclareMathOperator{\iso}{iso}
\DeclareMathOperator{\classical}{classical}
\DeclareMathOperator{\FV}{FV}
\DeclareMathOperator{\dom}{dom}

\DeclareMathOperator{\LEVEL}{\textsc{level}}
\DeclareMathOperator{\FINALMERGE}{\textsc{FinalMerge}}
\DeclareMathOperator{\MERGE}{\textsc{merge}}
\DeclareMathOperator{\inj}{inj}

\DeclareMathOperator{\cptp}{CPTP}

\newcommand{\op}[2]{| #1 \rangle\!\langle #2 |}
\newcommand{\ip}[2]{\langle #1 | #2 \rangle}

\definecolor{codegreen}{rgb}{0,0.6,0}
\definecolor{codegray}{rgb}{0.5,0.5,0.5}
\definecolor{codeblue}{rgb}{0,0,0.9}
\definecolor{codeteal}{rgb}{0.1,0.5,0.5}
\definecolor{codepurple}{rgb}{0.58,0,0.82}
\definecolor{codeorange}{rgb}{0.5,0.2,0}
\definecolor{backcolour}{rgb}{1,1,1}

\lstdefinelanguage{qunity}
{
morekeywords={def,type,end,if,then,else,endif,fail,of,try,catch,lambda,let,in,ctrl,match,pmatch},
morekeywords=[2]{left,right,u3,gphase,rphase,pi,euler,Void,Unit,Bit,Maybe,Array,List,Num,BinaryTree},
morekeywords=[3]{deutsch,plus,minus,had,f,add_const,id,multiply_helper,multiply,grover_iter,grover,equal_superpos,is_odd_sum,not,equal_superpos_list,mod_mult,mod_exp,order_finding,repeated,num_to_state,fst,snd,adjoint,reverse,qft,and,reflect,add,double,gate_1q,controlled_1q,cnot,apply_phase,phase_estimation,multi_and,couple,rotations,cdkm_maj,cdkm_uma,rev_adder_helper,rev_adder,answer},
morecomment=[l]{//},
morecomment=[s]{/*}{*/},
literate= 
    {unit}{()}1 
    {->}{$\to$}2 
    {|>}{$\rhd$}2 
    {oplus}{$\oplus$}2 
    {\&}{{{\color{codeorange}\$}}}1
    {@}{{{\color{codeorange}@}}}1
    {\#}{{{\color{codeteal}\#}}}1
    {\&0}{{{\color{codeblue}\$0}}}2
    {\&1}{{{\color{codeblue}\$1}}}2
    {\&Nothing}{{{\color{codeblue}\$Nothing}}}8
    {@Just}{{{\color{codeblue}@Just}}}5
    {\&ListEmpty}{{{\color{codeblue}\$ListEmpty}}}{10}
    {@ListCons}{{{\color{codeblue}@ListCons}}}9
    {\&Leaf}{{{\color{codeblue}\$Leaf}}}5
    {@Node}{{{\color{codeblue}@Node}}}5
    {'a}{{{\color{codeblue}'a}}}2
    {\#a}{{{\color{codeteal}\#a}}}2
    {\#n}{{{\color{codeteal}\#n}}}2
    {\#m}{{{\color{codeteal}\#m}}}2
    {\#i}{{{\color{codeteal}\#i}}}2
    {\#j}{{{\color{codeteal}\#j}}}2
    {\#k}{{{\color{codeteal}\#k}}}2
    {\#p}{{{\color{codeteal}\#p}}}2
    {\#n\_iter}{{{\color{codeteal}\#n\_iter}}}6
}

\lstdefinestyle{mystyle}{
    backgroundcolor=\color{backcolour},   
    commentstyle=\color{codegreen},
    keywordstyle=\color{codepurple},
    keywordstyle=[2]\color{codeblue},
    keywordstyle=[3]\color{codeorange},
    basicstyle=\ttfamily\normalsize,
    breakatwhitespace=false,         
    breaklines=true,                 
    captionpos=b,                    
    keepspaces=true,
    numbers=left,                    
    numbersep=5pt,                  
    showspaces=false,                
    showstringspaces=false,
    showtabs=false,                  
    tabsize=2,
    numbers=none
}

\lstdefinestyle{small}{
    style=mystyle,
    basicstyle=\ttfamily\footnotesize\fontsize{5}{5}\selectfont
}

\usepackage{letltxmacro}
\newcommand*{\SavedLstInline}{}
\LetLtxMacro\SavedLstInline\lstinline
\DeclareRobustCommand*{\lstinline}{%
  \ifmmode
    \let\SavedBGroup\bgroup
    \def\bgroup{%
      \let\bgroup\SavedBGroup
      \hbox\bgroup
    }%
  \fi
  \SavedLstInline
}

\usepackage{pifont}
\newcommand{\cmark}{\ding{51}}
\newcommand{\xmark}{\ding{55}}

\usepackage{etoolbox} 
\newtoggle{comments}
\toggletrue{comments}

\iftoggle{comments}{
  \newcommand{\fixme}[1]{\textbf{\textcolor{red}{[ Fixme: #1]}}}
  \newcommand{\todo}[1]{\textbf{\textcolor{green}{[ TODO: #1 ]}}}
  \newcommand{\rnr}[1]{\textbf{\textcolor{blue}{[ Robert: #1 ]}}}
  \newcommand{\leo}[1]{\textbf{\textcolor{red}{[ Leo: #1 ]}}}
  \newcommand{\mikhail}[1]{\textbf{\textcolor{orange}{[ Mikhail: #1 ]}}}
}{
  \newcommand{\fixme}[1]{}
  \newcommand{\todo}[1]{}
  \newcommand{\rnr}[1]{}
  \newcommand{\leo}[1]{}
  \newcommand{\mikhail}[1]{}
}



\ifshort

\setcopyright{cc}
\setcctype{by}
\acmDOI{10.1145/3763056}
\acmYear{2025}
\acmJournal{PACMPL}
\acmVolume{9}
\acmNumber{OOPSLA2}
\acmArticle{278}
\acmMonth{10}
\received{2025-03-25}
\received[accepted]{2025-08-12}

\else

\setcopyright{cc}
\setcctype{by}

\fi

\begin{document}

\title{Compositional Quantum Control Flow with Efficient Compilation in Qunity}

\author{Mikhail Mints}
\orcid{0009-0004-4508-353X}
\affiliation{%
  \institution{California Institute of Technology}
  \city{Pasadena}
  \country{USA}
}
\email{mmints@caltech.edu}

\author{Finn Voichick}
\orcid{0000-0002-1913-4178}
\affiliation{%
  \institution{University of Maryland}
  \city{College Park}
  \country{USA}
}
\email{finn@umd.edu}

\author{Leonidas Lampropoulos}
\orcid{0000-0003-0269-9815}
\affiliation{%
  \institution{University of Maryland}
  \city{College Park}
  \country{USA}
}
\email{leonidas@umd.edu}

\author{Robert Rand}
\orcid{0000-0001-6842-5505}
\affiliation{%
  \institution{University of Chicago}
  \city{Chicago}
  \country{USA}
}
\email{rand@uchicago.edu}

\begin{abstract}
Most existing quantum programming languages are based on the quantum circuit model of computation, as higher-level abstractions are particularly challenging to implement---especially ones relating to quantum control flow. The Qunity language, proposed by Voichick et al., offered such an abstraction in the form of a quantum control construct, with great care taken to ensure that the resulting language is still realizable. However, Qunity lacked a working implementation, and the originally proposed compilation procedure was very inefficient, with even simple quantum algorithms compiling to unreasonably large circuits.

In this work, we focus on the efficient compilation of high-level quantum control flow constructs, using Qunity as our starting point. We introduce a wider range of abstractions on top of Qunity's core language that offer compelling trade-offs compared to its existing control construct. We create a complete implementation of a Qunity compiler, which converts high-level Qunity code into the quantum assembly language OpenQASM 3. We develop optimization techniques for multiple stages of the Qunity compilation procedure, including both low-level circuit optimizations as well as methods that consider the high-level structure of a Qunity program, greatly reducing the number of qubits and gates used by the compiler.
\end{abstract}

\maketitle

\section{Introduction}

In recent years, many \textit{high-level quantum programming
languages} have been proposed, aiming to allow algorithm implementers
to work at a higher level of abstraction compared to the quantum
circuit model of computation. Languages such as QML~\cite{QML},
Silq~\cite{Silq}, Tower~\cite{tower}, and Qunity~\cite{Voichick_2023} walk a
fine line between providing constructs familiar from classical
computing and being realizable in quantum hardware. High-level quantum control flow constructs are particularly challenging since quantum programs must be compilable to fixed-length quantum circuits generated by classical computation~\cite{yuan-limits}. 
As a result, proposed high-level quantum languages offer branching
statements with significant restrictions on the expressive powers
of their individual branches.

Among these languages, Qunity stands out for its unique handling of control flow that emphasizes compositionality. Qunity naturally extends classical programming constructs into the domain of quantum computation and admits a \emph{compositional denotational semantics} defined in terms of quantum operators (acting on state vectors) and superoperators (acting on density matrices), allowing for more flexible design than what is possible with unitary gate-based quantum circuits. Qunity's \texttt{ctrl} construct offers a generalization of pattern matching that can coherently control on the output of an arbitrary (possibly irreversible) quantum or classical computation, relying on the BQP subroutine theorem \cite{Watrous2009}. The construct matches the scrutinee against a set of reversible classical patterns and outputs a superposition of the corresponding outcomes, provided that certain conditions (like the orthogonality of the patterns and the consistent erasure of quantum variables) hold.
Consider, for example, Deutsch's algorithm~\cite{deutsch} as written in abstract Qunity syntax:
\begin{align*}
	&\texttt{deutsch}(f) \defeq \\
	&\quad\quad\texttt{let } x \texttt{ =}_{\color{gray}{\Bit}}\; {(\had\; \zero)} \texttt{ in} \\
	&\quad\quad\left(\cntrl{(f\; x)}{\color{gray}\Bit}{\zero &\mapsto x \\ \one &\mapsto x \triangleright \gphase{\Bit}{\pi}}{\Bit}\right) \triangleright \had
\end{align*}

\noindent
Given an \emph{arbitrary} quantum oracle $f$ that takes in and outputs a single qubit,
\texttt{deutsch} will test whether $f$ is constant or not using only a single evaluation (rather than two as in classical computing).
Unlike in other quantum programming languages, Qunity's quantum control flow allows the programmer to write subroutines in a style similar to classical programming, without having to explicitly translate irreversible programs into unitary circuit forms. For instance, to input the constant-one oracle into the above example of Deutsch's algorithm, the programmer can just use $\lambda x \xmapsto{{\color{gray}\Bit}} 1$ instead of constructing a reversible map of the form $U_f \ket{x, y} = \ket{x, y \oplus f(x)} = \ket{x, y \oplus 1}$ as is typical in circuit-style programming. However, Qunity's type system places several constraints on the expressions in the \texttt{ctrl} block to ensure its realizability: specifically, the right-hand-side patterns must \emph{erase} the quantum variables in the scrutinee expression in a consistent way so that the compiler can perform automatic uncomputation.

To address these limitations, we introduce two new pattern matching constructs to Qunity, which trade off the subroutine capabilities of \texttt{ctrl} for greater flexibility in designing pure and mixed quantum computations. The new \texttt{match} construct corresponds to classical or mixed pattern matching and has minimal restrictions, allowing the programmer to easily write classical logic (which may then become a quantum subroutine). By contrast, \texttt{pmatch} (or pure match) allows for symmetric pattern matching between orthogonal sets of pure expressions, similarly to Sabry et al.~\cite{SPM}. This allows us to easily express reversible quantum programs without the erasure restrictions of \texttt{ctrl} and simultaneously aids in compiling Qunity code to efficient circuits.

Compiling Qunity is a key contribution of this work. Qunity's high degree of abstraction from the quantum circuit model of computation makes compilation particularly challenging. In addition to control flow and pattern matching, Qunity also generalizes the notion of error handling through \texttt{try/catch} statements (which can be viewed as projective measurements) and provides support for quantum sum types (whose semantics corresponds to direct sums of Hilbert spaces). To represent these abstractions, the compiler must encode Qunity types and contexts into quantum registers and allocate ancillary qubits to express Qunity's non-unitary semantics in terms of unitary gates, a process which can lead to significant inefficiencies. While \citeauthor{Voichick_2023} describe an algorithm for converting Qunity into low-level qubit circuits, this procedure is more of a
proof-of-realizability than a practical compiler: even simple quantum algorithms tend to be compiled into prohibitively large circuits. In this work, we develop a practical Qunity compiler, incorporating optimizations at various stages of the compilation process that significantly reduce the number of qubits and gates used in the final compiled circuits.

We begin by providing technical background on Qunity (\Cref{sec:background}), followed by a series of examples to introduce
the intricacies of the language.
Then, we make the following contributions:
\begin{itemize}
    \item We develop the first implementation of the Qunity language. We add a surface syntax to the core Qunity language, allowing the user to define parameterized types and subroutines that may involve recursion and higher-order operations (\Cref{sec:examples}).
    \item We introduce new pattern-matching constructs into the core Qunity language, to avoid the restrictions imposed by the quantum control construct and improve the language's expressiveness (\Cref{sec:extending}).
    \item We implement the first working Qunity compiler that converts high-level Qunity code into low-level quantum circuits in OpenQASM 3, as well as a Qunity interpreter (\Cref{sec:compiler_design}).
    \item We design optimizations for the compilation procedure to reduce the number of qubits and gates used in the final compiled circuit, acting at several stages of the compilation process (\Cref{sec:low_level_opt} and \Cref{sec:high_level_opt}).
\end{itemize}

\section{Qunity's Goals, Syntax, and Typing}
\label{sec:background}

In this section, we discuss the central ideas and motivations behind the construction of Qunity's type system and semantics, 
and we will introduce our new surface syntax.

\subsection{Unified Programming in Qunity}

Qunity's primary goal is a \emph{unified} treatment of quantum and classical computation. Traditional models such as Knill's quantum random access machine (QRAM) \cite{Knill_1996} and dynamic circuits~\cite{OpenQASM} have a classical computer constructing and running quantum circuits, using the measurement results in classical control flow to determine what circuits to run next. This reflects the paradigm of ``quantum data, classical control''~\cite{Selinger2004a}, creating a clear separation between classical and quantum components of an algorithm. 

On the other hand, Qunity aims to blur the line between ``classical'' and ``quantum'' as much as possible, introducing quantum language constructs that generalize classical ones while emphasizing compositionality. 
While Qunity still allows, in principle, to compile Qunity programs to dynamic circuits, it is 
not tailored for 
quantum algorithms that rely heavily on classical computation (e.g. ones that require
floating point manipulation such as variational hybrid quantum-classical algorithms~\cite{McClean_2016} and quantum machine learning~\cite{Wang_2024}).
Instead, Qunity is designed to support the implementation of complex quantum algorithms that operate on a level of abstraction above the circuit model, require manipulating complex data structures in quantum superposition, or use irreversible programs as subroutines in a reversible quantum computation.

\subsection{Surface Syntax} 
\label{sec:surface}

Qunity's core language does not support higher-order functions and recursion, due to theoretical limitations on ways in which these ideas can be generalized to the quantum setting~\cite{yuan-limits}. Furthermore, Qunity's core syntax does not allow subroutines to be named and called in a program multiple times. For these reasons, we augment Qunity with a \emph{surface syntax} which provides a layer of \emph{metaprogramming} on top of the core Qunity language. This system allows the user to create parameterized types, expressions, programs, and real number expressions that get evaluated at compile time during the \emph{preprocessing} stage. User-defined types can be written as \emph{variants} with named \emph{constructors}, which are evaluated into the core language's left and right injections during the preprocessing stage.
The full grammar of the surface syntax, as used in Qunity source files, can
be found in \AppRef{B}{app:surface_grammar}.

For concreteness, let us revisit the Deutsch example from the introduction and present it in terms of 
the new surface syntax:

\ifshort\else\newpage\fi

\begin{lstlisting}
def &deutsch{@f : Bit -> Bit} : Bit :=
  let x = &plus in
  ctrl @f(x) [
    &0 -> x;
    &1 -> x |> gphase{pi}
  ]
  |> @had
end
\end{lstlisting}

In this syntax, the symbols \lstinline|&|, \lstinline|@|, and \lstinline|#| are used as sigils to write names for Qunity expressions, programs, and numbers, respectively. Qunity types have names starting with a capital letter, and quantum variables (which are part of the core syntax and remain after the preprocessing stage) start with a lowercase letter or underscore. The syntax \lstinline!|>! is a shorthand for function application: \lstinline!x |> @f! is equivalent to \lstinline|@f(x)|. Additionally, \lstinline|let x = y in z| is syntactic sugar for \lstinline|(lambda x -> z)(y)|. Global phase (\lstinline|gphase|) can be defined as syntactic sugar over the primitive \lstinline|rphase| (relative phase), with \lstinline|gphase{r}| equivalent to \lstinline|rphase{_, r, r}|.

The type \lstinline{Bit} is our first example of a user-defined datatype, which is provided by the Qunity standard library:
\begin{lstlisting}
type Bit := &0 | &1 end
\end{lstlisting}
Here, \lstinline{Bit} is defined as a variant type with two constructors that do not take any arguments. When the preprocessor evaluates this into the base Qunity syntax, it converts the variant types into sum types and the constructors into left and right injections. So, \lstinline|&0| becomes $\lef{\texttt{Unit}}{\texttt{Unit}}\unit$, and \lstinline|&1| becomes $\rit{\texttt{Unit}}{\texttt{Unit}}\unit$.
We can also define \lstinline{@had} and \lstinline|&plus| to represent the Hadamard gate and the $\ket{+}$ state:
\begin{lstlisting}
def @had : Bit -> Bit := u3{pi/2, 0, pi} end
def &plus : Bit := @had(&0) end
\end{lstlisting}
where \lstinline{u3} is a Qunity primitive for general single-qubit gates with the type $\texttt{Unit} \oplus \texttt{Unit} \rightsquigarrow \texttt{Unit} \oplus \texttt{Unit}$ in the core Qunity language.

\subsection{Typing and Semantics}

Qunity has two distinct typing judgments for expressions: \emph{pure} expression typing and \emph{mixed} expression typing. These correspond to two distinct but interrelated semantics: a pure semantics that is defined in terms of norm non-increasing operators acting on state vectors and a mixed semantics that is defined in terms of trace non-increasing superoperators acting on density matrices.

We write $\Gamma \partition \Delta \vdash e : T$ to indicate that expression $e$ has pure type $T$ with respect to classical context $\Gamma$ and quantum context $\Delta$. Similarly, we write $\Gamma \partition \Delta \Vdash e : T$ (using $\Vdash$ instead of $\vdash$) to indicate that $e$ has \emph{mixed type} $T$. The classical contexts here are not to be viewed as literally containing classical variables---a variable being in the classical context indicates that it can only be accessed by copying it in the classical basis, and, unlike quantum variables, its relevance is not enforced by the type system. \AppRef{A.2}{app:type_system} contains a more detailed discussion of the role of these contexts.

We write $\vdash f : T \rightsquigarrow T'$ to indicate that the program $f$ is typed as a \emph{coherent map} from $T$ to $T'$, and we write $\vdash f : T \Rrightarrow T'$ to indicate that $f$ is typed as a \emph{quantum channel}. Coherent maps have operator semantics, acting on pure states, while quantum channels act on density matrices.

The semantics of a pure expression corresponds to a linear map $\msem{\sigma : \Gamma \partition \Delta \vdash e : T} \in \cL(\cH(\Delta), \cH(T))$, sending quantum states in the Hilbert space $\cH(\Delta)$ associated with the quantum context $\Delta$ to states in the space $\cH(T)$ associated with the type $T$ (see \AppRef{A.3}{app:semantics} for definitions of these Hilbert spaces). Similarly, $\msem{\vdash f : T \rightsquigarrow T'} \in \cL(\cH(T), \cH(T'))$ sends states from the space $\cH(T)$ to $\cH(T')$. Pure expressions and coherent maps are used to represent \emph{reversible} quantum computation, which does not discard quantum information and has a well-defined \emph{adjoint}. Note that ``reversible'' does not necessarily mean ``invertible'' and the adjoint of an operator representing the semantics of a pure expression or coherent map does not need to be its inverse. Thus, unlike in the quantum circuit model and languages such as SPM \cite{SPM}, these linear maps do not necessarily need to be unitary. Instead, they are restricted to the much broader class of \emph{contractions} (norm non-increasing operators), which includes \emph{projectors} and non-unitary \emph{isometries}. Consider, for example, the expression \lstinline|lambda &0 -> &0|, which has the operator semantics of a projector:
\[
\ket{0}\bra{0} = \bm{1 & 0 \\ 0 & 0}.
\]
Using \lstinline|&0| as a pattern in the \lstinline|lambda| effectively creates an assertion that the input is the $\ket{0}$ state. If \lstinline|&1| is given as input, the result is $0$ (the zero vector), signifying an ``error state'' or an ``exception'' being thrown and if $\lstinline|&plus|$ (whose semantics is $\ket{+} = \frac{1}{\sqrt{2}}(\ket{0} + \ket{1})$) is given as input, the result is $\frac{1}{\sqrt{2}}\ket{0}$, which can be viewed as a quantum superposition of the assertion succeeding and failing.

Despite the freedom gained by not restricting pure semantics to unitaries, Qunity's type system ensures that pure expressions and programs never discard quantum information, placing relevance constraints on variables in quantum contexts. However, sometimes maintaining reversibility is unnecessary and requires a large amount of tedious bookkeeping. For instance, if a programmer wishes to implement an AND gate reversibly, they would need to implement a $3$-bit operation such as $(a, b, c) \mapsto (a, b,(a \wedge b) \oplus c)$, corresponding to the Toffoli gate. Keeping track of this extra data can be inconvenient and unnecessary. Mixed typing and semantics allow the Qunity programmer to create decoherence by irreversibly discarding quantum information when necessary. The semantics of mixed expressions and quantum channels is described by \emph{trace non-increasing superoperators} acting on the space of \emph{density matrices}. In quantum mechanics, density matrices describe the state of an open quantum system that has interacted with its environment. Quantum algorithm designers often work with both the state vector and density matrix formalisms, and the design of the Qunity language allows both to be implemented in a convenient way. Consider, for example, the Qunity program \lstinline{lambda x -> (x, x) |> lambda (x, y) -> x}. Here, \lstinline[breaklines=false]{lambda x -> (x, x)} implements the isometry $\ket{0, 0}\bra{0} + \ket{1, 1}\bra{1}$,\footnote{Note that Qunity allows variables to be reused: This does not violate the no-cloning theorem \cite{nocloning}, since this corresponds to \emph{sharing} a state along the classical basis (an operation that creates entanglement) rather than \emph{cloning} it.} while \lstinline{lambda (x, y) -> x} cannot be typed as a coherent map since it introduces decoherence by discarding the variable $y$. Taken as a whole, the program shares \lstinline|x| along the classical basis and discards the new copy of it, which effectively performs a \emph{measurement}. If the pure state $\ket{+}$ is input into this program, it transforms into the maximally mixed state $\frac{1}{2}(\ket{0}\bra{0} + \ket{1}\bra{1})$, which is a probabilistic mixture (rather than a quantum superposition) of $\ket{0}$ and $\ket{1}$.

Finally, Qunity's \lstinline|try|/\lstinline|catch| construct can be viewed as performing a measurement to determine whether an exception has occurred and performing another computation if it has. This can convert a norm-decreasing operator into a trace-preserving superoperator. For example, the Qunity expression \lstinline{try &plus |> lambda &0 -> &0 catch &plus} corresponds to the density matrix
\[
\frac{1}{2}\ket{0}\bra{0} + \frac{1}{2}\ket{+}\bra{+} = \bm{3/4 & 1/4 \\ 1/4 & 1/4}.
\]

A density matrix can always be viewed as a partial trace of a pure state on a larger system, and a quantum channel can always be represented as the application of an isometry followed by a partial trace using the Stinespring dilation \cite{stinespring-dilation}. 

The design of Qunity's type system and semantics thus allows for a rich interplay between pure and mixed programming, allowing concepts prominent in quantum algorithm \emph{analysis} to be easily extended into the realm of quantum algorithm \emph{implementation}.

\section{Qunity by Example}\label{sec:examples}

In this section, we demonstrate several examples of quantum algorithms implemented in Qunity. These examples showcase the power of Qunity's compositional quantum control flow constructs: \lstinline|ctrl|, as well as the new \lstinline|pmatch| and \lstinline|match| constructs that we introduce to overcome some of the limitations of \lstinline|ctrl|. \Cref{sec:extending} will discuss these in more technical detail.

\subsection{Order Finding}\label{sec:examples_order_finding}

The following example is a simplified demonstration of the order finding algorithm, which is the quantum part of Shor's algorithm for integer factorization \cite{Shor_1997}. Given relatively prime integers $N$ and $a < N$, the goal is to find an integer $r$ such that $a^r \equiv 1 \pmod{N}$. For the purposes of this demonstration, we will assume that $N$ is a power of two -- the resulting program is not useful for Shor's algorithm (as it assumes $N$ is \emph{odd} and orders modulo $2^n$ are easy to calculate classically~\cite{order_mod_pow2}), but it is simpler to implement and illustrative of Qunity's features.

In our surface language, we can define a datatype of arrays of a given length, and then represent numbers \lstinline|Num{#n}| in little-endian form using this datatype:

\begin{lstlisting}
type Array{#n, 'a} := if #n <= 0 then Unit
                      else 'a * Array{#n - 1, 'a} endif end
type Num{#n} := Array{#n, Bit} end
\end{lstlisting}

Using this representation, we can write the following code to add a fixed number \lstinline|#a| to a quantum \lstinline|#n|-bit number (overflowing so that the addition is modulo $2^{\lstinline|#n|}$) in a reversible way:

\begin{lstlisting}
def @add_const{#n, #a} : Num{#n} -> Num{#n} :=
  if #n <= 0 then
    @id{Unit}
  else
    pmatch [
      (&0, x) -> (if #a % 2 = 0 then &0 else &1 endif,
                    x |> @add_const{#n - 1, (#a-#a%2)/2});
      (&1, x) -> (if #a % 2 = 0 then &1 else &0 endif,
                    x |> @add_const{#n - 1, (#a-#a%2)/2 + #a%2})
    ]
  endif
end
\end{lstlisting}
This showcases the use of the newly added \lstinline|pmatch| construct. This form of pattern matching, which has \emph{pure program} semantics, allows us to reversibly transform between two orthogonal sets of Qunity expressions. The type system imposes orthogonality requirements (\AppRef{F}{app:ortho_judgment}) on \emph{both} sides of the pattern-matching block in \lstinline|pmatch|, while \lstinline|ctrl| only imposes them on the left-hand side patterns. However, since \lstinline|pmatch| does not need to perform automatic uncomputation, it avoids the \emph{erasure judgment} (\AppRef{H}{app:erasure_judgment}) required by \lstinline|ctrl|, so the use of variables on the right-hand side of the pattern matching block is much less restricted in this respect. It would be difficult to define this operation using only the \lstinline|ctrl| construct.

For this example, \lstinline|@add_const| is defined recursively, using versions of itself with \lstinline[breaklines=false]|#n - 1| and different values of \lstinline|#a| depending on whether there is a carry bit. It is clear that \lstinline|@add_const{0, #a}| has unitary semantics since it is the identity map. Now, for any integer \lstinline|#n|, assuming that \lstinline|@add_const{#n - 1, #a}| is unitary, the typechecker will correctly type the \lstinline|pmatch| statement and conclude that \lstinline|@add_const{#n, #a}| is also unitary because the two sets of orthogonal expressions \emph{span} their respective spaces. This ensures that the orthogonality requirements are satisfied.

We can now use constant addition to define modular multiplication by an odd constant \lstinline|#a| modulo $2^{\lstinline|#n|}$, using the approach in~\citet{gidney_modinv}:

\begin{lstlisting}
def @mod_mult{#n, #a} : Num{#n} -> Num{#n} :=
  if #n = 1 then
    @id{Num{#n}}
  else
    lambda (x0, x1) ->
    let (x0, x1) = (x0, x1 |> @mod_mult{#n - 1, #a}) in
    ctrl x0 [
      &0 -> (x0, x1);
      &1 -> (x0, @add_const{#n - 1, (#a - 1) / 2}(x1))
    ]
  endif
end
\end{lstlisting}

This procedure works because if $a$ is odd, then
\[
(1 + 2x)a \equiv 1 + 2 x a + a - 1 \equiv 1 + 2 \parens{xa + \frac{a - 1}{2}}
\]
where all addition is modulo $2^n$. Now, we can define a modular exponentiation program that implements the operation $\ket{x, y} \mapsto \ket{x, y \cdot a^x}$, where we use \lstinline|ctrl| to coherently condition on the bits of $x$ and apply a modular multiplication by the current power of \lstinline|#a| (note that now we cannot use \lstinline|pmatch| since the RHS is not necessarily orthogonal):

\begin{lstlisting}
def @mod_exp{#m, #n, #a} : Num{#m}*Num{#n} -> Num{#m}*Num{#n} :=
  if #m = 0 then @id{Num{#m} * Num{#n}}
  else
    lambda ((x0, x1), y) -> let ((x0, x1), y) = ctrl x0 [
      &0 -> ((x0, x1), y);
      &1 -> ((x0, x1), @mod_mult{#n, #a}(y)) ]
    in let (x0, (x1, y)) =
      (x0, @mod_exp{#m - 1, #n, #a * #a}(x1, y))
    in ((x0, x1), y)
  endif
end
\end{lstlisting}

\ifshort\else\newpage\fi

With all the above machinery in place, we can finally define the main order finding procedure:

\begin{lstlisting}
def &order_finding{#n, #a} : Num{#n} :=
  (&repeated{#n, Bit, &plus}, &num_to_state{#n, 1})
  |> @mod_exp{#n, #n, #a}
  |> @fst{Num{#n}, Num{#n}}
  |> @adjoint{Num{#n}, Num{#n}, @qft{#n}}
  |> @reverse{#n, Bit}
end
\end{lstlisting}

The order finding procedure starts with a uniform superposition of the values of $x$ and $y = 1$, applies the modular exponentiation, then discards the second register (\lstinline|@fst| is defined as just \lstinline|lambda (x, y) -> x|) and applies the inverse quantum Fourier transform (see \AppRef{D.2}{app:qft_impl} for the implementation), reversing the result to display it in big-endian. Measuring the output will result in a random number of the form $N j/r$ for some integer $j$. For instance, running \lstinline[breaklines=false]|&order_finding{5, 13}| would output random multiples of $4$, since the order of $13$ modulo $2^5$ is $8$ and $2^5 / 8 = 4$. The procedure relies on the reversibility of the operations and the pure typing judgment to maintain quantum coherence before applying the inverse quantum Fourier transform.

\subsection{Grover's Algorithm}\label{sec:examples_grover}

An interesting example of the use of \lstinline{ctrl} is Grover's search algorithm~\cite{grover1996}: Given any quantum oracle that takes in a value of some type and outputs a bit, we can start from a uniform superposition of the type and amplify the amplitude of states for which the oracle outputs $1$. The generalized code for Grover's algorithm in Qunity is as follows:

\begin{lstlisting}
def @grover_iter{'a, &equal_superpos : 'a, 
                 @f : 'a -> Bit} : 'a -> 'a :=
  lambda x -> ctrl @f(x) [
    &0 -> x;
    &1 -> x |> gphase{pi}
  ] |> @reflect{'a, &equal_superpos}
end

def &grover{'a, &equal_superpos : 'a, 
            @f : 'a -> Bit, #n_iter} : 'a :=
  if #n_iter = 0 then
    &equal_superpos
  else
    &grover{'a, &equal_superpos, @f, #n_iter - 1} |>
    @grover_iter{'a, &equal_superpos, @f}
  endif
end
\end{lstlisting}

Here \lstinline|@grover_iter| takes in a type parameter \lstinline|'a|, a user-provided expression \lstinline|&equal_superpos| expected to create an equal superposition of all values of type \lstinline|'a| (as there is no way to do this generically), and an oracle \lstinline|@f : 'a -> Bit| that is the function input to Grover's algorithm. The result is a program that transforms inputs of type \lstinline|'a|, bringing them closer to the input that satisfies the oracle. Then, \lstinline|&grover| is a parameterized expression rather than a program: it becomes a Qunity expression after the parameters are substituted in the preprocessing stage. It repeatedly applies \lstinline|@grover_iter| to a state that begins as an equal superposition, for a given number of iterations.
Observe the use of the \lstinline|ctrl| construct in this example: this code works with any function \lstinline|@f|, including an irreversible one. This allows us to define quantum oracles for Grover's algorithm in a completely classical way, without ever needing to worry about the reversibility of the computations. 

The construction of such subroutines is greatly facilitated by the newly added \lstinline|match| construct - a form of mixed or irreversible pattern matching that has minimal restrictions, functioning like the matching constructs in classical programming languages. Consider, for instance, the following example of a function testing whether a list of bits has an odd number of ones:

\begin{lstlisting}
def @is_odd_sum{#n} : List{#n, Bit} -> Bit :=
  if #n = 0 then
    lambda l -> &0
  else
    lambda l -> match l [
      &ListEmpty{#n, Bit} -> &0;
      @ListCons{#n, Bit}(&0, l') -> @is_odd_sum{#n - 1}(l');
      @ListCons{#n, Bit}(&1, l') -> @not(@is_odd_sum{#n - 1}(l'))
    ]
  endif
end
\end{lstlisting}

Here, the \lstinline{List} datatype is a list of variable but bounded length, defined as:
\begin{lstlisting}
type List{#n, 'a} :=
  | &ListEmpty
  | @ListCons of (if #n <= 0 then Void
                    else 'a * List{#n - 1, 'a} endif)
end
\end{lstlisting}

The definition of \lstinline|@is_odd_sum| is typed as a mixed program, with the \lstinline|match| block typed as a mixed expression. It is defined recursively, using \lstinline|@is_odd_sum{#n - 1}| in the RHS of the pattern-matching block. This means that it would be difficult to write this in terms of just the \lstinline|ctrl| construct, which types its RHS expressions as pure. Since \lstinline|match| has mixed semantics, it is not directly responsible for performing any uncomputation and is able to discard quantum data: hence, it can avoid the erasure requirements present in \lstinline|ctrl|.

Now, we can use this as an oracle in Grover's algorithm simply by writing

\begin{lstlisting}
def #n := 2 end
&grover{List{#n, Bit}, &equal_superpos_list{#n},
        @is_odd_sum{#n}, 3}
\end{lstlisting}
where \lstinline|&equal_superpos_list| is an expression generating an equal superposition of all possible lists of bits with lengths bounded by \lstinline{#n} (implementation in \AppRef{D.7}{app:grover_with_lists_impl}). This code uses Grover's algorithm to amplify the probability of measuring the lists of bits (of length at most 2) the sum of whose elements is odd.

\section{Generalizing Qunity's Control Flow}
\label{sec:extending}

Qunity's \lstinline|ctrl| construct is very powerful, as it allows an arbitrary quantum computation, potentially involving decoherence, to be used as a subroutine in a pure, reversible computation. However, this feature is accompanied by some notable restrictions that can make this construct more difficult to use than necessary.

The first significant restriction is that the left-hand-side (LHS) expressions of the \lstinline|ctrl|-block must be classical---that is, they cannot include any invocations of \lstinline|u3| or \lstinline|rphase|, and so their semantics correspond to classical basis states. The second restriction arises from the requirement that the right-hand-side (RHS) expressions must satisfy the \emph{erasure judgment}, which stipulates that all variables in the quantum context $\Delta$ of the scrutinee expression must be present ``in the same way'' in each of the RHS expressions. If, for instance, $\Delta$ involves the variable $x$, this essentially forces the programmer to make all the RHS expressions have the form $(x, \dots)$ if they want to output something other than just $x$ itself. See \AppRef{H}{app:erasure_judgment} for the definition of the erasure inference rules.

The erasure requirement is essential to ensure that it is possible to implement \lstinline|ctrl|: without it, the semantics of \lstinline|ctrl| will not be reversible. At the circuit level, the context $\Delta$ needs to be shared to produce the RHS expressions, but applying the adjoint of the purified scrutinee expression inevitably creates another copy of $\cH(\Delta)$ alongside the output $\cH(T')$, and the erasure judgment is needed to coherently delete this. The classical requirement in the orthogonality judgment does not seem as obvious; indeed, applying an isometric transformation to an orthogonal set of states maintains their orthogonality. However, dropping the classical requirement from \lstinline|ctrl| typing would allow for non-physical norm-\emph{increasing} semantics. For instance, using the definition of \lstinline|ctrl|'s semantics from \AppRef{A.3}{app:semantics} on the expression
\begin{lstlisting}
&0 |> lambda x -> ctrl x [&plus -> x; &minus -> x]
\end{lstlisting}
would result in the physically impossible state $\sqrt{2}\ket{0}$ (with norm $\sqrt{2} > 1$). The correctness of the circuit construction in \AppRef{L.3}{app:typing_judgment_compilation} relies on this classical-basis assumption.

It would be useful to have a construct that defines a reversible transformation between two orthogonal sets of expressions, which do not necessarily need to be in the classical basis, while avoiding the erasure requirement. Also, if the programmer is implementing classical logic (which may be part of a larger quantum algorithm), it may be useful to have a construct that defines \emph{irreversible} control flow in a manner most similar to a \lstinline|match| statement in classical programming languages, again without the restrictions of the erasure judgment. These ideas give rise to two new Qunity primitives: \lstinline|pmatch| and \lstinline|match|.

\begin{table}[ht]
\caption{Comparison of Qunity's three pattern-matching constructs (extended version in \AppRef{C}{app:control_flow_comparison}).}
\begin{minipage}{\columnwidth}
\begin{center}
\begin{tabular}{llll} \toprule
Feature & \lstinline|ctrl| & \lstinline|match| & \lstinline|pmatch| \\\midrule
Purely typed & \cmark & \xmark & \cmark \\
Mixed scrutinee allowed & \cmark & \cmark & \xmark \\
Mixed expressions allowed on RHS & \xmark & \cmark & \xmark \\
Non-orthogonal expressions allowed on RHS & \cmark & \cmark & \xmark \\
Non-classical expressions allowed on LHS & \xmark & \xmark & \cmark \\
Avoids erasure requirements & \xmark & \cmark & \cmark \\
Can be used in the RHS of a \lstinline|ctrl| & \cmark & \xmark & \cmark \\
Can be used in the RHS of a \lstinline|match| & \cmark & \cmark & \cmark \\
Can be used in the LHS or RHS of a \lstinline|pmatch| & \cmark* & \xmark & \cmark* \\
\bottomrule
\end{tabular}
\footnotetext{*if the isometry judgment holds}
\end{center}
\end{minipage}
\end{table}

\subsection{The \texttt{pmatch} Construct}\label{sec:pmatch}

The \lstinline|pmatch| construct (the ``p'' stands for ``pure'') has a design that is very similar to that of the symmetric pattern-matching language (SPM) described by \citet{SPM}. We can write a typing judgment for \lstinline|pmatch| as follows:

\[
\inference{\ortho{T}{e_1, \dots, e_n} & \qquad \varnothing \partition \Delta_j \vdash e_j : T \;\forall j \\ \ortho{T'}{e_1', \dots, e_n'} & \qquad \varnothing \partition \Delta_j \vdash e_j' : T' \;\forall j}{\vdash \pmatch{T}{e_1 &\mapsto e_1' \\ &\cdots \\ e_n &\mapsto e_n'}{T'} : T \rightsquigarrow T'}[\textsc{T-Pmatch}]
\]

\noindent
Each pattern and branch must be appropriately typed using the same quantum context $\Delta_j$, and both the LHS and RHS must satisfy the orthogonality judgment (\AppRef{F}{app:ortho_judgment}). If the classical requirement as in \lstinline|ctrl| was still in place, then \lstinline|pmatch|\footnote{Called ``match'' by \citeauthor{Voichick_2023}, but we change the terminology as \textcolor{codepurple}{\texttt{match}} will be a different construct.} could be described as syntactic sugar over the \lstinline|ctrl| construct, using a technique known as \emph{specialized erasure} \cite{Voichick_2023}.
Instead, we make \lstinline|pmatch| more general by allowing expressions in an arbitrary basis on both sides. This is done at the expense of losing some of the capabilities of \lstinline|ctrl|: since the semantics of \lstinline|pmatch| is that of a pure program rather than a pure expression, it does not include a scrutinee directly and thus does not have the power to uncompute arbitrary mixed expressions via purification. It also has a requirement dictating that the same contexts must be used symmetrically in the LHS and RHS, making it possible to simply flip the two sides to invert the program.

We can describe the semantics of \lstinline|pmatch| as:
\[
\msem{\vdash \pmatch{T}{e_1 &\mapsto e_1' \\ &\cdots \\ e_n &\mapsto e_n'}{T'} : T \rightsquigarrow T'}\ket{v} =
\sum_{j=1}^n \msem{\varnothing : \varnothing \partition \Delta_j \vdash e_j' : T'}
\msem{\varnothing : \varnothing \partition \Delta_j \vdash e_j : T}\adj \ket{v}
\]

Each term of the above sum is formed by applying the adjoint of the semantics of an LHS expression, $\msem{\varnothing : \varnothing \partition \Delta_j \vdash e_j : T}\adj : \cH(T) \rarr \cH(\Delta_j)$, followed by an application of the corresponding RHS expression $\msem{\varnothing : \varnothing \partition \Delta_j \vdash e_j' : T'} : \cH(\Delta_j) \rarr \cH(T')$. The images of the operators $\msem{\varnothing : \varnothing \partition \Delta_j \vdash e_j : T}$ correspond to a set of orthogonal subspaces of $\cH(T)$. In the special case where all the LHS expressions are classical and $\ket{v}$ is a classical basis state, at most one of the terms in the sum will be nonzero. The variables in the pattern will be extracted from the input, and the result will be the RHS expression applied to their values. In general, all the branches of the \lstinline|pmatch| block can be taken in superposition. If the LHS patterns are \emph{spanning} (\AppRef{F}{app:ortho_judgment}), then \lstinline|pmatch| will act as an \emph{isometry}, as every input can be described as a linear combination of the patterns. If the RHS expressions are also spanning, then the semantics of \lstinline|pmatch| will correspond to a \emph{unitary} operator.

For an example of the benefit provided by \lstinline|pmatch|, consider the following Qunity code which prepares an equal superposition of \lstinline|Maybe{Bit}|, which is isomorphic to a qutrit:

\begin{lstlisting}
&0
|> u3{2 * arccos(sqrt(1 / 3)), 0, 0}
|> pmatch [
  &0 -> &Nothing{Bit};
  &1 -> @Just{Bit}(&plus)
]
\end{lstlisting}

Here, we are easily able to transform between two orthogonal bases of expressions, one in the type \lstinline|Bit| and the other in the type \lstinline|Maybe{Bit}|. The semantics of this can be written as:
\begin{align*}
& \msem{\texttt{pmatch}\{\dots\}}
\msem{\texttt{u3}(\arccos(\sqrt{1/3}, 0, 0))}\ket{\texttt{\$0}} = \\
={}& \parens{
\msem{\texttt{Nothing}\{\texttt{Bit}\}}
\msem{\texttt{\$0}}^\dagger +
\msem{\texttt{@Just}\{\texttt{Bit}\}}
\msem{\texttt{\$plus}}
\msem{\texttt{\$1}}^\dagger
}
\parens{
\frac{1}{\sqrt{3}} \ket{\texttt{\$0}}
+ \frac{\sqrt{2}}{\sqrt{3}} \ket{\texttt{\$1}}
} = \\
={}& \parens{
\ket{\texttt{Nothing}\{\texttt{Bit}\}}
\bra{\texttt{\$0}} +
\frac{1}{\sqrt{2}}
\ket{\texttt{@Just}\{\texttt{Bit}\}(\texttt{\$0})}
\bra{\texttt{\$1}} +
\frac{1}{\sqrt{2}}
\ket{\texttt{@Just}\{\texttt{Bit}\}(\texttt{\$1})}
\bra{\texttt{\$1}}
} \\
& \hspace{27em} \parens{
\frac{1}{\sqrt{3}} \ket{\texttt{\$0}}
+ \frac{\sqrt{2}}{\sqrt{3}} \ket{\texttt{\$1}}
} = \\
={}& \frac{1}{\sqrt{3}} \parens{
\ket{\texttt{Nothing}\{\texttt{Bit}\}} +
\ket{\texttt{@Just}\{\texttt{Bit}\}(\texttt{\$0})} +
\ket{\texttt{@Just}\{\texttt{Bit}\}(\texttt{\$1})}
}.
\end{align*}

If we were to implement this only using \lstinline|ctrl|, we would need to write something like the following:
\begin{lstlisting}
&0
|> u3 {2 * arccos(sqrt(1 / 3)), 0, 0}
|> lambda x -> ctrl x [
  &0 -> (x, &Nothing {Bit});
  &1 -> (x, @Just{Bit}(&plus))
]
|> lambda (
  ctrl x' [
    &Nothing{Bit} -> (&0, x');
    @Just{Bit}(_) -> (&1, x')
  ]
) -> x'
\end{lstlisting}
\noindent
This is an instance of the specialized erasure pattern: in order to respect the erasure judgment, we must explicitly uncompute the variable \texttt{x} through an inverted \lstinline|ctrl| expression.

Moreover, isometries can be combined with \lstinline{pmatch} in more complicated ways, such as the following piece of code 
which allows us to view the rest of the tuple in the standard or Fourier basis depending on the value of the first element:
\begin{lstlisting}
&repeated{4, Bit, &minus} |> pmatch [
  (&0, (x, y)) -> (@qft{2}(&plus, (x, ())), @qft{2}(y));
  (&1, @qft{3}(x, y)) ->
      (@qft{2}(&minus, (x, ())), @add_const{2, 1}(y));
]
\end{lstlisting}
Note that we can use the Quantum Fourier Transform (\lstinline|@qft|, given in \AppRef{D.2}{app:qft_impl}) and \lstinline|@add_const| (from \Cref{sec:examples_order_finding}) on both sides of the pattern-matching block since the type system can recognize that both of them have isometric (and in fact, unitary) semantics.

\subsection{The \texttt{match} Construct}\label{sec:match}

In a sense, \lstinline|match| is a step in the opposite direction from \lstinline|pmatch|: while \lstinline|pmatch| extends the ``pure'' capabilities of \lstinline|ctrl|, \lstinline|match| extends its ``mixed'' capabilities. The goal of the \lstinline|match| construct is to have as few restrictions as possible. We can create a typing judgment for \lstinline|match|, ensuring that the types of the scrutinee, patterns, and branches align in the current context, while also requiring orthogonality:
\[
\inference{\Gamma \partition \Delta, \Delta_0 \Vdash e : T \qquad \ortho{T}{e_1, \dots, e_n} \qquad \varnothing \partition \Gamma_j \vdash e_j : T \; \forall j \\
\classical(e_j) \; \forall j \qquad \Gamma, \Gamma_j \partition \Delta, \Delta_1 \Vdash e_j' : T' \; \forall j}
{\Gamma \partition \Delta, \Delta_0, \Delta_1 \Vdash \match{e}{T}{e_1 &\mapsto e_1' \\ &\cdots \\ e_n &\mapsto e_n'}{T'} : T'}[\textsc{T-Match}]
\]
and we can define its semantics as:
\begin{align*}
\msem{\sigma : \Gamma \partition \Delta, \Delta_0, \Delta_1 \Vdash \match{e}{T}{e_1 &\mapsto e_1' \\ &\cdots \\ e_n &\mapsto e_n'}{T'} : T'} \parens{\ket{\tau, \tau_0, \tau_1}\bra{\tau', \tau_0', \tau_1'}} \\
= \sum_{v \in \VV(T)} \bra{v} \parens{\msem{\sigma : \Gamma \partition \Delta, \Delta_0 \Vdash e : T}\parens{\ket{\tau, \tau_0}\bra{\tau', \tau_0'}}} \ket{v} \cdot \\
\cdot \sum_{j=1}^n \sum_{\sigma_j \in \VV(\Gamma_j)}
\bra{\sigma_j} \msem{\varnothing : \varnothing \partition \Gamma_j \vdash e_j : T}\adj \ket{v} \cdot \\
\cdot \msem{\sigma, \sigma_j : \Gamma, \Gamma_j \partition \Delta,\Delta_1 \Vdash e_j' : T'}\parens{\ket{\tau, \tau_1}\bra{\tau', \tau_1'}}
\end{align*}

This is very similar to the semantics of \lstinline|ctrl|, described in \AppRef{A.3}{app:semantics}. Here,
\[
\bra{v} \parens{\msem{\sigma : \Gamma \partition \Delta, \Delta_0 \Vdash e : T}\parens{\ket{\tau, \tau_0}\bra{\tau', \tau_0'}}} \ket{v}
\]
gives the \emph{probability} of the expression $e$ being measured as $\ket{v}$ in the classical basis, with the given classical and quantum context variables. The expression $e$ is thus treated as a ``classical black box'': thus, there is no way to distinguish the pure state $\frac{1}{\sqrt{2}}(\ket{0} + \ket{1})$ from the mixed state $\frac{1}{2}(\ket{0}\bra{0} + \ket{1}\bra{1})$, since their probability distributions when measured in the classical basis are identical. This behavior is exactly the same as in \lstinline|ctrl|. However, unlike \lstinline|ctrl|, which has pure semantics, \lstinline|match| has mixed semantics: it uses the superoperator $\msem{\sigma, \sigma_j : \Gamma, \Gamma_j \partition \Delta,\Delta_1 \Vdash e_j' : T'}$ formed from the RHS expression $e_j'$.

The \lstinline|match| construct is designed to be used in circumstances where we do not care if our computation is reversible or not. For instance, it is convenient to implement a logical-AND operation using \lstinline|match| as follows:
\begin{lstlisting}
def @and : Bit * Bit -> Bit :=
  lambda x -> match x [
    (&1, &1) -> &1;
    else -> &0;
  ]
end
\end{lstlisting}

Note that the \lstinline{else} keyword here is a syntactic construct that gets converted to the remaining expressions spanning \lstinline|Bit * Bit| (that is, \lstinline|(&1, &0)| and \lstinline|(&0, _)|) during the preprocessing stage. The \lstinline|match| construct no longer has the erasure requirement from \lstinline|ctrl|, simply because it does not need to be reversible and can always discard any extra data it has, without the need to uncompute it. Like for \lstinline|pmatch|, there exists a version of \lstinline|match| that can be built as syntactic sugar over \lstinline|ctrl|. This version can be constructed by simply combining all the variables in $\Delta$ into a tuple and pairing it with each of the RHS expressions, and then feeding the result into \lstinline|@snd|. However, the additional feature of \lstinline|match| when it is implemented as a primitive is the fact that the RHS expressions, unlike in \lstinline|ctrl|, do not need to be pure. See the example in \Cref{sec:examples_grover} for an instance where this becomes useful for quantum algorithm implementation. The \lstinline|match| construct is designed to be maximally similar to pattern matching in classical languages, avoiding the restrictions placed on \lstinline|ctrl| other than the orthogonality requirement for the LHS. Using the \lstinline|match| construct, Qunity programmers may write long and complex classical algorithms without worrying about reversibility, and then use them as subroutines in larger quantum algorithms.

\section{Compiler Design}\label{sec:compiler_design}

\begin{wrapfigure}{r}{0.5\textwidth}
\vspace{-1.5em}
\centering
\resizebox{\linewidth}{!}{%
\begin{circuitikz}
\tikzstyle{every node}=[font=\LARGE]
\draw  (6.25,14.5) rectangle  node {\LARGE Qunity Source Code} (11.25,13.25);
\draw  (6.25,12) rectangle  node {\LARGE Qunity Surface Syntax} (11.25,10.75);
\draw [->, >=Stealth] (8.75,13.25) -- (8.75,12)node[pos=0.5, fill=white]{Lexing, Parsing};
\draw  (6.25,9.5) rectangle  node {\LARGE Qunity Core Language} (11.25,8.25);
\draw [->, >=Stealth] (8.75,10.75) -- (8.75,9.5)node[pos=0.5, fill=white]{Preprocessing, Definition Evaluation};
\draw  (-1.25,4.5) rectangle  node {\LARGE Superoperators as Matrices} (6.25,3.25);
\draw [->, dashed, >=Stealth] (6.25,5.75) -- (2.5,4.5)node[pos=0.5, fill=white]{Semantics Evaluation};
\draw  (11.25,4.5) rectangle  node {\LARGE Intermediate Representation} (18.75,3.25);
\draw [->, >=Stealth] (11.25,5.75) -- (15,4.5)node[pos=0.5, fill=white]{Compilation Step 1};
\draw  (11.25,2) rectangle  node {\LARGE Low-Level Circuit Specification} (18.75,0.75);
\draw [->, >=Stealth] (15,3.25) -- (15,2)node[pos=0.5, fill=white]{Compilation Step 2};
\draw  (12.5,-0.5) rectangle  node {\LARGE Low-Level Circuit} (17.5,-1.75);
\draw [->, >=Stealth] (15,0.75) -- (15,-0.5)node[pos=0.5, fill=white]{Circuit Instantiation, Postprocessing};
\draw  (12.5,-3) rectangle  node {\LARGE OpenQASM} (17.5,-4.25);
\draw [->, >=Stealth] (15,-1.75) -- (15,-3)node[pos=0.5, fill=white]{Conversion};
\draw  (-1.25,0.75) rectangle  node {\LARGE Final Density Matrices} (6.25,-0.5);
\draw [->, dashed, >=Stealth] (2.5,3.25) -- (2.5,0.75)node[pos=0.5, fill=white]{Matrix Computation};
\draw  (-1.25,-3) rectangle  node {\LARGE Output Distribution (or Sample)} (6.25,-4.25);
\draw [->, dashed, >=Stealth] (2.5,-0.5) -- (2.5,-3)node[pos=0.5, fill=white]{Born Rule};
\draw [->, dashed, >=Stealth] (12.5,-1) -- (6.25,0.25)node[pos=0.5, fill=white]{Circuit Simulation};
\draw [->, >=Stealth] (12.5,-3.75) .. controls (11.5,-5.5) and (7.25,-5.5) .. (6.25,-3.75) node[pos=0.5, fill=white]{Execution on Quantum Hardware};
\draw [->, dashed, >=Stealth] (12.5,-3.75) .. controls (11.5,-1.75) and (7.25,-1.75) .. (6.25,-3.75) node[pos=0.5, fill=white]{QASM Simulation};
\draw  (5,7) rectangle  node {\LARGE Typing Judgment Proof Structures} (12.5,5.75);
\draw [->, >=Stealth] (8.75,8.25) -- (8.75,7)node[pos=0.5, fill=white]{Typechecking};
\end{circuitikz}
}%

\caption{Diagram showing the procedures used in the Qunity interpreter (left path), and the Qunity compiler (right path). Solid and dashed arrows indicate polynomial and exponential time processes respectively.}
\label{fig:overview}
\end{wrapfigure}
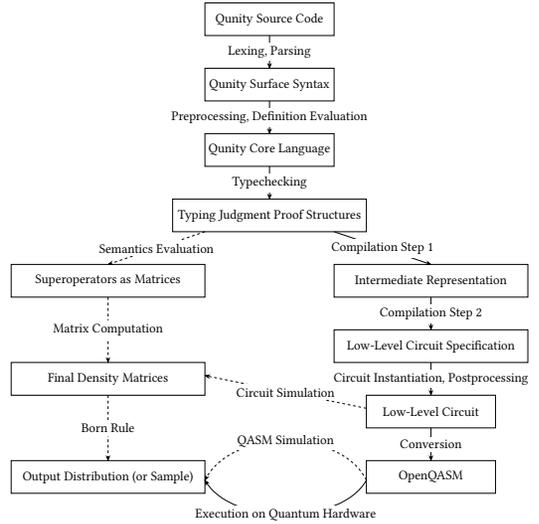

Figure \ref{fig:overview} shows an outline of the main components of the Qunity interpreter and compiler, which are implemented in OCaml. 

The preprocessor first evaluates the surface syntax definitions and produces a single Qunity expression, with the same abstract syntax as in \citeauthor{Voichick_2023} 
The typechecker then uses Qunity's extended typing judgments to check if the expression is well-typed and outputs a data structure representing a proof of typing, effectively elaborating the Qunity expression with helpful information that can later be used by the interpreter and the compiler, such as the contexts associated with its subexpressions. 
We can then either interpret or compile this elaborated expression:
On the left side of \Cref{fig:overview}, we depict a simulator explicitly calculating the matrices corresponding to its semantics. This is an exponentially slow process, however, as it is equivalent to simulating a quantum computer on a classical computer. 
The main focus of this work is the compilation process shown on the right side, transforming a high-level elaborated Qunity expression into OpenQASM 3. This process can be performed efficiently on a classical computer, and the final output is a quantum circuit in a form that can be run on quantum hardware.

\subsection{Formalism}

Before making improvements to Qunity's original proof-of-concept compilation scheme, we found it necessary to revise some of the mathematical formalism that governs how the low-level circuits should be constructed in order to establish stronger correctness invariants throughout the process. In particular, many Qunity values are stored in a quantum register, represented as bitstrings using a type-based \emph{encoding}. However, the space of all possible encodings does not necessarily span the space of all possible register states---for instance, a qutrit (a $3$-level quantum system) can be stored in two bits, but only $3$ out of $4$ possible bit strings will be valid encodings. The compilation procedure in \citeauthor{Voichick_2023} includes some low-level circuits (namely, the direct sum) that can potentially output invalid encodings due to the way flag and output registers in the circuit's components are combined. The compilation lemmas do not address this possibility, which makes some parts of the compilation procedure incorrect. Allowing redundancy in the encodings would be problematic because the extra qubits would need to be treated as flag or garbage qubits when inverting the direct sum injections: a ``non-standard'' encoding would either cause an error to be raised or introduce unwanted decoherence.

We considered two methods of resolving this problem. The first possibility was to create an error-checking circuit component, which would ensure that any state in the invalid encoding subspace is sent to the error subspace (with nonzero values for the flag qubits). However, the issue with this approach is that such an error-checking component would need to be inserted after every circuit component in the compilation procedure, making it much more complex and inefficient. The alternative approach was to ensure that no states enter the invalid encoding subspace under any conditions. We can create a mathematical definition of this as follows:

\begin{definition}\label{def:valid_enc_space}
For a type $T$, define the space of valid encodings as
\[
\WW(T) = \Span\braces{\ket{\enc(v)} \mid v \in \VV(T)},
\]
where $\enc$ is the encoding function that maps values to bitstrings:
\begin{align*}
\enc(\unit) &= \texttt{""} \\
\enc(\lef{T_0}{T_1}\: v) &= \texttt{"0"} \doubleplus \enc(v) \doubleplus \texttt{"0"}^{\max\{\size(T_0), \size(T_1)\} - \size(T_0)} \\
\enc(\rit{T_0}{T_1}\: v) &= \texttt{"1"} \doubleplus \enc(v) \doubleplus \texttt{"0"}^{\max\{\size(T_0), \size(T_1)\} - \size(T_1)} \\
\enc(\pair{v_0}{v_1}) &= \enc(v_0) \doubleplus \enc(v_1),
\end{align*}
and the size of a type is defined as
\begin{align*}
\size(\Void) &= 0 \\
\size(\Unit) &= 0 \\
\size(T_0 \oplus T_1) &= 1 + \max\{\size(T_0), \size(T_1)\} \\
\size(T_0 \otimes T_1) &= \size(T_0) + \size(T_1).
\end{align*}
\end{definition}

\begin{definition}\label{def:possible_to_implement_op}
We say that it is possible to implement a norm non-increasing operator \\ $E : \cH(T) \rarr \cH(T')$ if there is a low-level circuit implementing a unitary operator \\ $U : \CC^{2^{\size(T) + p}} \rarr \CC^{2^{\size(T') + f}}$ for some $p, f \in \NN$ such that the following conditions hold:

\begin{itemize}
\item $\bra{\enc(v'), 0^{\otimes f}}U\ket{\enc(v), 0^{\otimes p}} = \bra{v'}E\ket{v}$ for all $v \in \VV(T)$, $v' \in \VV(T')$,
\item $U\ket{\enc(v), 0^{\otimes p}} \in \WW(T') \otimes \CC^{2^{f}}$ for all $v \in \VV(T)$, and
\item $U\adj\ket{\enc(v'), 0^{\otimes f}} \in \WW(T) \otimes \CC^{2^{p}}$ for all $v' \in \VV(T')$.
\end{itemize}
\end{definition}

\begin{definition}\label{def:possible_to_implement_superop}
We say that it is possible to implement a trace non-increasing operator \\ $\cE : \cL(\cH(T)) \rarr \cL(\cH(T'))$ if there is a low-level circuit implementing a unitary operator $U : \CC^{2^{\size(T) + p}} \rarr \CC^{2^{\size(T') + f + g}}$ for some $p, f, g \in \NN$ such that for all $v_1, v_2 \in \VV(T)$, $v_1', v_2' \in \VV(T')$,
\[
\bra{v_1'} \cE\parens{\op{v_1}{v_2}} \ket{v_2'} =
\sum_{b \in \{0,1\}^g} \bra{\enc(v_1'), 0^{\otimes f}, b} U \op{\enc(v_1), 0^{\otimes p}}{\enc(v_2), 0^{\otimes p}} U\adj \ket{\enc(v_2'), 0^{\otimes f}, b},
\]
and, for all $v \in \VV(T)$,
\[
U \ket{\enc(v), 0^{\otimes p}} \in \WW(T') \otimes \CC^{2^{f + g}}.
\]
\end{definition}

These definitions are formed by augmenting Definitions 6.2 and 6.3 from \citeauthor{Voichick_2023} with conditions that guarantee encoding validity is preserved. Now, any optimizations to the low-level circuit components need to respect the encoding-validity-preserving property.

\subsection{Gates, Circuits, and Circuit Specifications}

The final output of the compiler takes the form of a low-level qubit circuit, consisting of a series of gates, which can be:
\begin{itemize}[leftmargin=*]
    \item The identity
    \item A single-qubit $U_3(\theta, \phi, \lambda)$ gate
    \item A global phase gate
    \item A reset gate, which measures and resets a qubit to the $\ket{0}$ state
    \item An error-measurement gate, which measures a qubit and signals an error if it is in the $\ket{1}$ state
    \item A swap gate
    \item A controlled gate, with a list of control qubits
    \item Special zero-state labels and potential deletion labels, whose uses are detailed in Section \ref{sec:postprocessing}.
\end{itemize}

These gates can very straightforwardly be converted into OpenQASM 3 code. However, each gate must already contain information that specifies the indices of the qubits it is applied to, which makes this data structure lack modularity. When compiling the intermediate representation into low-level circuits, we often want to separately construct two circuit components representing different operators, and then feed the output of one into the input of the other, or maybe reuse the same input and prep qubits for controlled versions of two different operators (see Section \ref{sec:improve_dirsum}), which is difficult to do with just the gates. The Qunity compiler uses a system of \textit{circuits} and \textit{circuit specifications} to allow for the construction of low-level circuits out of smaller components. This system is constructed in a way that facilitates the tracking of qubits in various registers to ensure that the components fit together properly.

A circuit consists of a series of gates combined with additional data that specifies the roles of the qubits used by it - specifically, which qubit indices constitute the input registers, the output registers, the prep register (additional inputs prepared in the $\ket{0}$ state), the flag register (output qubits expected to be in the $\ket{0}$ state in the absence of error), and the garbage register (output qubits that are measured and discarded, corresponding to a partial trace operation).

A circuit specification is a wrapper around a function, \texttt{circ\_fun}, that can construct a circuit given a list of input registers, a set of used wires (qubits that are in use elsewhere and cannot be used as new prep qubits), and additional instantiation settings that can dictate certain aspects of how the circuit is built. In addition to the circuit, this function also outputs the updated set of used wires. The circuit specification also includes data on the required input and output register dimensions. Circuit specifications are more modular than gates or circuits, because they are not tied to specific qubit indices - these are not realized until their \texttt{circ\_fun} is called. Thus, it is possible to easily construct circuit specifications from other circuit specifications through various circuit construction functions. For instance, one such function could take in two circuit specifications, and define its \texttt{circ\_fun} to instantiate the first specification into a circuit, and then use its output registers and used wires set to instantiate the second specification into a circuit, and then output a combined circuit whose gate is a sequence of the two original circuits' gates. The possible ways in which circuit specifications can be combined correspond closely to the low-level circuits defined in \AppRef{K}{app:low_level_compilation}.

\subsection{The Intermediate Representation}

The intermediate representation used in the Qunity compiler describes high-level circuit diagrams in which the wires correspond to Hilbert spaces associated with Qunity types and contexts. This representation is built from \emph{operators} that take in some list of input registers and output some registers. This level abstracts away the ``flag'' and ``garbage'' utility registers and can thus describe operators with non-unitary semantics. Operators in the intermediate representation can be either built from smaller components using primitive constructors (corresponding to specific low-level circuit components), or they can be constructed from a series of commands in a fashion similar to a simple imperative language. These custom operators can associate registers with variable names and apply other operators to them, enforcing the condition that every variable must be used exactly once. This design makes it convenient to describe high-level circuits corresponding to each Qunity typing rule.

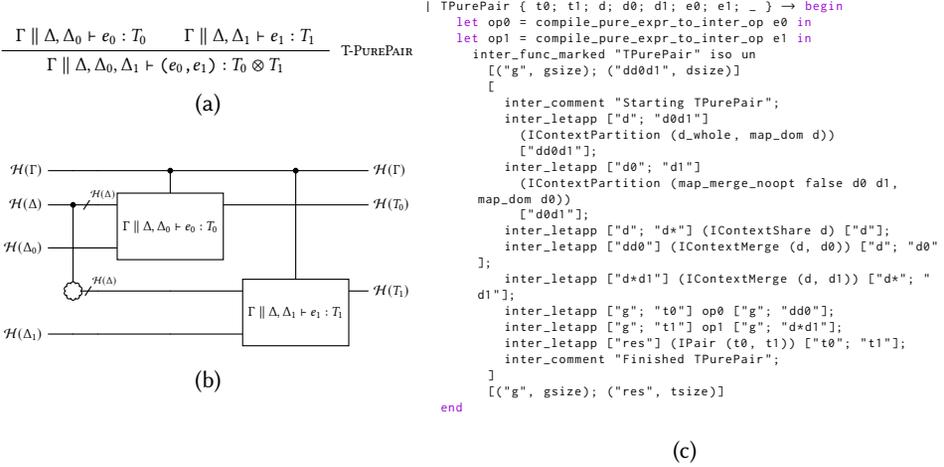
\begin{figure}[ht]
\centering
\begin{subfigure}[t]{0.4\textwidth}
    \begin{subfigure}[t]{\textwidth}
        \vspace{0.5cm}
        \centering
        \resizebox{\textwidth}{!}{%
        \(
       	\inference{\Gamma\partition \Delta, \Delta_0 \vdash e_0 : T_0 \qquad \Gamma\partition \Delta, \Delta_1 \vdash e_1 : T_1}{\Gamma\partition \Delta, \Delta_0, \Delta_1 \vdash \pair {e_0} {e_1} : T_0 \otimes T_1}[\textsc{T-PurePair}]
        \)
        }%
        \caption{}
    \end{subfigure}
    \hfill
    \begin{subfigure}[t]{\textwidth}
        \vspace{0.5cm}
        \centering
        \resizebox{\textwidth}{!}{%
        \begin{quantikz}
    	\lstick{$\cH(\Gamma)$} &&& \ctrl{1} & \ctrl{3} & \rstick{$\cH(\Gamma)$} \\
    	\lstick{$\cH(\Delta)$} & \ctrl{2} & \qwbundle{\cH(\Delta)} & \gate[2]{{\Gamma \partition \Delta, \Delta_0 \vdash e_0 : T_0}} && \rstick{$\cH(T_0)$} \\
    	\lstick{$\cH(\Delta_0)$} &&& \\
         \setwiretype{n} & \gate[style={cloud}]{} & \setwiretype{q} \qwbundle{\cH(\Delta)} && \gate[2]{{\Gamma \partition \Delta, \Delta_1 \vdash e_1 : T_1}} & \rstick{$\cH(T_1)$} \\
    	\lstick{$\cH(\Delta_1)$} &&&&
    \end{quantikz}
        }%
        \caption{}
        \label{fig:inter_rep_example_diagram}
    \end{subfigure}
\end{subfigure}
    \begin{subfigure}[t]{0.5\textwidth}
        \centering
        \begin{lstlisting}[language=caml,style=small]
| TPurePair { t0; t1; d; d0; d1; e0; e1; _ } -> begin
    let op0 = compile_pure_expr_to_inter_op e0 in
    let op1 = compile_pure_expr_to_inter_op e1 in
      inter_func_marked "TPurePair" iso un
        [("g", gsize); ("dd0d1", dsize)]
        [
          inter_comment "Starting TPurePair";
          inter_letapp ["d"; "d0d1"]
            (IContextPartition (d_whole, map_dom d))
            ["dd0d1"];
          inter_letapp ["d0"; "d1"]
            (IContextPartition (map_merge_noopt false d0 d1, map_dom d0))
            ["d0d1"];
          inter_letapp ["d"; "d*"] (IContextShare d) ["d"];
          inter_letapp ["dd0"] (IContextMerge (d, d0)) ["d"; "d0"];
          inter_letapp ["d*d1"] (IContextMerge (d, d1)) ["d*"; "d1"];
          inter_letapp ["g"; "t0"] op0 ["g"; "dd0"];
          inter_letapp ["g"; "t1"] op1 ["g"; "d*d1"];
          inter_letapp ["res"] (IPair (t0, t1)) ["t0"; "t1"];
          inter_comment "Finished TPurePair";
        ]
        [("g", gsize); ("res", tsize)]
  end
        \end{lstlisting}
        \vspace{-0.25cm}
        \caption{}
        \label{fig:inter_rep_example_code}
    \end{subfigure}
    \caption{(a) The typing judgment \textsc{T-PurePair} for typing pairs of Qunity expressions with a product type. \\ (b) The compilation circuit for \textsc{T-PurePair} as presented in \AppRef{L.3}{app:typing_judgment_compilation}. \\ (c) Excerpt from the \texttt{compile\_pure\_expr\_to\_inter\_op} function in the Qunity compiler, written in OCaml. This shows the compilation of the \textsc{T-PurePair} typing judgment into the intermediate representation, corresponding to the circuit in (b).}
    \label{fig:inter_rep_example}
    \Description{}
\end{figure}

In the first step of the compilation procedure, the typing judgment proofs obtained from Qunity expressions by the typechecker are converted into custom operators in the intermediate representation. This can be done easily due to the close correspondence between the intermediate representation and the circuit diagrams described in \AppRef{L}{app:high_level_compilation}. Figure \ref{fig:inter_rep_example_code} shows the code for compiling \textsc{T-PurePair} into the intermediate representation. We can see that this describes an operator that is a map of the form $\cH(\Gamma) \otimes \cH(\Delta, \Delta_0, \Delta_1) \rarr \cH(\Gamma) \otimes \cH(T)$. First, the typing judgment proofs of $e_0$ and $e_1$ are compiled into the intermediate representation. Then, they are used in the \texttt{inter\_func\_marked}, which is a wrapper around the \texttt{IFunc} (user-defined operator) constructor, adding an extra operator around it using isometry and unitary judgment information. The \texttt{inter\_comment} is just used to create annotations for debugging purposes. Next, the input register corresponding to the space $\cH(\Delta, \Delta_0, \Delta_1)$ is partitioned into separate registers corresponding to $\cH(\Delta)$, $\cH(\Delta_0)$, and $\cH(\Delta_1)$. Note that the compiler maintains an invariant that in a register corresponding to a context, the values of all variables must be stored in lexicographic order. Thus, the partitioning and merging circuits are not necessarily the identity and may involve some rearrangement of wires, which is done at the low level through a special circuit construction function. Then, the share gate (whose low-level implementation involves a new prep register of the same size as the input and CNOT gates) is applied, followed by the previously compiled operator circuits, and the result is combined into a single register. While it can be somewhat verbose, this code is essentially equivalent to the circuit diagram shown in Figure \ref{fig:inter_rep_example_diagram}.

\section{Low-Level Optimizations}\label{sec:low_level_opt}

This section describes optimizations that modify the second step of the Qunity compilation procedure (where the intermediate representation is instantiated into low-level circuits that keep track of qubit indices) and additional post-processing steps that are applied to the resulting circuits. Each time a circuit component is instantiated, the compiler needs to recursively instantiate its constituent parts, and decide how to connect the inputs and outputs of the parts to each other. This process may also require allocation of additional prep qubits, which increases the complexity of the generated circuit. The goal of the low-level optimizations is to improve the process of constructing these low-level circuits, to avoid unnecessarily using many qubits at once.

\subsection{Improving the Direct Sum Circuits}\label{sec:improve_dirsum}

An important component of the second step of the compilation procedure is being able to implement the direct sum of two operators. Given circuits $U_0$ implementing $E_0 : \cH(T_0) \rarr \cH(T_0')$ (with input size $s_0$, output size $s_0'$, prep size $p_0$, and flag size $f_0$), and $U_1$ implementing $E_1 : \cH(T_1) \rarr \cH(T_1')$ (with corresponding sizes $s_1, s_1', p_1, f_1$), the goal is to implement the direct sum circuit $E_0 \oplus E_1 : \cH(T_0 \oplus T_1) \rarr \cH(T_0' \oplus T_1')$. The implementation must be in accordance with Definition \ref{def:possible_to_implement_op}.

An initial na\"ive implementation that still respects the encoding-validity-preserving property specified in \Cref{def:possible_to_implement_op} uses $\min\{s_0, s_1\} + p_0 + p_1$ prep wires. Because the number of prep wires scales with the size of the input register, this can cause the number of qubits used by the circuit to grow rapidly when an operator is summed with itself multiple times, which is a common pattern in the compilation of quantum control.

We have created an improved version of the direct sum circuit, which still ensures encoding validity but reuses input and prep qubits as much as possible. \Cref{tab:dirsum_counts} shows the number of prep qubits needed for this improved circuit in different cases. The full circuits and a proof of correctness for them are shown in \AppRef{I}{app:dirsum_proof}.

\begin{table}[ht]
    \caption{Prep qubits needed for the improved direct sum circuit in $4$ different cases. The $4$ other cases can be obtained from these ones by applying a commutativity isomorphism (which is just an $X$ gate).}
    \label{tab:dirsum_counts}
    \centering
    \begin{tabular}{cc}
    \toprule
    Case & Number of required prep qubits\\
    \midrule
    $s_0 \geq s_1$, $s_0' \geq s_1'$, $p_0 \geq p_1$ & $p_0 + \max\{0, f_1 - f_0\}$ \\
    $s_0 \geq s_1$, $s_0' \geq s_1'$, $p_0 \leq p_1$ & $p_1 + \max\{0, f_1 - f_0 + p_0 - p_1\}$ \\
    $s_0 \geq s_1$, $s_0' \leq s_1'$, $p_0 \geq p_1$ & $p_0 + s_1' - s_0'$ \\
    $s_0 \geq s_1$, $s_0' \leq s_1'$, $p_0 \leq p_1$ & $p_1 + \max\{0, s_1' - s_0' + p_0 - p_1\}$ \\
    \bottomrule
    \end{tabular}
\end{table}

\subsection{Recycling of Garbage and Flag Qubits}\label{sec:recycling}

An important optimization to the Qunity compiler comes from the possibility of measuring and resetting garbage and flag qubits in the middle of a circuit and reusing them as new prep qubits for subsequent circuit components - that is, \emph{recycling} the garbage and flag qubits. This has to be done with some care, because there are situations in which we want to avoid resetting these qubits. Specifically, we may be taking the purification of a circuit, which feeds its garbage qubits into the output register, or we may be using an error-handling circuit, which uses the flag register to coherently toggle a new prep qubit, outputting values in the space $\cH \oplus \CC$, where $\cH$ is the output space of the original circuit. This means that a circuit specification has to ``know,'' at the time of its instantiation into a circuit, whether or not it is allowed to recycle its garbage or flag register. This is done by using the instantiation settings, which are passed to a circuit specification at the time of instantiation, and are used to decide whether the qubits can be recycled. Purification is implemented as a function that takes in a circuit specification and outputs another one, whose instantiation function passes in a value of \texttt{false} for the instantiation settings' \texttt{reset\_garb} field when calling the wrapped circuit specification's instantiation function. The same thing happens with the error-handling circuits and the \texttt{reset\_flag} field. This optimization reduces the qubit count of the circuits output by the compiler, since now, if there are many circuit components in series that each can throw an error and are not part of a larger error-handling component, the compiler does not have to allocate new qubits every time.

While more advanced recycling techniques exist, such as the one found in \citet{qubit_recycling}, the simple recycling technique we introduced here is closely tied to the specifics of Qunity's compilation procedure, and thus is not directly comparable. A potential future step would be to incorporate some of the more advanced qubit recycling methods from the literature as part of the \emph{postprocessing} stage, where they can be directly applied to a quantum circuit without the need to integrate with any Qunity-specific data structures.

\subsection{Post-Processing Optimizations}\label{sec:postprocessing}

We have introduced a post-processing step that optimizes the quantum circuits generated by the Qunity compiler. The overall procedure for low-level optimization converts the input circuit into a list of gates and then repeatedly applies \emph{optimization passes} to this list, until it reaches a fixed point. In a single optimization pass, the gate list is processed from left to right, detecting possible patterns that can be optimized. This process is guaranteed to eventually terminate, since each change made to the circuit has to decrease the number of qubits or gates. The currently implemented pass includes the following:

\begin{itemize}
\item Canceling a gate that is immediately followed by its adjoint.
\item Combining adjacent controlled-$U$ and anti-controlled-$U$ gates into just $U$.
\item Removing physical swap gates without controls and relabeling the subsequent wires.
\item Performing a commutation pass: taking a gate and commuting it to the right until either it cannot commute past something (in which case it is reverted to its original position) or it cancels with something.
\item Classical propagation - moving through the circuit and keeping track of qubits we know to be in a classical state, modifying any gates that are controlled on them by removing controls or deleting the gates.
\item Detecting regions where a CNOT gate is applied, sharing one qubit to another one initially in the $\ket{0}$ state, then any other gates are the Pauli $X$ gate or only use either qubit as a control, and then an identical CNOT is applied. In this case, it is possible to transfer the controls on the second qubit to the first one, which can make it possible to delete the second qubit with the procedure described below.
\item Deleting qubits (or gates on a qubit between measurements) if certain conditions are met.
\end{itemize}

We will now discuss this last point in more detail. This mechanism relies on high-level information from the \emph{isometry judgment} in the Qunity typechecker (see \AppRef{G}{app:isometry_judgment}) to identify candidate qubits for potential deletion. These candidates are either garbage qubits or flag qubits that we are certain will always be in the $\ket{0}$ state when measured.

The candidate qubits are identified during the procedure of instantiating a circuit specification into a concrete circuit. The instantiation settings, which include Boolean fields \texttt{reset\_garb} and \texttt{reset\_flag} for whether or not the garbage and flag registers should be reset, as described in Section \ref{sec:recycling}, also contain a field \texttt{iso} for whether or not this circuit can be interpreted as an isometry and thus whether we can expect the flag qubits to be in the $\ket{0}$ state. This variable is set to \texttt{true} by a special wrapper circuit specification that acts very similarly to the purification and error-handling circuits, and is added to the circuit based on the results of the isometry judgment in the first step of the compilation process. Now, during instantiation, if \texttt{reset\_garb} is \texttt{true}, then a special gate called a \emph{potential deletion label} is added to the qubits in the garbage register. Similarly, if both \texttt{reset\_flag} and \texttt{iso} are \texttt{true}, the potential deletion labels are also added to the flag register. Additionally, instead of measuring flag qubits known to be in the $\ket{0}$ state, the compiler adds a special \textit{zero-state label} onto them.

However, not all candidate qubits are safe to delete because they may be entangled with the output qubits at some point during the computation, possibly being used as ancillas. In the post-processing stage, all the deletion labels are first shifted to the left as far as possible until they hit a measurement gate or a zero-state label to ensure that they are selecting an entire region of interest. Then, during an optimization pass, when seeing a potential deletion label, the segment between it and the next measurement is evaluated to determine whether it is safe to delete. The procedure determines that all gates in the region involving the selected qubit are safe to delete if they are all Pauli $X$ gates, controlled $X$ gates with the qubit as a target, or uncontrolled global phase gates. This is because if the qubit starts in the $\ket{0}$ state, then applying only these gates will not have any effect on the reduced density matrix of the system obtained by removing the qubit. At the end of the optimization procedure, if a qubit has no more gates remaining on it, it is completely removed from the circuit.

\section{High-Level Optimizations}\label{sec:high_level_opt}

This section describes optimizations that modify the first step of the Qunity compilation procedure, where the typing judgment proof structures are converted to the intermediate representation.

\subsection{Simplifying the Orthogonality Circuit}

A significant part of the circuit for compiling the \lstinline|ctrl|, \lstinline|match|, and \lstinline|pmatch| constructs is the following circuit component:

\begin{quantikz}
\lstick{$\cH(T)$} & \gate{\msem{\ortho{T}{e_1, \dots, e_n}}} & \qwbundle{\cH(T)^{\oplus n}} && \gate{\bigoplus_j \msem{e_j}\adj} & \rstick{$\bigoplus_j \cH(\Delta_j)$}
\end{quantikz}

This circuit is quite inefficient: it uses $O(n)$ extra qubits to store the register corresponding to $\cH(T)^{\oplus n}$, and it involves many expensive associativity isomorphisms to transform the direct sum structure of the given space to something of the form $\cH(T) \oplus (\cH(T) \oplus (\cH(T) \oplus \dots))$. In this work, we introduce two major changes to this: allowing direct sums of operators to be taken \emph{along a binary tree}, and replacing the two above circuit components with a single one that goes directly from $\cH(T)$ to $\bigoplus_j \cH(\Delta_j)$, without the intermediate step of $\cH(T)^{\oplus n}$. Refer to \AppRef{J}{app:binary_trees} for the definitions and notations used to describe sums of operators along a binary tree and leveling a tree to a given height.

We can derive a binary tree from the structure of the orthogonality judgment (\AppRef{F}{app:ortho_judgment}), as follows:

\begin{definition}[Tree derived from the orthogonality judgment]\label{def:ortho_tree}\leavevmode
\begin{itemize}
\item The trees derived from \textsc{O-Void}, \textsc{O-Unit}, and \textsc{O-Var} are all $\Leaf$.
\item The tree derived from \textsc{O-IsoApp} is the same as that derived from $\ortho{T}{e_1, \dots, e_n}$.
\item The \textsc{O-Sum} rule states that if $e_1, \dots, e_m$ are orthogonal in $T$ and $e_1', \dots, e_m'$ are orthogonal in $T'$, then the left injections of the $e_j$ and the right injections of the $e_j'$ together are orthogonal in $T \oplus T'$. The tree derived from \textsc{O-Sum} is a root whose left subtree is the tree derived from the \texttt{left} expressions and whose right subtree is the tree derived from the \texttt{right} expressions.
\item The \textsc{O-Pair} rule states that if $e_1, \dots, e_m$ are orthogonal in $T$, and for each $1 \leq j \leq m$, the expressions $e'_{j,1}, \dots, e'_{j,n_j}$ are orthogonal in $T'$, then the list of pairs $(e_j, e'_{j,k})$ are orthogonal in $T \otimes T'$. The tree derived from \textsc{S-Pair} is constructed by taking the tree $\cR_0$ derived from $\ortho{T}{e_1, \dots, e_m}$, and replacing each $j$th leaf with the tree $\cR_1^j$ derived from $\ortho{T'}{e'_{j,1}, \dots, e'_{j,n_j}}$. This essentially decomposes a tensor product into a direct sum structure.
\item The tree derived from from \textsc{O-Sub} is constructed by removing all subtrees whose leaves only contain discarded expressions.
\end{itemize}
\end{definition}

Now, we want to define $\msem{\ortho{T}{e_1, \dots, e_n}}$ in such a way that when $\ortho{T}{e_1, \dots, e_n}$ holds with tree structure $\cR$, then for all $j$, for all $\tau_j \in \VV(\Delta_j)$, we have that
\[
\msem{\ortho{T}{e_1, \dots, e_n}}\msem{e_j}\ket{\tau_j} = \inj_j^\cR \ket{\tau_j}.
\]
That is, we are constructing an operator
\[
\msem{\ortho{T}{e_1, \dots, e_n}} : \cH(T) \rarr \bigoplus_{j : \cR} \cH(\Delta_j).
\]
This construction is given in \AppRef{L.2}{app:ortho_compilation}.

With this change, we can now modify the circuit for the compilation of \textsc{T-Ctrl} to use a single circuit component that converts the input $\cH(T)$ directly into a direct sum of the $\Delta_j$ over the tree derived from the orthogonality judgment. Then, for the rest of the circuit, all direct sums can be taken along this tree instead of a standardized associativity structure.

\subsection{Compilation of \texttt{pmatch}}\label{sec:pmatch_compilation}

The same construction as above can be used to construct a high-level circuit to compile \textsc{T-Pmatch} into the intermediate representation. This can be done as follows:

\resizebox{0.95\textwidth}{!}{
\begin{quantikz}
\lstick{$\cH(T)$} & \gate{\msem{\ortho{T}{e_1, \dots, e_n}}} & \qwbundle{\bigoplus_{j:\cR_0} \cH(\Delta_j)} &&& \gate{\textsc{TreeRearrange}(\cR_0, \cR_1)} & \qwbundle{\bigoplus_{j:\cR_1} \cH(\Delta_j)} &&& \gate{\msem{\ortho{T}{e_1, \dots, e_n}}\adj} & \rstick{$\cH(T')$}
\end{quantikz}
}

where $\cR_0$ and $\cR_1$ are binary trees derived from the respective orthogonality judgments, except their leaves are also labeled with the indices of the corresponding expressions. Then, \textsc{TreeRearrange} is an algorithm for transforming one binary tree into another using only operations from a certain set, encoded as a quantum circuit.

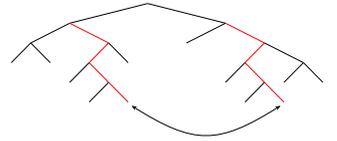
\begin{wrapfigure}{r}{0.3\textwidth}
\centering
\resizebox{\linewidth}{!}{%
\begin{circuitikz}
\tikzstyle{every node}=[font=\LARGE]
\draw [line width=1pt, short] (8.75,12) -- (3.75,10.75);
\draw [line width=1pt, short] (8.75,12) -- (13.75,10.75);
\draw [line width=1pt, short] (3.75,10.75) -- (1.25,9.5);
\draw [ color={rgb,255:red,255; green,0; blue,0}, line width=1pt, short] (3.75,10.75) -- (6.25,9.5);
\draw [line width=1pt, short] (13.75,10.75) -- (11.25,9.5);
\draw [ color={rgb,255:red,255; green,0; blue,0}, line width=1pt, short] (13.75,10.75) -- (16.25,9.5);
\draw [line width=1pt, short] (1.25,9.5) -- (0,8.25);
\draw [line width=1pt, short] (1.25,9.5) -- (2.5,8.25);
\draw [ color={rgb,255:red,255; green,0; blue,0}, line width=1pt, short] (6.25,9.5) -- (5,8.25);
\draw [line width=1pt, short] (6.25,9.5) -- (7.5,8.25);
\draw [line width=1pt, short] (5,8.25) -- (3.75,7);
\draw [ color={rgb,255:red,255; green,0; blue,0}, line width=1pt, short] (5,8.25) -- (6.25,7);
\draw [line width=1pt, short] (6.25,7) -- (5,5.75);
\draw [ color={rgb,255:red,255; green,0; blue,0}, line width=1pt, short] (6.25,7) -- (7.5,5.75);
\draw [ color={rgb,255:red,255; green,0; blue,0}, line width=1pt, short] (16.25,9.5) -- (15,8.25);
\draw [ color={rgb,255:red,255; green,0; blue,0}, line width=1pt, short] (15,8.25) -- (16.25,7);
\draw [ color={rgb,255:red,255; green,0; blue,0}, line width=1pt, short] (16.25,7) -- (17.5,5.75);
\draw [line width=1pt, short] (16.25,7) -- (15,5.75);
\draw [line width=1pt, short] (15,8.25) -- (13.75,7);
\draw [line width=1pt, short] (16.25,9.5) -- (18.75,8.25);
\draw [line width=1pt, short] (18.75,8.25) -- (17.5,7);
\draw [line width=1pt, short] (18.75,8.25) -- (20,7);
\draw [line width=1pt, <->, >=Stealth] (7.75,5.5) .. controls (11.75,3) and (13.25,3) .. (17.25,5.5);
\end{circuitikz}
}%
\caption{An example of the conditional commutation tree transformation, which swaps the two vertices shown in this diagram. Their paths must be identical except for the very first step from the root.
}
\label{fig:conditional_commute}
\end{wrapfigure}

The possible operations are:
\begin{itemize}
\item A right tree rotation, which can be implemented as the associativity isomorphism $(T_0 \oplus T_1) \oplus T_2 \rarr T_0 \oplus (T_1 \oplus T_2)$ (low-level circuit implementation in \AppRef{K}{app:low_level_compilation}).
\item A left tree rotation, implementable as the adjoint of the associativity isomorphism.
\item Switching the left and right subtree, which can be implemented with a single Pauli $X$ gate.
\item Other transformations applied to the subtrees of a given tree, which can be implemented using the direct sum circuit.
\item A sequence of two transformations, which can be implemented as a sequence of two circuits.
\item A special \textit{conditional commutation} transformation, shown in Figure \ref{fig:conditional_commute}. It can be implemented with a multi-controlled $X$ gate, where the target is the qubit corresponding to the root node, and the controls have states corresponding to the path taken.
\end{itemize}

The algorithm for transforming trees, expressed in terms of these allowed operations, is:

\begin{enumerate}
\item Transform the input tree so it has the same shape as the goal, ignoring the values at the leaves. This can be done by transferring nodes from one subtree to another making sure they have the correct numbers, and then making recursive calls on each subtree. The transferring of nodes can be accomplished by a sequence of tree rotations.
\item For each leaf value that needs to be transferred to a different location, do the following to switch the leaves at the two locations: \begin{enumerate}
    \item Bring the leaf at the larger depth (the one farther from the root) to the same depth as the other one. This can be accomplished by a sequence of tree rotations and commutations.
    \item Make the trees have identical paths from the root except for the first place in which their paths differ. This can be done by a sequence of commutation operations applied to the appropriate subtrees.
    \item Apply the special conditional commutation operation to switch the two leaves.
    \item Apply the operations in the first two steps in reverse order so that the only effect of the entire transformation is swapping the two leaves of interest.
\end{enumerate}
\end{enumerate}

In the OCaml implementation of the compiler, running this algorithm creates a record of all the tree transformations that need to be applied to perform the operation. Then, the resulting tree transformations are converted into the intermediate representation in the manner described above, thus completing the compilation of \lstinline|pmatch|. Beside the fact that this allows \lstinline|pmatch| to work in an arbitrary basis, it is also much more efficient than implementing a classical-basis version of it as syntactic sugar over \lstinline|ctrl|: the tree rearrangement algorithm attempts to avoid significantly increasing the height of the tree at any intermediate step, which means that it does not unnecessarily use additional qubits.

\section{Evaluation}\label{sec:results}

In this section we detail our evaluation of the Qunity compiler. We primarily focus on the effectiveness of the compiler in reducing compiled circuit size, but we also report on our 
differential unit testing that we used to gain  confidence in the correctness of our implementation.

\paragraph*{Compilation Efficiency}

We evaluate the Qunity compiler on several benchmark programs, shown in \Cref{tab:compilation_benchmarks}. We compare the most recent version of the compiler with the initially implemented version, which followed directly the compilation procedure of \citeauthor{Voichick_2023} (with the direct sum circuit corrected to preserve encoding validity but not optimized). This baseline did not support the new pattern-matching constructs, so where possible, the benchmark code was rewritten to equivalent expressions using only \lstinline|ctrl|. We also compare the performance of our compiler to low-level circuits constructed directly in Qiskit~\cite{Qiskit} version 1.1.1, using circuits from \lstinline[breaklines=false]|qiskit.circuit.library| when relevant. The gate counts reported in this table are after the circuit has been transpiled by Qiskit with \lstinline[language=python]{basis_gates=["u3", "cx"], optimization_level=3} (to only use single-qubit gates and CNOT).

\begin{table}[ht]
\caption{Benchmark tests comparing the performance of the optimized and unoptimized Qunity compilers with an implementation in Qiskit. The code used for these benchmarks is shown in \AppRef{D}{app:benchmark_test_programs}.
The missing entries in the unoptimized Qunity compiler result from the unoptimized compiler not supporting the new pattern-matching constructs. We also do not have a suitable reference for the order-finding circuit. The entry marked with ?? has such a large circuit that it was  infeasible to transpile it to obtain a gate count.}
\label{tab:compilation_benchmarks}
\centering
\begin{tabular}{>{\raggedright\arraybackslash}p{3.5cm}p{1cm}p{1cm}p{1cm}p{1cm}p{1cm}p{1cm}}
\toprule
\multirow{2}{*}{Benchmark} & \multicolumn{2}{c}{Qunity (unoptimized)} & \multicolumn{2}{c}{Qunity (optimized)} & \multicolumn{2}{c}{Qiskit} \\
\cmidrule{2-7}
& Qubits & Gates & Qubits & Gates & Qubits & Gates \\
\midrule
Phase conditioned on AND of $5$ qubits & 769 & \num{147871} & 9 & 121 & 5 & 101 \\[6mm]
Quantum Fourier Transform ($5$ qubits) & 15 & 282 & 6 & 143 & 5 & 54 \\[6mm]
Phase estimation example ($5$ qubits) & 75 & 1465 & 6 & 275 & 5 & 60 \\[6mm]
Order finding (\lstinline|#n=5, #a=13|) & -- & -- & 10 & 460 & -- & -- \\[6mm]
Reversible CDKM Adder~\cite{cdkm_ripplecarry} ($5$ bits) & 1031 & \num{160321} & 11 & 152 & 11 & 179 \\[6mm]
Grover ($5$-bit match oracle, $1$ iteration) & 1131 & ?? & 7 & 760 & 6 & 389 \\[6mm]
Grover (List sum oracle, \lstinline|#n=2|, $1$ iteration) & -- & -- & 11 & \num{13470} & 5 & 116 \\
\bottomrule
\end{tabular}
\end{table}

The results show a very significant improvement of the optimized Qunity compiler compared to the unoptimized one. In many instances, the optimized Qunity compiler reaches a gate count that is comparable to the handcrafted Qiskit implementation, whereas the baseline Qunity compiler resulted in orders of magnitude larger circuits. Still there is room for improvement: to uncover the limits of our compiler implementation, we included a benchmark that implements Grover's algorithm using the list sum oracle, which compiles into \num{13470} low-level gates, while it is possible to manually construct a small circuit that accomplishes the same task. The inefficiencies arise due to isomorphisms and transformations between types that are introduced during the compilation procedure: while the optimizations introduced in \Cref{sec:high_level_opt} help to simplify them to a large extent, there are still many inefficiencies present that are nontrivial to eliminate.

\begin{figure}
    \centering
    \begin{subfigure}{0.5\textwidth}
        \centering
        \includegraphics[width=\linewidth]{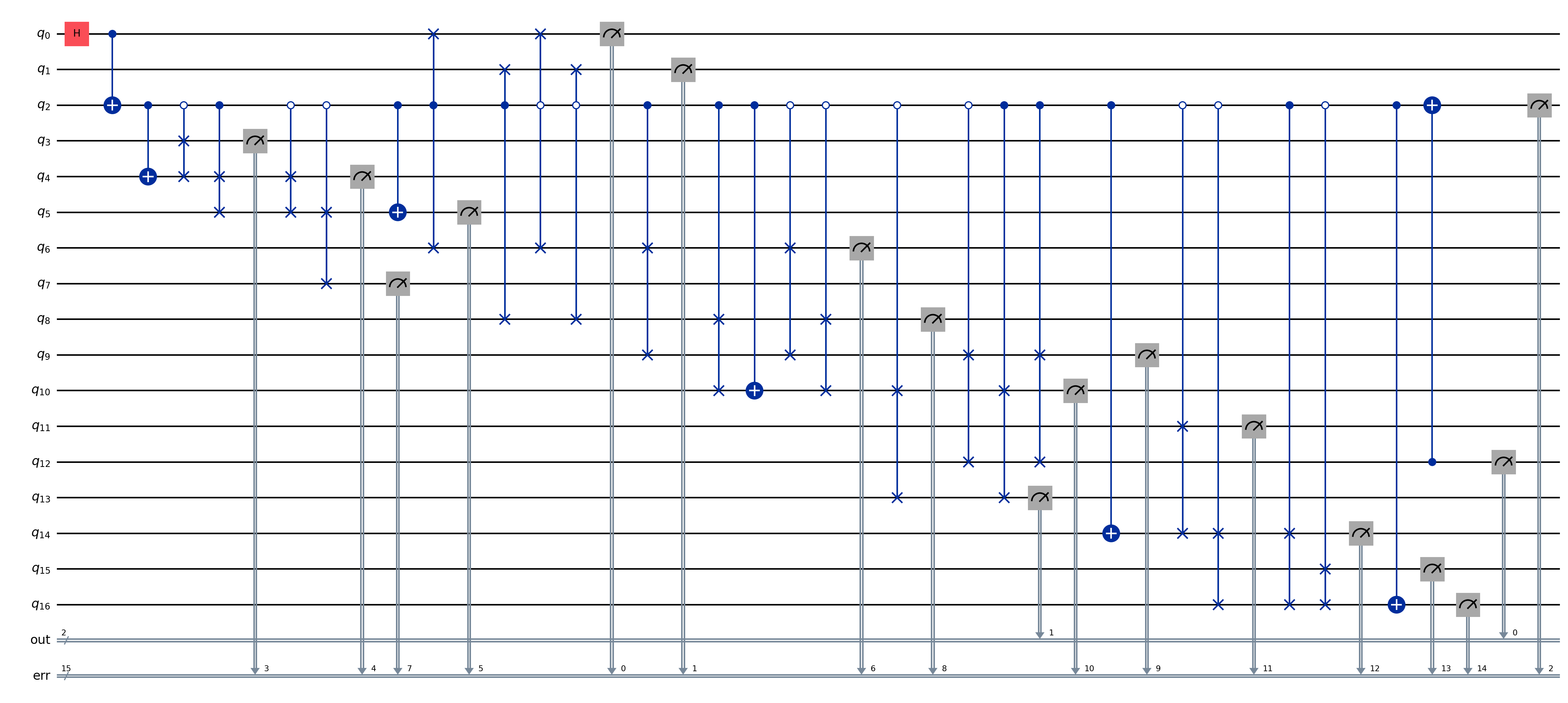}
        \caption{}
        \label{fig:bell_circuit_bad}
    \end{subfigure}
    \begin{subfigure}{0.2\textwidth}
        \centering
        \includegraphics[width=\linewidth]{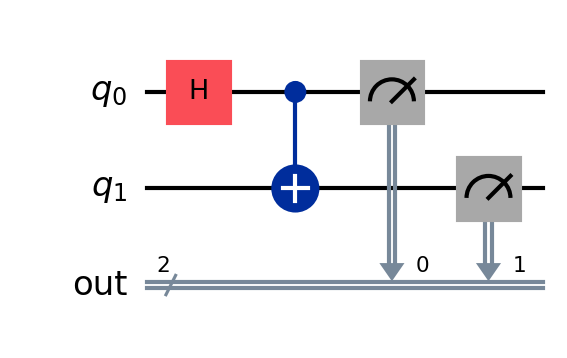}
        \caption{}
        \label{fig:bell_circuit_good}
    \end{subfigure}
    \caption[]{(a) The circuit produced by compiling a Qunity program for preparing the state $\frac{1}{\sqrt{2}}(\ket{00} + \ket{11})$ without compiler optimizations. 
    (b) The circuit produced using the the optimized compiler for the same program.\protect\footnotemark}
    \Description{}
    \label{fig:bell_circuit}
\end{figure}

\footnotetext{These circuits have measurements because the output register is always measured at the end when running the Qunity compiler.}

Still, for most programs, including smaller ones that we used for testing, our Qunity compiler
results in practical circuits. For one example, \Cref{fig:bell_circuit} demonstrates the improvement in the efficiency of the compiler resulting from the optimizations implemented in this work when preparing a Bell state using the quantum control construct. Note that the simplest way to prepare this state in Qunity is using entanglement through variable reuse: \lstinline!&plus |> lambda x -> (x, x)!, which compiles to the circuit in \Cref{fig:bell_circuit_good} even with the unoptimized compiler. However, for this example, the entanglement is achieved by defining a CNOT gate in terms of the \lstinline{ctrl} construct, which has a much more complicated compilation procedure. (This can be used to write circuit-style programs in Qunity: see \AppRef{E}{app:circuit_style}). The circuit produced by the unoptimized compiler, in \Cref{fig:bell_circuit_bad}, needs $17$ qubits and $27$ CNOT and controlled-SWAP gates just to represent an operation equivalent to a single CNOT gate. This is due to a combination of several factors, including the unoptimized direct sum circuits, the unnecessary associativity isomorphisms, and the redundancy in the original orthogonality circuit. It is difficult to simplify this circuit by post-processing alone: the higher-level optimizations described in this work are necessary. Using the combination of these optimization techniques, the compiler is able to reduce this circuit to a form that can be more easily handled by the post-processing optimizations and converted to the circuit in \Cref{fig:bell_circuit_good}.

\paragraph*{Differential Unit Testing}
Finally, we also created a test suite to verify the correctness of the compilation procedure, checking it against the output of the interpreter. The test proceeds as follows: given a Qunity expression $e$ of type $T$, we evaluate its semantics as a density matrix $\rho = \msem{\varnothing : \varnothing \partition \varnothing \Vdash e : T}$ using the Qunity interpreter. We run the compiler to generate a low-level qubit circuit, and simulate this circuit. Taking a projection onto the subspace where the flag register is in the $\ket{0}$ state and taking a partial trace over the garbage registers, we obtain a state $\tilde{\rho}$. Then, for each $v \in \VV(T)$, we check that
$\bra{v}\rho\ket{v} = \bra{\enc(v)}\tilde{\rho}\ket{\enc(v)}.$

Most programs in this suite are small, and many are just simple demonstrations of various patterns (such as pattern matching constructs nested in different ways, operations on datatypes, and error handling). We exclude programs from \Cref{tab:compilation_benchmarks} that are too expensive to simulate.

\section{Related Work}

Many existing quantum programming languages are based on the paradigm of ``quantum data, classical control,'' which was popularized by Selinger~\cite{Selinger2004a}. In this model, quantum circuits can be constructed by a classical computer, using classical data which may be obtained by performing measurements on qubits. Most of the widely-used quantum languages rely on this paradigm: from low-level languages like OpenQASM~\cite{OpenQASM}, Cirq~\cite{Cirq}, and Qiskit~\cite{Qiskit} to higher-level languages like Quipper~\cite{Quipper}, Q\#~\cite{Qsharp} and Qualtran~\cite{qualtran}.

However, we often want to express algorithms that use quantum control flow and manipulate complex data structures in quantum superposition: for instance, Grover's algorithm \cite{grover1996} over arbitrary oracles or quantum walk-based algorithms such as~\citet{Childs_2007}. This runs into problems, as fully general quantum control flow is physically infeasible, as argued by~\citet{Badescu_2015} and later~\citet{yuan-limits}. A number of languages have attempted to address this issue through restricting control flow, including QML~\cite{QML}, Symmetric Pattern Matching (SPM)~\cite{SPM}, Silq~\cite{Silq}, and Tower~\cite{tower}. Of these, Silq and Tower compile to circuits.

Venev's~\cite{venev2024} compiler for Silq programs is focused on the problem of \emph{uncomputation}, specifically in compiling and optimizing away the results of the common compute-copy-uncompute pattern used in the generation of quantum oracles.
Silq's uncomputation, however, is restricted to applications of \texttt{qfree} (strictly classical) functions - Qunity's \lstinline|ctrl| construct, on the other hand, enables uncomputation of an arbitrary mixed program (although doing this on a program that introduces superposition will lead to norm-decreasing semantics). While none of our examples make use of this feature, it may be useful in more advanced quantum algorithms such as the triangle finding algorithm \cite{triangle_finding_alg}, which uses a quantum subroutine that is itself a probabilistic quantum algorithm.

Tower's compiler, Spire~\cite{spire} is closer in its aims to this work. Like Qunity, Tower is a high-level language with sophisticated (QRAM-based) control flow, and Spire aims to optimize it at a high level. Spire uses two key optimizations: conditional flattening (which consists of combining nested if statements, reducing the controls on internal unitaries) and conditional narrowing, which pulls statements out of quantum ifs (again, reducing the number of controls). These two optimizations prove surprisingly valuable in practice, enabling a significant reduction in the circuits' $T$-complexity, and showing that high-level optimizations of quantum programs could be useful in practice.

Qunity is set apart from the previous work listed above by its focus on compositionality and its denotational semantics. It takes the unique and unusual approach of defining both a pure and mixed typing judgment, allowing for a richer interplay of reversible and irreversible programming. Qunity is similar to classical functional languages, while languages like Silq and Tower follow an imperative paradigm. The metaprogramming layer allows for higher-order functions and recursion without running into the fundamental limitations of implementing them as quantum constructs~\cite{Badescu_2015}. Qunity's pattern-matching constructs offer a novel and expressive way of expressing quantum control flow.

\section{Conclusion and Future Work}\label{sec:conclusion}

In this work, we show how to efficiently compile Qunity, a high-level quantum programming language with compositional semantics, by combining low-level optimizations that simplify the construction of qubit circuit components with higher-level optimizations that simplify manipulations of Hilbert spaces. 

There are more optimizations that we would like to implement in the future. For example, in instances where no purification is performed, the \lstinline|try|/\lstinline|catch| statement could be implemented by measuring the flag register, and applying some gates depending on the measurement outcome. Similarly, the \lstinline|match| construct could be implemented by measuring the scrutinee expression and using the result to evaluate one of the RHS expressions. Including this capability can simplify many Qunity programs with classical computation outside of quantum subroutines.

We could also change the target language of the Qunity compiler from OpenQASM to a quantum intermediate representation such as QIR \cite{qir}, MLIR \cite{mlir}, or HUGR \cite{hugr}. Note that these representations are not directly comparable to Qunity's IR, which is defined using operators closely associated to Qunity's control flow constructs. Compiling from Qunity's IR to one of these representations would allow Qunity to more easily take advantage of existing optimizers at the post-processing stage, and make the compilation to hybrid quantum-classical programs more feasible.

With improvements like these, we believe Qunity can truly become a practical quantum programming language, offering a promising alternative to circuit-based quantum languages and advancing the development of tools for high-level quantum programming and algorithm design.

\section*{Data Availability Statement}
The implementation of the Qunity interpreter and compiler is publicly available \cite{artifact}.

\begin{acks}
We thank Michael Vanier for his advice on this work. Mikhail Mints was supported by the Samuel P. and Frances Krown SURF fellowship from Caltech. This material is based upon work supported by the \grantsponsor{afosr}{Air Force Office of Scientific Research}{https://doi.org/10.13039/100000181} under Grant No. \grantnum{afosr}{FA95502310406}.
\end{acks}

\ifshort
\else

\newpage
\printbibliography

\appendix

\section*{Appendices}
\begin{enumerate}[label=\Alph*, align=parleft, leftmargin=*, labelsep=1em]
\item \hyperref[app:formal_defs]{Formal Definitions of Syntax, Typing Judgments, and Semantics}
\item \hyperref[app:surface_grammar]{Grammar of the Surface Syntax}
\item \hyperref[app:control_flow_comparison]{Comparison of Control Flow Constructs}
\item \hyperref[app:benchmark_test_programs]{Benchmark Test Programs}
\item \hyperref[app:circuit_style]{Circuit-Style Programming}
\item \hyperref[app:ortho_judgment]{The Orthogonality and Spanning Judgments}
\item \hyperref[app:isometry_judgment]{The Isometry Judgment}
\item \hyperref[app:erasure_judgment]{The Erasure Judgment}
\item \hyperref[app:dirsum_proof]{Direct Sum Circuit Correctness Proof}
\item \hyperref[app:binary_trees]{Notation and Definitions for Binary Trees}
\item \hyperref[app:low_level_compilation]{Proofs of Correctness for Low-Level Compilation}
\item \hyperref[app:high_level_compilation]{Proofs of Correctness for High-Level Compilation}
\end{enumerate}

\section{Formal Definitions of Syntax, Typing Judgments, and Semantics}\label{app:formal_defs}

\subsection{Syntax}\label{app:syntax}

\begin{figure}[ht]
  \centering
  \begin{minipage}[t]{.45\textwidth}
\begin{alignat*}{3}
	T \defeqq &\quad&&&\textit{(data type)} \\
						&&& \Void &\quad \textit{(bottom)} \\
						|&&& \Unit &\quad \textit{(unit)} \\
						|&&& T \oplus T &\quad \textit{(sum)} \\
						|&&& T \otimes T &\quad \textit{(product)} \\
						F \defeqq &\quad &&& \textit{(program type)} \\
						&&&T \rightsquigarrow T &\quad \textit{(coherent map)} \\
						|&&& T \Rrightarrow T &\quad \textit{(quantum channel)} \\
\end{alignat*}
\caption{Qunity types.}
\label{fig:types}
\smallskip
\begin{alignat*}{3}
	\Gamma &\defeqq&& &\textit{(context)} \\
								 &&& \varnothing & \textit{(empty)} \\
	&&|\;& \Gamma, x : T & \textit{(binding)} \\
	\dom(\varnothing) &\defeq&& \varnothing &        \textit{(dom-none)}\\
	\dom(\Gamma, x : T)&\defeq&& \dom(\Gamma) \cup \{ x \} \; & \textit{(dom-bind)}\\
\end{alignat*}
\caption{Typing contexts.}
\label{fig:context}
  \end{minipage}
  \hspace{.1in}
  \begin{minipage}[t]{.5\textwidth}
\begin{alignat*}{3}
	e \defeqq &&&&\quad\textit{(expression)} \\
    && \;& \unit &\quad \textit{(unit)} \\
    |&&& x &\quad \textit{(variable)} \\
    |&&& \pair e e &\quad \textit{(pair)} \\
    |&&& \cntrl{e}{T}{e &\mapsto e \\ &\cdots \\ e &\mapsto e}{T} &\quad \textit{(coherent control)} \\
    |&&& \match{e}{T}{e &\mapsto e \\ &\cdots \\ e &\mapsto e}{T} &\quad \textit{(decoherent match)} \\
    |&&& \trycatch{e}{e} &\quad \textit{(error recovery)} \\
    |&&& f\: e &\quad \textit{(application)} \\
f \defeqq &&&&\quad\textit{(program)} \\
    &&& \uthree r r r &\quad \textit{(qubit gate)} \\
    |&&& \lef T T &\quad \textit{(left tag)} \\
    |&&& \rit T T &\quad \textit{(right tag)} \\
    |&&& \lambda e \xmapsto{{\color{gray}{T}}} e &\quad \textit{(abstraction)} \\
    |&&& \rphase T e r r &\quad \textit{(relative phase)} \\
    |&&& \pmatch{T}{e &\mapsto e \\ &\cdots \\ e &\mapsto e}{T'} &\quad \textit{(symmetric matching)}
\end{alignat*}
\caption{Base Qunity syntax.}
\Description{}
\label{fig:syntax}
\end{minipage}
\end{figure}

\Cref{fig:syntax} shows the base abstract syntax of the Qunity language. Expressions $e$ are assigned types $T$ and programs $f$ are assigned program types $F$ (\Cref{fig:types}). Here, $x$ ranges over some set of variable names and $r$ ranges over a representation of real numbers. This work adds two new primitive syntactic constructs to the original Qunity language: $\lstinline|match|$ and $\lstinline|pmatch|$, discussed in \Cref{sec:extending}.

\subsection{Type System}\label{app:type_system}

\begin{figure}[ht]
\[
	\inference{}{\Gamma\partition \varnothing \vdash \Unit : \Unit}[\textsc{T-Unit}]
	\quad
	\inference{}{\Gamma, x : T, \Gamma'\partition \varnothing \vdash x : T}[\textsc{T-Cvar}]
	\quad
	\inference{x \notin \dom(\Gamma)}{\Gamma\partition x : T \vdash x : T}[\textsc{T-Qvar}]
\]
\vspace{2mm}
\[
	\inference{\Gamma\partition \Delta, \Delta_0 \vdash e_0 : T_0 \qquad \Gamma\partition \Delta, \Delta_1 \vdash e_1 : T_1}{\Gamma\partition \Delta, \Delta_0, \Delta_1 \vdash \pair {e_0} {e_1} : T_0 \otimes T_1}[\textsc{T-PurePair}]
    \quad
    \inference{\vdash f : T \rightsquigarrow T' \qquad \Gamma\partition \Delta \vdash e : T}{\Gamma\partition \Delta \vdash f\: e : T'}[\textsc{T-PureApp}]
\]
\vspace{2mm}
\[
	\inference{\Gamma \partition \Delta \Vdash e : T \qquad \ortho{T}{e_1, \ldots, e_n} \qquad \varnothing\partition \Gamma_j \vdash e_j : T \; \forall j \qquad \classical(e_j) \; \forall j \\ \erases{T'}(x; e_1', \ldots, e_n') \; \forall x \in \dom(\Delta) \qquad \Gamma, \Gamma_j\partition \Delta, \Delta' \vdash e_j' : T' \; \forall j}{\Gamma \partition \Delta, \Delta' \vdash \cntrl{e}{T}{e_1 &\mapsto e_1' \\ &\cdots \\ e_n &\mapsto e_n'}{T'} : T'}[\textsc{T-Ctrl}]
\]
\caption{Pure expression typing rules. Here, $\texttt{classical}(e)$ holds if $e$ does not include any use of \texttt{u3} or \texttt{rphase}.}
\Description{}
\label{fig:t-pure-exp}
\end{figure}

Consider an expression $e$ typed as $\Gamma \partition \Delta \vdash e : T$, that is, $e$ has pure type $T$ with classical context $\Gamma$ and quantum context $\Delta$. The semantics of such an expression corresponds to a \emph{pure state} in the Hilbert space $\cH(T)$ associated with the type $T$. \Cref{fig:t-pure-exp} shows the typing rules for the pure expression typing judgment. While the separation into ``classical'' and ``quantum'' contexts may seem unexpected for a language like Qunity that aims to unify classical and quantum computation, the classical contexts do in fact play an important role in the type system. When a variable $x$ is in a classical context $\Gamma$, it does not mean that it is literally in a classical register - it means that from the perspective of the current expression, we view it as being in a classical basis state and interact with it by sharing it relative to the classical basis. This also means that the type system does not place the same relevance restrictions on variables in classical contexts as it does on quantum variables: we may assume that the compiler automatically uncomputes the variable when it goes out of scope. The primary use case for these classical contexts is the \lstinline|ctrl| construct: as seen in the \textsc{T-Ctrl} rule, the quantum contexts of the left-hand-side (LHS) pattern expressions become classical contexts when typing the right-hand-side (RHS) expressions. This allows the programmer to safely ignore these variables in the RHS expressions, knowing that when compiling \textsc{T-Ctrl}, their uncomputation will be handled automatically.

\begin{figure}[ht]
\[
	\inference{\Gamma \partition \Delta \vdash e : T}{\Gamma \partition \Delta \Vdash e : T}[\textsc{T-Mix}]
	\quad
\inference{\Gamma \partition \Delta \Vdash e : T}{\Gamma \partition \Delta, \Delta_0 \Vdash e : T}[\textsc{T-Discard}]
	\]
\vspace{2mm}
	\[
	\inference{\Gamma \partition \Delta, \Delta_0 \Vdash e_0 : T_0 \qquad \Gamma \partition \Delta, \Delta_1 \Vdash e_1 : T_1}{\Gamma \partition \Delta, \Delta_0, \Delta_1 \Vdash \pair{e_0}{e_1} : T_0 \otimes T_1}[\textsc{T-MixedPair}]
\]
\vspace{2mm}
\[
	\inference{\Gamma \partition \Delta_0 \Vdash e_0 : T \qquad \Gamma \partition \Delta_1 \Vdash e_1 : T}{\Gamma \partition \Delta_0, \Delta_1 \Vdash \trycatch{e_0}{e_1} : T}[\textsc{T-Try}]
	\quad
	\inference{\vdash f : T \Rrightarrow T' \qquad \Gamma \partition \Delta \Vdash e : T}{\Gamma \partition \Delta \Vdash f\: e : T'}[\textsc{T-MixedApp}]
\]
\vspace{2mm}
\[
\inference{\Gamma \partition \Delta, \Delta_0 \Vdash e : T \qquad \ortho{T}{e_1, \dots, e_n} \qquad \varnothing \partition \Gamma_j \vdash e_j : T \; \forall j \\
\classical(e_j) \; \forall j \qquad \Gamma, \Gamma_j \partition \Delta, \Delta_1 \Vdash e_j' : T' \; \forall j}
{\Gamma \partition \Delta, \Delta_0, \Delta_1 \Vdash \match{e}{T}{e_1 &\mapsto e_1' \\ &\cdots \\ e_n &\mapsto e_n'}{T'} : T'}[\textsc{T-Match}]
\]
\caption{Mixed expression typing rules.}
\Description{}
\label{fig:t-mixed-exp}
\end{figure}

The semantics of an expression typed with the mixed typing judgment as $\Gamma \partition \Delta \Vdash e : T$ corresponds to a \emph{mixed state} in the Hilbert space $\cH(T)$ associated with the type $T$. Observe that the rule \textsc{T-Mix} allows pure expressions to be typed as mixed.

Note that this presentation differs from the original mixed expression typing rules in \citeauthor{Voichick_2023}, which did not include classical contexts since the mixed expression typing judgment has no relevance restrictions on the quantum variables. The primary reason for this change is consistency between \lstinline|ctrl| and the newly introduced \lstinline|match|: for instance, it allows expressions of the following form to be accepted by the typechecker:
\[
\match{x}{T}{y &\mapsto \cntrl{y}{T}{\text{(expressions that do not erase $y$)}}{T'} \\ &\qquad\qquad\qquad\cdots}{T'}
\]
Here, $y$ is passed into the \emph{classical} context of the \lstinline|ctrl|, which is initially typed as mixed by the \textsc{T-Match} rule, and then as pure through \textsc{T-Mix} rule. Because $y$ is in the classical context of the \lstinline|ctrl| expression, it is not required to satisfy the erasure judgment in \lstinline|ctrl| since it only applies to quantum context variables.

The \textsc{T-Discard} rule is also a new addition, as \citeauthor{Voichick_2023} used the \textsc{T-MixedAbs} program typing rule to make variable discarding possible. This new change is made to more easily deal with discarded variables in the \lstinline|match| construct.

These rules rely on several additional judgments, namely \texttt{ortho}, and \texttt{erases}. These are listed in \Cref{app:erasure_judgment} and \Cref{app:ortho_judgment}.

\begin{figure}[ht]
\[
	\inference{}{\vdash \uthree{r_\theta}{r_\phi}{r_\lambda} : \Bit \rightsquigarrow \Bit}[\textsc{T-Gate}]
\]
\vspace{2mm}
\[
	\inference{}{\vdash \lef{T_0}{T_1} : T_0 \rightsquigarrow T_0 \oplus T_1}[\textsc{T-Left}]
	\quad
	\inference{}{\vdash \rit{T_0}{T_1} : T_1 \rightsquigarrow T_0 \oplus T_1}[\textsc{T-Right}]
\]
\vspace{2mm}
\[
	\inference{\varnothing\partition \Delta \vdash e : T \qquad \varnothing\partition \Delta \vdash e' : T'}{\vdash \lambda e \xmapsto{{\color{gray}T}} e' : T \rightsquigarrow T'}[\textsc{T-PureAbs}]
	\quad
	\inference{\varnothing \partition \Delta \vdash e : T}{\vdash \rphase{T}{e}{r}{r'} : T \rightsquigarrow T}[\textsc{T-Rphase}]
\]
\vspace{2mm}
\[
\inference{\ortho{T}{e_1, \dots, e_n} & \qquad \varnothing \partition \Delta_j \vdash e_j : T \;\forall j \\ \ortho{T'}{e_1', \dots, e_n'} & \qquad \varnothing \partition \Delta_j \vdash e_j' : T' \;\forall j}{\vdash \pmatch{T}{e_1 &\mapsto e_1' \\ &\cdots \\ e_n &\mapsto e_n'}{T'} : T \rightsquigarrow T'}[\textsc{T-Pmatch}]
\]
\vspace{2mm}
\[
	\inference{\vdash f : T \rightsquigarrow T'}{\vdash f : T \Rrightarrow T'}[\textsc{T-Channel}]
	\quad
	\inference{\varnothing\partition \Delta \vdash e : T \qquad \varnothing \partition \Delta \Vdash e' : T'}{\vdash \lambda e \xmapsto{{\color{gray}T}} e' : T \Rrightarrow T'}[\textsc{T-MixedAbs}]
\]
\caption{Program typing rules.}
\Description{}
\label{fig:t-prog}
\label{fig:t-pure-prog}
\label{fig:t-mixed-prog}
\end{figure}

\subsection{Semantics}\label{app:semantics}

Qunity's semantics is defined in terms of operators and superoperators mapping between Hilbert spaces associated with Qunity types and contexts. These Hilbert spaces are defined as follows:

\begin{align*}
    \cH(\Void) &\defeq \{ 0 \} \\
    \cH(\Unit) &\defeq \CC \\
    \cH(T_0 \oplus T_1) &\defeq \cH(T_0) \oplus \cH(T_1) \\
    \cH(T_0 \otimes T_1) &\defeq \cH(T_0) \otimes \cH(T_1) \\
    \cH(x_1 : T_1, \ldots, x_n : T_n) &\defeq \cH(T_1) \otimes \cdots \otimes \cH(T_n)
\end{align*}

For a type $T$, we define the set of classical expressions of that type as follows:
\begin{align*}
    \VV(\Void) &\defeq \varnothing \\
    \VV(\Unit) &\defeq \{ \unit \} \\
    \VV(T_0 \oplus T_1) &\defeq \{ \lef{T_0}{T_1}\: v_0 \mid v_0 \in \VV(T_0)\} \cup \{ \rit{T_0}{T_1}\: v_1 \mid v_1 \in \VV(T_1)\} \\
    \VV(T_0 \otimes T_1) &\defeq \{ \pair{v_0}{v_1} \mid v_0 \in \VV(T_0), v_1 \in \VV(T_1) \}
\end{align*}

For a context $\Delta$, we have a set of \emph{valuations}, defined as
\[
    \VV(x_1 : T_1, \ldots, x_n : T_n) \defeq \{ x_1 \mapsto v_1, \ldots, x_n \mapsto v_n \mid v_1 \in \VV(T_1), \ldots, v_n \in \VV(T_n) \}
\]

The space $\cH(T)$ is spanned by an orthonormal basis $\{\ket{v} : v \in \VV(T)\}$, where we define
\begin{align*}
    \ket{\unit} &\defeq 1 \\
    \ket{\lef{T_0}{T_1}\: v_0} &\defeq \ket{v_0} \oplus 0 \\
    \ket{\rit{T_0}{T_1}\: v_1} &\defeq 0 \oplus \ket{v_1} \\
    \ket{\pair{v_0}{v_1}} &\defeq \ket{v_0} \otimes \ket{v_1}
\end{align*}

The basis states for the space $\cH(\Delta)$ are defined by $\ket{\tau}$ for valuations $\tau \in \VV(\Delta)$.
\begin{alignat*}{2}
    \ket{x_1 \mapsto v_1, \ldots, x_n \mapsto v_n} &\defeq \ket{v_1} \otimes \cdots \otimes \ket{v_n}
\end{alignat*}

Now, we can define the denotational semantics of Qunity. For pure expression semantics, we say that if $\Gamma \partition \Delta \vdash e : T$ and $\sigma \in \VV(\Gamma)$, then $\msem{\sigma : \Gamma \partition \Delta \vdash e : T} \in \cL(\cH(\Delta), \cH(T))$ defines the pure semantics of $e$. Here $\sigma$ is a valuation of $\Gamma$, representing ``classical data''. For mixed expression semantics, we have that if $\Gamma \partition \Delta \Vdash e : T$ and $\sigma \in \VV(\Gamma)$, then $\msem{\sigma : \Gamma \partition \Delta \Vdash e : T} \in \cL(\cL(\cH(\Delta)), \cL(\cH(T)))$, defining a superoperator acting on the space of density matrices. For program semantics, we have that $\msem{\vdash f : T \rightsquigarrow T'} \in \cL(\cH(T), \cH(T'))$, and $\msem{\vdash f : T \Rrightarrow T'} \in \cL(\cL(\cH(T)), \cL(\cH(T')))$.

\begin{figure}[ht]
\begin{alignat*}{3}
	\msem{\sigma : \Gamma \partition \varnothing \vdash \unit : \Unit} \ket{\varnothing} &\defeq&\;& \ket{\unit} \\
	\msem{\sigma : \Gamma \partition \varnothing \vdash x : T} \ket{\varnothing} &\defeq&& \ket{\sigma(x)} \\
	\msem{\sigma : \Gamma \partition x : T \vdash x : T} \ket{x \mapsto v} &\defeq&& \ket{v} \\
	\msem{\sigma : \Gamma \partition \Delta, \Delta_0, \Delta_1 \vdash \pair{e_0}{e_1} : T_0 \otimes T_1} \ket{\tau, \tau_0, \tau_1} &\defeq&& \msem{\sigma : \Gamma \partition \Delta, \Delta_0 \vdash e_0 : T_0} \ket{\tau, \tau_0} \\
		&&&\otimes \msem{\sigma : \Gamma \partition \Delta, \Delta_1 \vdash e_1 : T_1} \ket{\tau, \tau_1} \\
	\msem{\sigma : \Gamma \partition \Delta, \Delta' \vdash \cntrl{e\hspace{-2mm}}{T}{e_1 &\mapsto e_1' \\ &\cdots \\ e_n &\mapsto e_n'}{T'} \hspace{-2mm}: T'} \ket{\tau,\tau'} &\defeq&& \sum_{v \in \VV(T)} \bra{v} \msem{\Gamma \partition \Delta \Vdash e : T}\left( \op{\sigma, \tau}{\sigma, \tau} \right) \ket v \\
         &&& \cdot \sum_{j=1}^n \sum_{\sigma_j \in \VV(\Gamma_j)} \bra{\sigma_j} \msem{\varnothing : \varnothing \partition \Gamma_j \vdash e_j : T}^\dagger \ket v \\
         &&&\cdot \msem{\sigma, \sigma_j : \Gamma, \Gamma_j \partition \Delta,\Delta' \vdash e_j' : T'} \ket{\tau, \tau'} \\
	\msem{\sigma : \Gamma \partition \Delta \vdash f\; e : T'} \ket{\tau} &\defeq&& \msem{\vdash f : T \rightsquigarrow T'} \msem{\sigma : \Gamma \partition \Delta \vdash e : T} \ket{\tau}
\end{alignat*}
\caption{Pure expression semantics.}
\Description{}
\label{fig:sem-pure-expr}
\end{figure}

\begin{figure}[ht]
\begin{align*}
	\msem{\vdash \uthree{r_\theta}{r_\phi}{r_\lambda} : \Bit \rightsquigarrow \Bit} \ket{\zero} &\defeq \cos(r_\theta / 2) \ket{\zero} + e^{i r_\phi} \sin(r_\theta / 2) \ket{\one} \\
	\msem{\vdash \uthree{r_\theta}{r_\phi}{r_\lambda} : \Bit \rightsquigarrow \Bit} \ket{\one} &\defeq - e^{i r_\lambda} \sin(r_\theta / 2) \ket{\zero}  + e^{i(r_\phi + r_\lambda)} \cos(r_\theta / 2) \ket{\one} \\
	\msem{\vdash \lef{T_0}{T_1} : T_0 \rightsquigarrow T_0 \oplus T_1}\ket{v} &\defeq \ket {\lef{T_0}{T_1}\: v} \\
	\msem{\vdash \rit{T_0}{T_1} : T_1 \rightsquigarrow T_0 \oplus T_1}\ket{v} &\defeq \ket {\rit{T_0}{T_1}\: v} \\
	\msem{\vdash \lambda e \xmapsto{T} e' : T \rightsquigarrow T'}\ket{v} &\defeq \msem{\varnothing : \varnothing \partition \Delta \vdash e' : T'} \msem{\varnothing : \varnothing \partition \Delta \vdash e : T}^\dagger \ket{v} \\
	\msem{\vdash \rphase{T}{e}{r}{r'} : T \rightsquigarrow T} \ket{v} &\defeq e^{i r} \msem{\varnothing : \varnothing \partition \Delta \vdash e : T} \msem{\varnothing : \varnothing \partition \Delta \vdash e : T}^\dagger \ket{v}
\\[-0.5em]&\qquad\qquad + e^{i r'} \left(\mathbb{I} - \msem{\varnothing \partition \Delta \vdash e : T} \msem{\varnothing : \varnothing \partition \Delta \vdash e : T}^\dagger\right) \ket{v} \\
\msem{\vdash \pmatch{T}{e_1 &\mapsto e_1' \\ &\cdots \\ e_n &\mapsto e_n'}{T'} : T \rightsquigarrow T'}\ket{v} &\defeq
\sum_{j=1}^n \msem{\varnothing : \varnothing \partition \Delta_j \vdash e_j' : T'}
\msem{\varnothing : \varnothing \partition \Delta_j \vdash e_j : T}\adj \ket{v}
\end{align*}
\caption{Pure program semantics.}
\Description{}
\label{fig:sem-pure-prog}
\end{figure}

\begin{figure}[ht]
\begin{alignat*}{2}
    \msem{\sigma : \Gamma \partition \Delta \Vdash e : T}\left(\op{\tau}{\tau'}\right) &\defeq&\;& \msem{\sigma : \Gamma \partition \Delta \vdash e : T} \op{\tau}{\tau'} \msem{\sigma : \Gamma \partition \Delta \vdash e : T}^\dagger \\
	\msem{\sigma : \Gamma \partition \Delta, \Delta_{0} \Vdash e : T}\left(\rho \otimes \rho_0\right) &\defeq&\;& \tr(\rho_0) \msem{\sigma : \Gamma \partition \Delta \Vdash e : T}(\rho) \\
	\msem{\sigma : \Gamma \partition \Delta, \Delta_0, \Delta_1 \Vdash \pair{e_0}{e_1} : T_0 \otimes T_1}\left(\op{\tau, \tau_0, \tau_1}{\tau', \tau_0', \tau_1'}\right) &\defeq&& \msem{\sigma : \Gamma \partition \Delta, \Delta_0 \Vdash e_0 : T_0}\left(\op{\tau, \tau_0}{\tau', \tau_0'}\right) \\
     &&&\otimes \msem{\sigma : \Gamma \partition \Delta, \Delta_1 \Vdash e_1 : T_1}\left(\op{\tau, \tau_1}{\tau', \tau_1'}\right) \\
\msem{\sigma : \Gamma \partition \Delta_0,\Delta_1 \Vdash \trycatch{e_0}{e_1} : T}\left( \rho_0 \otimes \rho_1 \right) &\defeq&& \tr(\rho_1) \msem{\sigma : \Gamma \partition \Delta_0 \Vdash e_0 : T}\left( \rho_0 \right) \\&&&+ (\tr(\rho_0) - \tr(\msem{\sigma : \Gamma \partition \Delta_0 \Vdash e_0 : T}\left( \rho_0 \right))) \\&&&\cdot \msem{\sigma : \Gamma \partition \Delta_1 \Vdash e_1 : T}\left( \rho_1 \right)
\end{alignat*}
\begin{align*}
\msem{\sigma : \Gamma \partition \Delta, \Delta_0, \Delta_1 \Vdash \match{e}{T}{e_1 &\mapsto e_1' \\ &\cdots \\ e_n &\mapsto e_n'}{T'} : T'} \parens{\ket{\tau, \tau_0, \tau_1}\bra{\tau', \tau_0', \tau_1'}} \defeq \\
= \sum_{v \in \VV(T)} \bra{v} \parens{\msem{\sigma : \Gamma \partition \Delta, \Delta_0 \Vdash e : T}\parens{\ket{\tau, \tau_0}\bra{\tau', \tau_0'}}} \ket{v} \cdot \\
\cdot \sum_{j=1}^n \sum_{\sigma_j \in \VV(\Gamma_j)}
\bra{\sigma_j} \msem{\varnothing : \varnothing \partition \Gamma_j \vdash e_j : T}\adj \ket{v} \cdot \\
\cdot \msem{\sigma, \sigma_j : \Gamma, \Gamma_j \partition \Delta,\Delta_1 \Vdash e_j' : T'}\parens{\ket{\tau, \tau_1}\bra{\tau', \tau_1'}}
\end{align*}
\caption{Mixed expression semantics.}
\Description{}
\label{fig:sem-mixed-expr}
\end{figure}

\begin{figure}[ht]
\begin{align*}
	& \msem{\vdash f : T \Rrightarrow T'}(\op{v}{v'}) \defeq \msem{\vdash f : T \rightsquigarrow T'}\op{v}{v'} \msem{\vdash f : T \rightsquigarrow T'}^\dagger \\
	& \msem{\vdash \lambda e \xmapsto{{\color{gray}T}} e' : T \Rrightarrow T'}(\op{v}{v'}) \defeq \\
	&\qquad \msem{\varnothing : \varnothing \partition \Delta \Vdash e' : T'}\left(\msem{\varnothing : \varnothing \partition \Delta \vdash e : T}^\dagger \op{v}{v'} \msem{\varnothing : \varnothing \partition \Delta \vdash e : T}\right)
\end{align*}
\caption{Mixed program semantics.}
\Description{}
\label{fig:sem-mixed-prog}
\end{figure}

\clearpage

\section{Grammar of the Surface Syntax}\label{app:surface_grammar}

A type name \synt{tname} consists of a capital letter followed by zero or more alphanumeric characters, underscores, or apostrophes. An expression name \synt{ename} consists of a dollar sign \lit{\$} followed by one or more alphanumeric characters, underscores, or apostrophes. Program names \synt{fname}, real names \synt{rname}, and type variables \synt{tvar} are similar, using \lit{@}, \lit{\#}, and \lit{'} respectively. Quantum variables \synt{qvar} must start with a lowercase letter or underscore. Real constants \synt{const} are integers (one or more digits possibly preceded by a minus sign).

\vspace{\baselineskip}

\grammarindent1cm
\begin{grammar}

<qfile> ::= <def>* <e>

<def> ::= `type' <tname> <sig> `:=' <t> `end'
\alt `type' <tname> <sig> `:='  (`|' <tname> | <fname> `of' <t>)+ `end'
\alt `def' <ename> <sig> `:' <t> `:=' <e> `end'
\alt `def' <fname> <sig> `:' <t> `->' <t> `:=' <f> `end'
\alt `def' <rname> <sig> `:=' <r> `end'

<sig> ::= `{' <param> (`,' <param>)* `}'

<param> ::= <tvar>
\alt <tname> `:' <t>
\alt <fname> `:' <t> `->' <t>
\alt <rname>

<ge> ::= <t> | <e> | <f> | <r>

<t> ::= `Void'
\alt `Unit'
\alt <t> `*' <t>
\alt <tvar>
\alt <tname> (`{' <ge>, (`,' <ge>)* `}')?
\alt `(' <t> `)'
\alt `if' <be> `then' <t> `else' <t> `endif'

<e> ::= `()'
\alt <qvar>
\alt `(' <e> `,' <e> `)'
\alt `ctrl' <e> `[' (<e> `->' <e> `;')* (`else' -> <e>)? `]'
\alt `match' <e> `[' (<e> `->' <e> `;')* (`else' -> <e>)? `]'
\alt `try' <e> `catch' <e>
\alt <f> `(' <e> `)'
\alt <e> `|>' <f>
\alt `let' <e> `=' <e> `in' <e>
\alt <ename> (`{' <ge>, (`,' <ge>)* `}')?
\alt `(' <e> `)'
\alt `if' <be> `then' <e> `else' <e> `endif'

<f> ::= `u3' `{' <r> `,' <r> `,' <r> `}'
\alt `lambda' <e> `->' <e>
\alt `gphase' `{' <r> `}'
\alt `rphase' `{' <e> `,' <r> `,' <r> `}'
\alt `pmatch' `[' (<e> `->' <e> `;')* `]'
\alt <fname> (`{' <ge>, (`,' <ge>)* `}')?
\alt `(' <f> `)'
\alt `if' <be> `then' <f> `else' <f> `endif'

<r> ::= `pi'
\alt `euler'
\alt <const>
\alt <r> `+' <r>
\alt <r> `-' <r>
\alt <r> `*' <r>
\alt <r> `/' <r>
\alt <r> `^' <r>
\alt <r> `\%' <r>
\alt `sin' `(' <r> `)'
\alt `cos' `(' <r> `)'
\alt `tan' `(' <r> `)'
\alt `arcsin' `(' <r> `)'
\alt `arccos' `(' <r> `)'
\alt `arctan' `(' <r> `)'
\alt `exp' `(' <r> `)'
\alt `ln' `(' <r> `)'
\alt `log2' `(' <r> `)'
\alt `sqrt' `(' <r> `)'
\alt `ceil' `(' <r> `)'
\alt `floor' `(' <r> `)'
\alt <rname> (`{' <ge>, (`,' <ge>)* `}')?
\alt `(' <r> `)'
\alt `if' <be> `then' <r> `else' <r> `endif'

<be> ::= `!' <be>
\alt <be> `&&' <be>
\alt <be> `||' <be>
\alt <r> <cmp> <r>
\alt `(' <be> `)'

<cmp> ::= `=' | `!=' | `<=' | `<' | `>=' | `>'
\end{grammar}

\newpage

\section{Comparison of Control Flow Constructs}\label{app:control_flow_comparison}

\begin{longtable}[]{@{}
  | >{\raggedright\arraybackslash}p{(\columnwidth - 6\tabcolsep) * \real{0.16}}
  | >{\raggedright\arraybackslash}p{(\columnwidth - 6\tabcolsep) * \real{0.28}}
  | >{\raggedright\arraybackslash}p{(\columnwidth - 6\tabcolsep) * \real{0.28}}
  | >{\raggedright\arraybackslash}p{(\columnwidth - 6\tabcolsep) * \real{0.28}}@{} |}
\hline
& \lstinline|ctrl| & \lstinline|match| & \lstinline|pmatch| \\
\hline
What kind of semantics? & Pure expression & Mixed expression & Pure
program \\
\hline
Restrictions & \begin{minipage}[t]{\linewidth}\raggedright
\begin{itemize}
\item
  LHS must satisfy orthogonality
\item
  LHS must be pure
\item
  RHS must be pure
\item
  RHS must satisfy erasure judgment for scrutinee quantum context
\item
  LHS must be classical
\end{itemize}
\end{minipage} & \begin{minipage}[t]{\linewidth}\raggedright
\begin{itemize}
\item
  LHS must satisfy orthogonality
\item
  LHS must be pure
\item
  LHS must be classical
\end{itemize}
\end{minipage} & \begin{minipage}[t]{\linewidth}\raggedright
\begin{itemize}
\item
  LHS must satisfy orthogonality
\item
  RHS must satisfy orthogonality
\item
  LHS must be pure
\item
  RHS must be pure
\item
  Each pair of matching expressions must share the same (quantum)
  context
\item
  Since it is a program, it is typed without context and
  does not contain a scrutinee directly
\end{itemize}
\end{minipage} \\
\hline
Benefits & \begin{minipage}[t]{\linewidth}\raggedright
\begin{itemize}
\item
  Scrutinee can be mixed while the entire expression is still pure
\item
  Can use classical context variables originating from an outer
  \lstinline|ctrl| or \lstinline|match| in an unrestricted way
\item
  Can be used to apply relative phases (conditional global phases)
\end{itemize}
\end{minipage} & \begin{minipage}[t]{\linewidth}\raggedright
\begin{itemize}
\item
  RHS can be mixed
\item
  Erasure judgment is not required
\item
  Can use classical context variables originating from an outer
  \lstinline|ctrl| or \lstinline|match| in an unrestricted way
\item
  No ``unavoidable source of error'' when scrutinee is not classical
\end{itemize}
\end{minipage} & \begin{minipage}[t]{\linewidth}\raggedright
\begin{itemize}
\item
  LHS and RHS can be in an arbitrary basis
\item
  Erasure judgment is not required
\item
  Can be used to apply relative phases (conditional global phases)
\end{itemize}
\end{minipage} \\
\hline
Can be used in a pattern / has an adjoint & Yes & No & Yes \\
\hline
Can be used in the LHS of a \lstinline|ctrl| or \lstinline|match| & No & No &
No \\
\hline
Can be used in the RHS of a \lstinline|ctrl| & Yes & No & Yes \\
\hline
Can be used the RHS of a \lstinline|match| & Yes & Yes & Yes \\
\hline
Can be used in the LHS or RHS of a \lstinline|pmatch| & Yes, if the isometry
judgment holds & No & Yes, if the isometry judgment holds \\
\hline
Can contain a mixed expression in the scrutinee & Yes & Yes & N/A \\
\hline
Can contain a mixed expression on the LHS & No & No & No \\
\hline
Can contain a mixed expression on the RHS & No & Yes & No \\
\hline
Applying a global phase to a RHS expression makes a semantic difference
& Yes & No & Yes \\
\hline
When to use? & \begin{minipage}[t]{\linewidth}\raggedright
\begin{itemize}
\item
  You want to control something on a subroutine which may involve
  measurement/decoherence, but you want your entire expression to be
  reversible
\item
  You want to conditionally apply a global phase depending on the result
  of a mixed expression
\item
  You are able to keep the scrutinee variables on the RHS in
  satisfaction of the erasure judgment
\end{itemize}
\end{minipage} & \begin{minipage}[t]{\linewidth}\raggedright
\begin{itemize}
\item
  You do not care whether your expression is reversible
\item
  You want to conditionally apply some expressions that involve
  measurement/decoherence
\item
  You want something that most closely corresponds to a \lstinline|match|
  in classical programming languages
\item
  What you want to do can be accomplished by measuring the outcome of
  the input expression and then evaluating the appropriate RHS
  expression
\item
  You want to avoid the erasure requirement
\end{itemize}
\end{minipage} & \begin{minipage}[t]{\linewidth}\raggedright
\begin{itemize}
\item
  You want to reversibly map between two different bases
\item
  The LHS patterns are not in the classical basis.
\item
  You want to avoid the erasure requirement
\item
  You do not need any outside variables
\end{itemize}
\end{minipage} \\
\hline
\end{longtable}

\newpage

\section{Benchmark Test Programs}\label{app:benchmark_test_programs}

For these examples, the quantum registers in all circuits were initialized to a superposition state to prevent classical propagation optimizations from taking place. Not all of the examples are useful quantum algorithms, but they serve to evaluate the performance of the Qunity compiler.

\subsection{Phase Conditioned on AND of 5 Qubits}\label{app:multi_and_impl}

Qunity implementation:

\begin{lstlisting}
def @multi_and{#n} : Array{#n, Bit} -> Bit :=
  if #n = 0 then
    lambda () -> &1
  else
    lambda (x0, x1) -> @and(x0, @multi_and{#n - 1}(x1))
  endif
end

&repeated{5, Bit, &plus} |> lambda x -> ctrl @multi_and{5}(x) [
  &0 -> x;
  &1 -> x |> gphase{pi}
]
\end{lstlisting}

\begin{figure}[ht]
    \centering
    \includegraphics[width=0.8\linewidth]{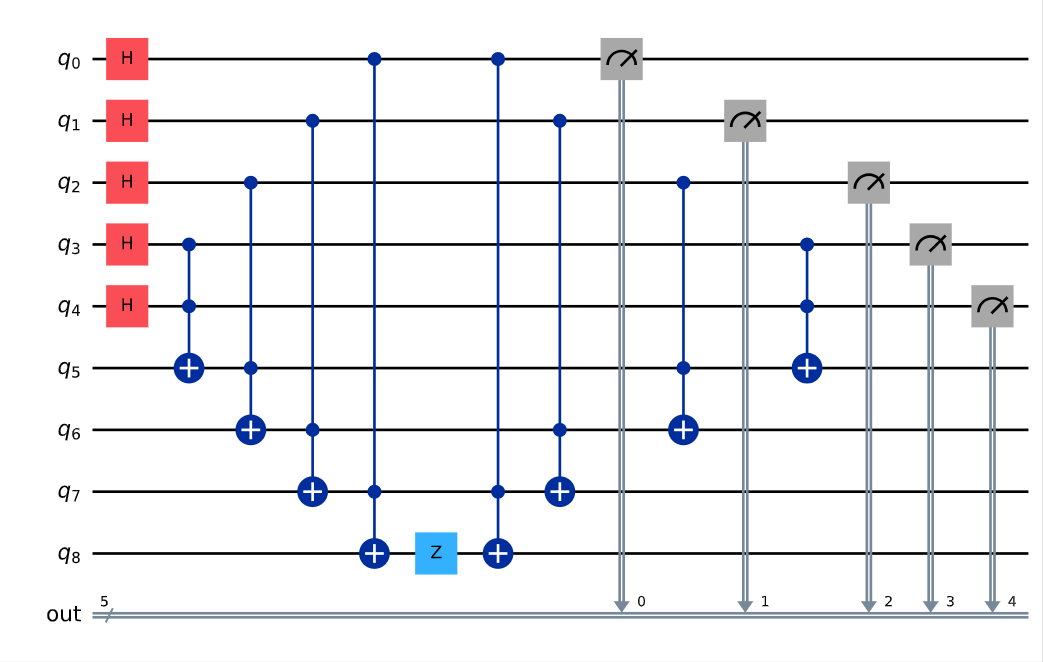}
    \caption{Compiled circuit for phase conditioned on AND of 5 qubits.}
    \Description{}
    \label{fig:multi_and_phase_circuit}
\end{figure}

Qiskit implementation:

\begin{lstlisting}[language=python]
circuit = QuantumCircuit(5, 5)
circuit.h(range(5))
circuit.mcp(math.pi, [0, 1, 2, 3], [4])
circuit.measure(range(5), range(5))
\end{lstlisting}

\subsection{Quantum Fourier Transform (5 Qubits)}\label{app:qft_impl}

Qunity implementation:

\begin{lstlisting}
def @couple{#k} : Bit * Bit -> Bit * Bit :=
  lambda (x0, x1) -> (x1, x0) |> rphase{(&1, &1), 2 * pi / (2 ^ #k), 0}
end

def @rotations{#n} : Array{#n, Bit} -> Array{#n, Bit} :=
  if #n <= 0 then
    @id{Unit}
  else if #n = 1 then
    lambda (x, ()) -> (@had(x), ())
  else
    lambda (x0, x) ->
      let (x0, (y0', y)) = (x0, @rotations{#n - 1}(x)) in
      let ((y0, y1), y) = (@couple{#n}(x0, y0'), y) in
      (y0, (y1, y))
  endif
  endif
end

def @qft{#n} : Array{#n, Bit} -> Array{#n, Bit} :=
  if #n <= 0 then
    @id{Unit}
  else
    lambda x ->
      let (x0, x') = @rotations{#n}(x) in
      (x0, @qft{#n - 1}(x'))
  endif
end

@qft{5}(&0, (&plus, (&plus, (&0, (&0, ()))))) |> @reverse{5, Bit}
\end{lstlisting}

\begin{figure}[ht]
    \centering
    \includegraphics[width=\linewidth]{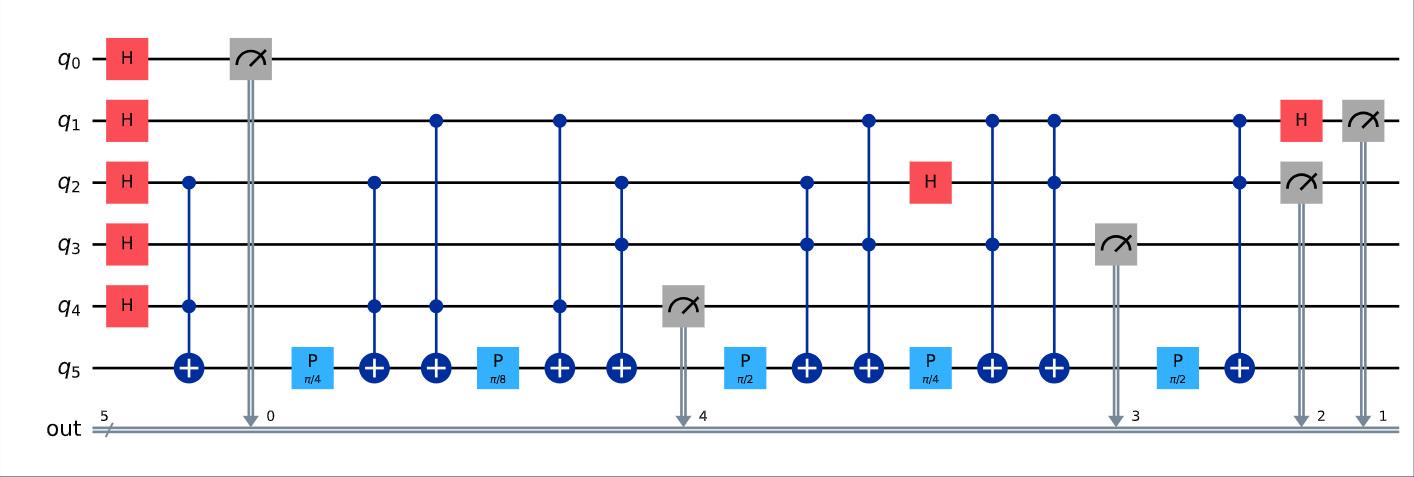}
    \caption{Compiled circuit for the Quantum Fourier Transform.}
    \Description{}
    \label{fig:fourier_transform_circuit}
\end{figure}

Qiskit implementation:

\begin{lstlisting}[language=python]
from qiskit.circuit.library import QFT
circuit = QuantumCircuit(5, 5)
circuit.h([1, 2])
circuit.append(QFT(5), [0, 1, 2, 3, 4])
circuit.measure([0, 1, 2, 3, 4], [4, 3, 2, 1, 0])
\end{lstlisting}

\newpage

\subsection{Phase Estimation Example (5 Qubits)}\label{app:phase_estimation_impl}

Qunity implementation:

\begin{lstlisting}
def @apply_phase{#n, #p} : Num{#n} -> Num{#n} :=
  if #n = 0 then
    @id{Unit}
  else
    lambda (x0, x') ->
    (ctrl x0 [
      &0 -> x0;
      &1 -> x0 |> gphase{2 * pi * #p}
    ], @apply_phase{#n - 1, 2 * #p}(x'))
  endif
end

def &phase_estimation{#n, #p} : Num{#n} :=
  &repeated{#n, Bit, &plus}
  |> @apply_phase{#n, #p}
  |> @adjoint{Num{#n}, Num{#n}, @qft{#n}}
  |> @reverse{#n, Bit}
end

&phase_estimation{5, 1/3}
\end{lstlisting}

\begin{figure}[ht]
    \centering
    \includegraphics[width=\linewidth]{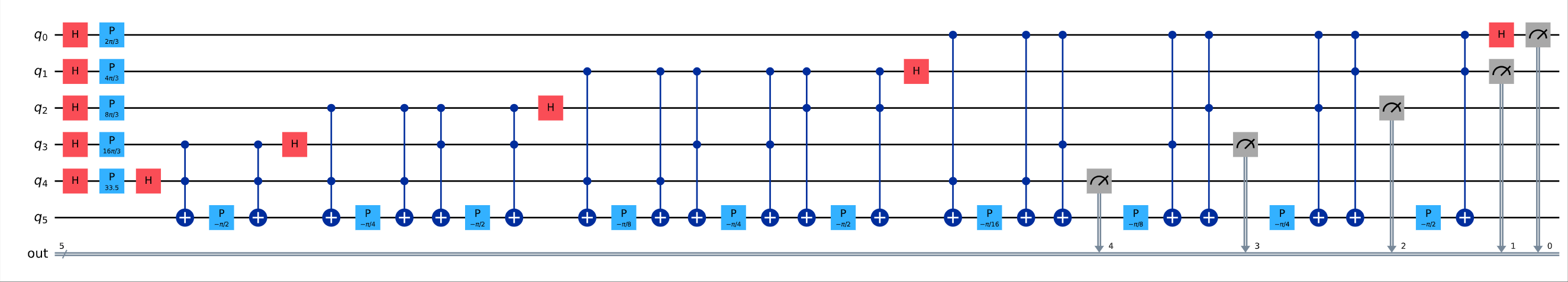}
    \caption{Compiled circuit for phase estimation.}
    \Description{}
    \label{fig:phase_estimation_circuit}
\end{figure}

Qiskit implementation:

\begin{lstlisting}[language=python]
from qiskit.circuit.library import PhaseEstimation, GlobalPhaseGate

circuit = QuantumCircuit(5, 5)
circuit.append(PhaseEstimation(5, GlobalPhaseGate(2 * math.pi / 3)), range(5))
circuit.measure(range(5), range(5))
\end{lstlisting}

\subsection{Order Finding}\label{app:order_finding_impl}

For the Qunity implementation, see \Cref{sec:examples_order_finding}.

\begin{figure}[ht]
    \centering
    \includegraphics[width=\linewidth]{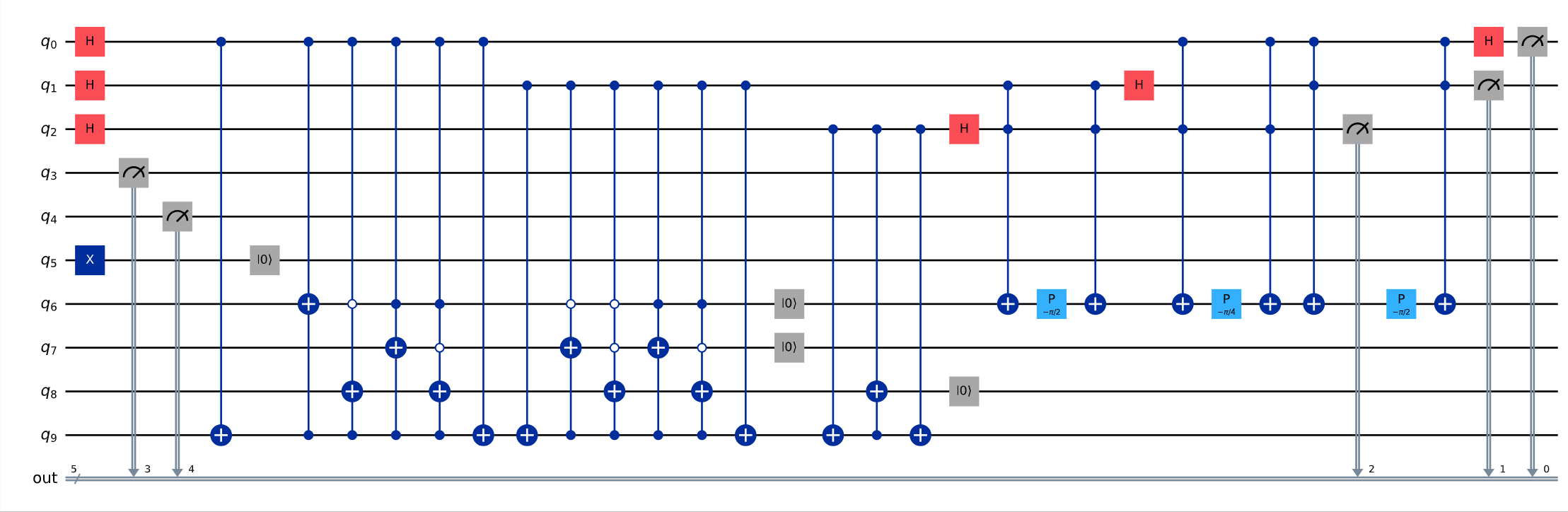}
    \caption{Compiled circuit for order finding.}
    \Description{}
    \label{fig:order_finding_circuit}
\end{figure}

\subsection{Reversible CDKM Adder (5 Bits)}\label{app:adder_impl}

Qunity implementation:

\begin{lstlisting}
def @cdkm_maj : Bit * Bit * Bit -> Bit * Bit * Bit :=
  lambda ((a, b), c) -> ctrl c [
    &0 -> (c, (a, b));
    &1 -> (c, (@not(a), @not(b)))
  ] |> lambda (c, (a, b)) -> ctrl (a, b) [
    (&1, &1) -> ((a, b), @not(c));
    else -> ((a, b), c);
  ]
end

def @cdkm_uma : Bit * Bit * Bit -> Bit * Bit * Bit :=
  lambda ((a, b), c) -> ctrl (a, b) [
    (&1, &1) -> ((a, b), @not(c));
    else -> ((a, b), c);
  ] |> lambda ((a, b), c) ->
  (a, (b, (c, ()))) |> @cnot{3, 2, 0} |> @cnot{3, 0, 1} |>
  lambda (a, (b, (c, ()))) -> ((a, b), c)
end

def @rev_adder_helper{#n} : Num{#n} * Num{#n} * Bit -> Num{#n} * Num{#n} * Bit :=
  if #n <= 0 then
    @id{Num{#n} * Num{#n} * Bit}
  else
    lambda (((a0, a1), (b0, b1)), c) ->
    let (((ca, ba), c'), (a1, b1)) =
      (@cdkm_maj((c, b0), a0), (a1, b1)) in
    let ((ca, ba), ((a1, s1), c'')) =
      ((ca, ba), @rev_adder_helper{#n - 1}((a1, b1), c')) in
    let ((a1, s1), ((c, s0), a0)) =
      ((a1, s1), @cdkm_uma((ca, ba), c'')) in
    (((a0, a1), (s0, s1)), c)
  endif
end

def @rev_adder{#n} : Num{#n} * Num{#n} -> Num{#n} * Num{#n} :=
  lambda (a, b) -> ((a, b), &0) |> @rev_adder_helper{#n} |> lambda (x, &0) -> x
end

@rev_adder{5}(&repeated{5, Bit, &plus}, &repeated{5, Bit, &plus})
\end{lstlisting}

\begin{figure}[ht]
    \centering
    \includegraphics[width=\linewidth]{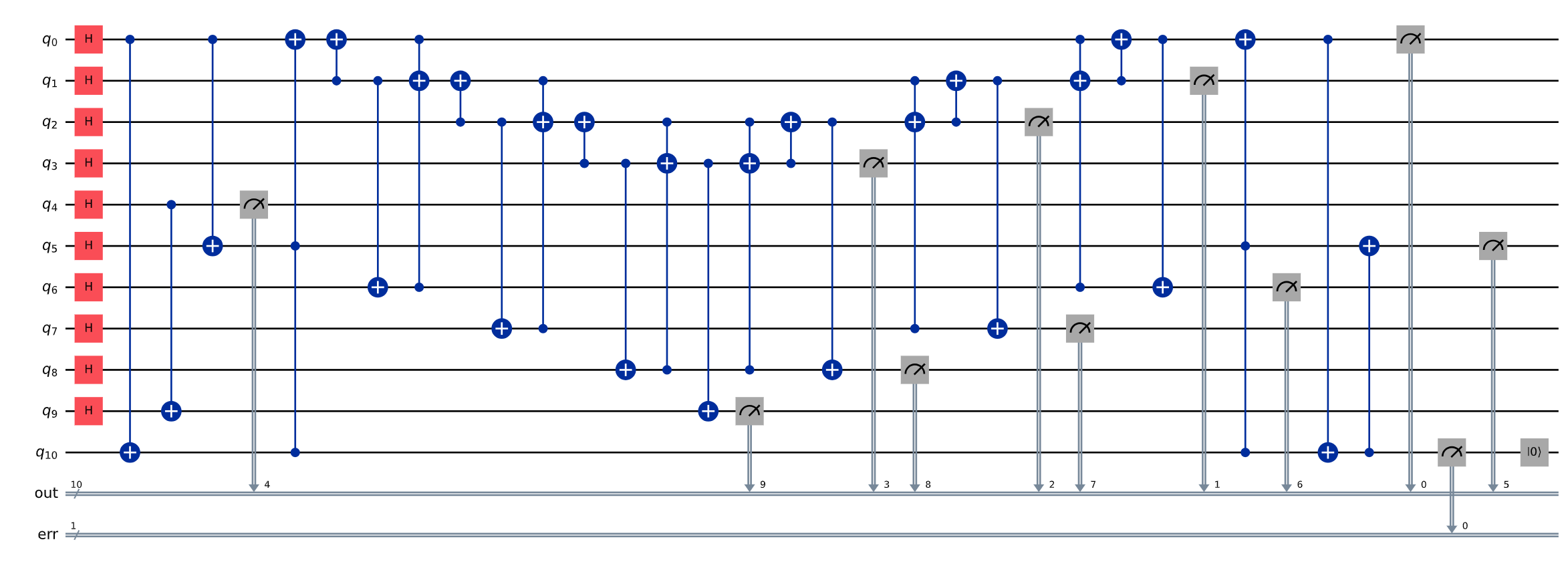}
    \caption{Compiled circuit for the reversible CDKM adder.}
    \Description{}
    \label{fig:cdkm_adder_circuit}
\end{figure}

Qiskit implementation:

\begin{lstlisting}[language=python]
from qiskit.circuit.library import CDKMRippleCarryAdder

circuit = QuantumCircuit(11, 10)
circuit.h(range(10))
circuit.append(CDKMRippleCarryAdder(5, kind="fixed"), range(11))
circuit.measure(range(10), range(10))
\end{lstlisting}

\subsection{Grover (5-Bit Match Oracle, 1 Iteration)}\label{app:grover_impl}

Qunity implementation (for the definition of \lstinline|&grover|, see \Cref{sec:examples_grover}):

\begin{lstlisting}
def #n := 5 end
def &answer : Num{#n} := (&0, (&1, (&1, (&0, (&0, ()))))) end

def @f : Num{#n} -> Bit :=
  lambda x -> ctrl x [
    @answer -> (x, &1);
    else -> (x, &0)
  ] |> @snd{Num{#n}, Bit}
end

&grover{Num{#n}, &repeated{#n, Bit, &plus}, @f, 1}
\end{lstlisting}

\begin{figure}[ht]
    \centering
    \includegraphics[width=0.8\linewidth]{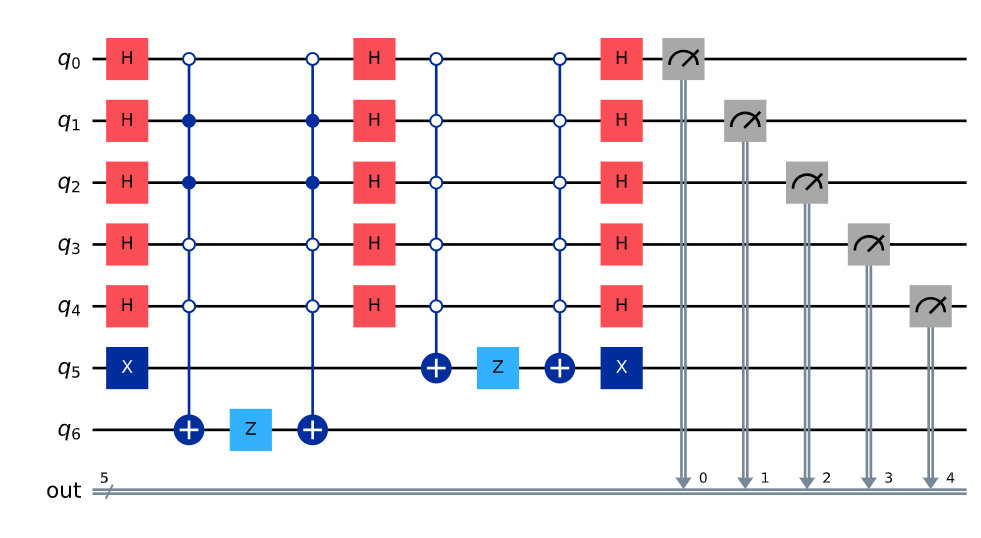}
    \caption{Compiled circuit for Grover's algorithm with a simple \lstinline|match| oracle.}
    \Description{}
    \label{fig:grover_circuit}
\end{figure}

Qiskit implementation:

\begin{lstlisting}[language=python]
circuit = QuantumCircuit(6, 5)
circuit.x(5)
circuit.h(range(6))
circuit.mcx([0, 1, 2, 3, 4], 5, ctrl_state="01100")
circuit.h(range(5))
circuit.mcx([0, 1, 2, 3, 4], 5, ctrl_state="00000")
circuit.h(range(5))
circuit.measure([0, 1, 2, 3, 4], [4, 3, 2, 1, 0])
\end{lstlisting}

\subsection{Grover (List Sum Oracle)}\label{app:grover_with_lists_impl}

Qunity implementation (for the definition of \lstinline|&grover|, see \Cref{sec:examples_grover}):
\begin{lstlisting}
def &equal_superpos_list{#n} : List{#n, Bit} :=
  if #n = 0 then
    &ListEmpty{0, Bit}
  else
    $0
    |> u3{2 * arccos(sqrt(1 / (2 ^ (#n + 1) - 1))), 0, 0}
    |> pmatch [
      &0 -> &ListEmpty{#n, Bit};
      &1 -> @ListCons{#n, Bit}(&plus,
                              &equal_superpos_list{#n - 1})
    ]
  endif
end

def @is_odd_sum{#n} : List{#n, Bit} -> Bit :=
  if #n = 0 then
    lambda l -> &0
  else
    lambda l -> match l [
      &ListEmpty{#n, Bit} -> &0;
      @ListCons{#n, Bit}(&0, l') -> @is_odd_sum{#n - 1}(l');
      @ListCons{#n, Bit}(&1, l') -> @not(@is_odd_sum{#n - 1}(l'))
    ]
  endif
end

def #n := 2 end

&grover{List{#n, Bit}, &equal_superpos_list{#n},
                       @is_odd_sum{#n}, 1}
\end{lstlisting}

\begin{figure}[ht]
    \centering
    \includegraphics[width=0.8\linewidth]{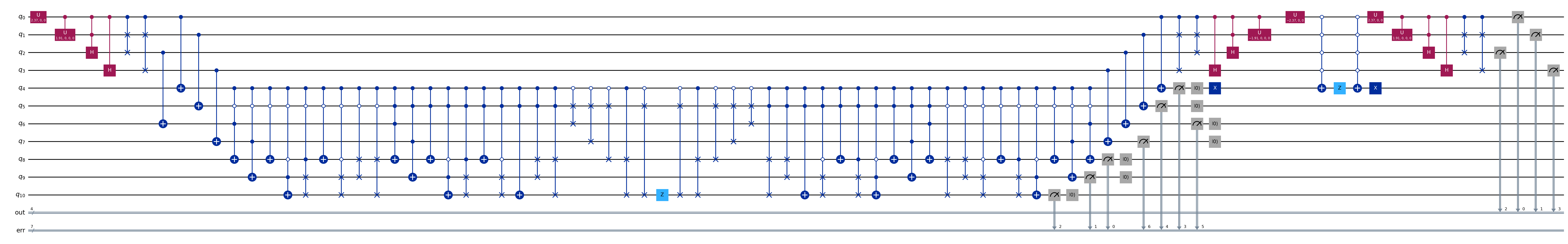}
    \caption{Compiled circuit for Grover's algorithm with a ``list sum'' oracle.}
    \Description{}
    \label{fig:grover_with_lists_circuit}
\end{figure}

\newpage

Qiskit implementation:

\begin{lstlisting}[language=python]
circuit = QuantumCircuit(5, 4)
circuit.x(4)
circuit.h(4)
circuit.ry(2 * math.acos(math.sqrt(1 / (2**3 - 1))), 0)
circuit.ch(0, 1)
circuit.cry(2 * math.acos(math.sqrt(1 / (2**2 - 1))), 0, 2)
circuit.ch(2, 3)
circuit.cx(1, 4)
circuit.cx(3, 4)
circuit.ch(2, 3)
circuit.cry(-2 * math.acos(math.sqrt(1 / (2**2 - 1))), 0, 2)
circuit.ch(0, 1)
circuit.ry(-2 * math.acos(math.sqrt(1 / (2**3 - 1))), 0)
circuit.mcx([0, 1, 2, 3], 4, ctrl_state="0000")
circuit.ry(2 * math.acos(math.sqrt(1 / (2**3 - 1))), 0)
circuit.ch(0, 1)
circuit.cry(2 * math.acos(math.sqrt(1 / (2**2 - 1))), 0, 2)
circuit.ch(2, 3)
circuit.measure([0, 1, 2, 3], [0, 1, 2, 3])
\end{lstlisting}

\section{Circuit-Style Programming}\label{app:circuit_style}

While Qunity allows for programming on a much higher level of abstraction than the quantum circuit model, it is possible to work with Qunity as if it was a low-level circuit language, by defining the following constructs:

\begin{lstlisting}
def @gate_1q{#n, #i, @f : Bit -> Bit} : Array{#n, Bit} -> Array{#n, Bit} :=
  if #i <= 0 then
    lambda (x, y) -> (@f(x), y)
  else
    lambda (x, y) -> (x, @gate_1q{#n - 1, #i - 1, @f}(y))
  endif
end

def @controlled_1q{#n, #i, #j, @f : Bit -> Bit} : Array{#n, Bit} -> Array{#n, Bit} :=
  if #i > #j then
    lambda x -> x
    |> @reverse{#n, Bit}
    |> @controlled_1q{#n, #n - 1 - #i, #n - 1 - #j, @f}
    |> @reverse{#n, Bit}
  else
    if #i <= 0 then
      lambda (x, y) -> ctrl x [
        &0 -> (x, y);
        &1 -> (x, @gate_1q{#n - 1, #j - 1, @f}(y))
      ]
    else
      lambda (x, y) -> (x, @cnot{#n - 1, #i - 1, #j - 1}(y))
    endif
  endif
end

def @cnot{#n, #i, #j} : Array{#n, Bit} -> Array{#n, Bit} :=
  @controlled_1q{#n, #i, #j, @not}
end
\end{lstlisting}

This allows us to use a Qunity array as a quantum register and apply gates to it, accessing qubits by index. Since we have defined single-qubit gates and CNOT in this way, this can in principle be used to represent any quantum computation. If performance is critical, a Qunity programmer may switch to writing in circuit style for some subroutines, since these constructs compile efficiently to the corresponding low-level gates.

\section{The Orthogonality and Spanning Judgments}\label{app:ortho_judgment}

We generalize the orthogonality judgment (used for typing Qunity's control flow constructs), as well as the spanning judgment (now used for the isometry judgment in \Cref{app:isometry_judgment}). The orthogonality judgment in \Cref{fig:ortho} allows the Qunity typechecker to statically check whether a given set of expressions corresponds to a set of orthogonal subspaces in the Hilbert space corresponding to their type. The spanning judgment in \Cref{fig:spanning} corresponds to checking that the expressions are orthogonal \emph{and} that the direct sum of their corresponding subspaces is the whole Hilbert space. Note that both of these judgments correspond to sufficient but not necessary conditions: anything that they recognize as orthogonal (spanning) must be orthogonal (spanning), but the converse is not true.

The main modifications are the $\textsc{O-IsoApp}$ rule allowing the application of arbitrary \emph{isometric} programs to all expressions in the set (which maintains orthogonality), and the $\textsc{S-UnApp}$ rule allowing the application of arbitrary \emph{unitary} programs to all expressions in the set (which maintains the spanning property). This depends on the \emph{unitary judgment}, which is itself defined in terms of the isometry judgment (\Cref{app:isometry_judgment}) as follows:
\[
\inference{\iso(f) \quad \dim(T) = \dim(T')}{\un(\vdash f : T \rightsquigarrow T')}[U-Prog]
\]
That is, if a Qunity program has isometric semantics and maps between two types whose associated Hilbert spaces have the same dimension, it must have unitary semantics.

\begin{figure}[ht]
\[
\inference{}{\ortho{\Void}{}}[\textsc{O-Void}]
\qquad
\inference{}{\ortho{\Unit}{\unit}}[\textsc{O-Unit}]
\qquad
\inference{}{\ortho{T}{x}}[\textsc{O-Var}]
\]
\[
\inference{\ortho{T}{e_1, \ldots, e_n} \qquad \iso(\vdash f : T \rightsquigarrow T')}{\ortho{T'}{f\: e_1, \ldots, f\: e_n}}[\textsc{O-IsoApp}]
\]
\vspace{2mm}
\[
\inference{\ortho{T_0}{e_1, \ldots, e_n} \qquad \ortho{T_1}{e_1', \ldots, e_{n'}'}}{\ortho{T_0 \oplus T_1}{\lef{T_0}{T_1}e_1, &\ldots, \lef{T_0}{T_1}e_n, \\ \rit{T_0}{T_1}e_1', &\ldots, \rit{T_0}{T_1}e_{n'}'}}[\textsc{O-Sum}]
\]
\vspace{2mm}
\[
\inference{\ortho{T_0}{e_1, \ldots, e_m} \qquad \ortho{T_1}{e_{j,1}', \ldots, e_{j,n_j}'} \; \forall j \\ \FV(e_j) \cap \bigcup_{k=1}^{n_m} \FV(e_{j,k}') = \varnothing \; \forall j}{\ortho{T_0 \otimes T_1}{ \pair{e_1}{e_{1,1}'}, &\ldots, \pair{e_1}{e_{1,n_1}'}, \\ &\ldots, \\ \pair{e_m}{e_{m,1}'}, &\ldots, \pair{e_m}{e_{m,n_m}'} }}[\textsc{O-Pair}]
\]
\vspace{2mm}
\[
\inference{\ortho{T}{e_1', \ldots, e_m'} \qquad [e_1, \ldots, e_n] \textup{ is a subsequence of } [e_1', \ldots, e_m']}{\ortho{T}{e_1, \ldots, e_n}}[\textsc{O-Sub}]
\]
\caption{Orthogonality inference rules, modified from the original spanning rules to allow for arbitrary isometry application. Here, $\FV(e)$ means the set of free variables in $e$ (this does not include variables in patterns).}
\Description{}
\label{fig:ortho}
\end{figure}

\begin{figure}[ht]
\[
\inference{}{\spanning{\Void}{}}[\textsc{S-Void}]
\qquad
\inference{}{\spanning{\Unit}{\Unit}}[\textsc{S-Unit}]
\qquad
\inference{}{\spanning{T}{x}}[\textsc{S-Var}]
\]
\[
\inference{\spanning{T}{e_1, \ldots, e_n} \qquad \un(\vdash f : T \rightsquigarrow T')}{\spanning{T'}{f\: e_1, \ldots, f\: e_n}}[\textsc{S-UnApp}]
\]
\vspace{2mm}
\[
\inference{\spanning{T_0}{e_1, \ldots, e_n} \qquad \spanning{T_1}{e_1', \ldots, e_{n'}'}}{\spanning{T_0 \oplus T_1}{\lef{T_0}{T_1}e_1, &\ldots, \lef{T_0}{T_1}e_n, \\ \rit{T_0}{T_1}e_1', &\ldots, \rit{T_0}{T_1}e_{n'}'}}[\textsc{S-Sum}]
\]
\vspace{2mm}
\[
\inference{\spanning{T_0}{e_1, \ldots, e_m} \qquad \spanning{T_1}{e_{j,1}', \ldots, e_{j,n_j}'} \; \forall j \\ \FV(e_j) \cap \bigcup_{k=1}^{n_m} \FV(e_{j,k}') = \varnothing \; \forall j}{\spanning{T_0 \otimes T_1}{ \pair{e_1}{e_{1,1}'}, &\ldots, \pair{e_1}{e_{1,n_1}'}, \\ &\ldots, \\ \pair{e_m}{e_{m,1}'}, &\ldots, \pair{e_m}{e_{m,n_m}'} }}[\textsc{S-Pair}]
\]
\caption{Spanning inference rules, extended to allow the application of arbitrary unitary programs.}
\Description{}
\label{fig:spanning}
\end{figure}

\clearpage

\begin{lemma}\label{lem:ortho_judgment}
Suppose that $\ortho{T}{e_1, \dots, e_n}$ holds. Take any $e_i, e_j$ with $i \neq j$. Let
\[
\msem{e_i} = \msem{\varnothing : \varnothing \partition \Delta_i \vdash e_i : T} : \cH(\Delta_i) \rarr \cH(T)
\]
and similarly for $\msem{e_j}$. Then, the images of the operators $\msem{e_i}$ and $\msem{e_j}$ as subspaces of $\cH(T)$ are orthogonal.
\end{lemma}

\begin{proof}
We prove this by induction on the rule used to prove the orthogonality judgment.

\textsc{O-Void}, \textsc{O-Unit}, \textsc{O-Var}: The statement is vacuously true in these cases, since the sets contain less than two expressions.

\textsc{O-IsoApp}: In this case, we use \Cref{lem:iso_judgment} to claim that the semantics $\msem{\vdash f : T \rightsquigarrow T'}$ is isometric. Since an isometric operator applied to orthogonal vectors preserves their orthogonality, it must be that $\msem{f e_i}$ and $\msem{f e_j}$ are orthogonal subspaces of $T'$.

\textsc{O-Sum}: If the two expressions are both ``left'' or both ``right'', then the claim follows by induction, since we assume that both subsequences form orthogonal sets and the left and right injections are isometric. If one is ``left'' and the other is ``right,'' the image of $\cH(T_0)$ under \textsc{left} and the image of $\cH(T_1)$ under \textsc{right} are orthogonal subspaces in $\cH(T_0 \oplus T_1) = \cH(T_0) \oplus \cH(T_1)$.

\textsc{O-Pair}: 

The \textsc{O-Pair} rule assumes that for all $j$, we have
\[
\FV(e_j) \cap \bigcup_{k=1}^{n_m} \FV(e_{j,k}') = \varnothing,
\]
that is, $e_j$ and $e_{j,k}'$ do not share any free variables. Thus, in the typing rule \textsc{T-PurePair} for the expression $(e_j, e_{j,k}')$ we must have $\Delta = \varnothing$, and so substituting quantum contexts $\Delta_j$ and $\Delta_{j,k}'$ for $\Delta_0, \Delta_1$, the semantics (omitting empty classical contexts for brevity) is given by
\[
\msem{\Delta_j, \Delta_{j,k}' \vdash (e_j, e_{j,k}') : T_0 \otimes T_1} \ket{\tau_0, \tau_1} =
\msem{\Delta_j, \vdash e_j : T_0} \ket{\tau_0} \otimes
\msem{\Delta_{j,k}', \vdash e_j' : T_1} \ket{\tau_1}.
\]
Consider two expressions $(e_i, e_{i,k}')$ and $(e_j, e_{j,l}')$. We then have that the images are
\[
\msem{(e_i, e_{i,k}')}(\cH(\Delta_i, \Delta_{i,k}')) = \msem{e_i}(\cH(\Delta_i)) \otimes \msem{e_{i,k}'}(\cH(\Delta_{i,k}'))
\]
\[
\msem{(e_j, e_{j,l}')}(\cH(\Delta_j, \Delta_{j,l}')) = \msem{e_j}(\cH(\Delta_j)) \otimes \msem{e_{j,l}'}(\cH(\Delta_{j,l}')).
\]
In the case where $i \neq j$, we have that $\msem{e_i}(\cH(\Delta_i))$ and $\msem{e_j}(\cH(\Delta_j))$ are orthogonal subspaces by the assumption that $\ortho{T_0}{e_1, \ldots, e_m}$. If $i = j$, then $\msem{e_{i,k}'}(\cH(\Delta_{i,k}'))$ and $\msem{e_{j,l}'}(\cH(\Delta_{j,l}'))$ are orthogonal subspaces by the assumption that $\ortho{T_1}{e_{j,1}', \ldots, e_{j,n_j}'}$. So, in all cases, the claim holds.

\textsc{O-Sub}: The claim holds trivially in this case, since any subsequence of a sequence of orthogonal expressions is orthogonal. This completes the proof.

\end{proof}

\begin{lemma}\label{lem:spanning_judgment}
Suppose that $\spanning{T}{e_1, \dots, e_n}$ holds. Then, $\ortho{T}{e_1, \dots, e_n}$ holds and
\[
\bigoplus_{j=1}^n \msem{e_j}(\cH(\Delta_j)) \cong \cH(T).
\]
\end{lemma}

\begin{proof}
The first claim is clear, since each typing rule for the spanning judgment is stronger than the corresponding rule for the orthogonality judgment. Now, we prove the spanning property by induction:

\textsc{S-Void}: The claim holds trivially since $\cH(\Void) = \{0\}$, which we consider to be the direct sum of an empty set of spaces.

\textsc{S-Unit}: The image of $\msem{\unit}$ is $\Span\{\ket{\unit}\} = \CC = \cH(\Unit)$.

\textsc{S-Var}: Since $\msem{x : T}$ maps every $\ket{x \mapsto v} \in \cH(\Delta)$ for $v \in \VV(T)$ to $\ket{v} \in \cH(T)$, it is clear by definition of the Hilbert spaces that it implements an isomorphism $\cH(\Delta) \rarr \cH(T)$, and thus its image is all of $\cH(T)$.

\textsc{S-UnApp}: By the unitary rule, we must have that $T$ and $T'$ have the same dimension and $\msem{f}$ is an isometry (and thus a unitary). Thus, it sends orthogonal vectors in $\cH(T)$ to orthogonal vectors in $\cH(T')$. Since the dimensions are equal and the images of $\msem{e_1}, \dots, \msem{e_n}$ contain a basis of $\cH(T)$, we must have that the images of $\msem{f e_1}, \dots, \msem{f e_n}$ contain a basis of $\cH(T')$ and so their direct sum is all of $\cH(T')$. Note that this depends on the correctness of the isometry judgment (\ref{app:isometry_judgment}), which in turn depends on the spanning judgment, but this is not an issue since we are really performing simultaneous induction.

\textsc{S-Sum}: By assumption, we have $\spanning{T_0}{e_1, \dots, e_n}$ and $\spanning{T_1}{e_1', \dots, e_{n'}'}$, so the direct sum of the images contains all of
\[
\braces{\ket{\lef{T_0}{T_1} \: v_0} \mid v_0 \in \VV(T_0)\} \cup \{ \ket{\rit{T_0}{T_1} \: v_1} \mid v_1 \in \VV(T_1)},
\]
which forms a basis for $\cH(T_0 \oplus T_1)$.

\textsc{S-Pair}: By assumption, we have $\spanning{T_0}{e_1, \dots, e_m}$, which means that
\[
\bigoplus_{j=1}^m \msem{e_j}(\cH(\Delta_j)) = \cH(T_0).
\]
And, since for each $j$, we have $\spanning{T_0}{e_{j,1}', \dots, e_{j,n_j}'}$, we also have
\[
\bigoplus_{k=1}^{n_j} \msem{e_{j,k}'}(\cH(\Delta_{j,k}')) = \cH(T_1),
\]
so it must be the case that
\begin{align*}
& \bigoplus_{j=1}^m \bigoplus_{k=1}^{n_j} \msem{(e_j, e_{j,k}')}(\cH(\Delta_j, \Delta_{j,k}')) \\
={}& \bigoplus_{j=1}^m \bigoplus_{k=1}^{n_j} \Span\{\msem{(e_j, e_{j,k}')}\ket{\tau_0, \tau_1} \mid \tau_0 \in \VV(\Delta_j), \tau_1 \in \VV(\Delta_{j,k}')\} = \\
={}& \bigoplus_{j=1}^m \bigoplus_{k=1}^{n_j} \Span\{\msem{e_j}\ket{\tau_0} \otimes \msem{e_{j,k}'}\ket{\tau_1} \mid \tau_0 \in \VV(\Delta_j), \tau_1 \in \VV(\Delta_{j,k}')\} = \\
={}& \bigoplus_{j=1}^m \msem{e_j}(\cH(\Delta_j)) \otimes \bigoplus_{k=1}^{n_j} \msem{e_{j,k}'}(\cH(\Delta_{j,k}')) = \cH(T_0) \otimes \cH(T_1) = \cH(T_0 \otimes T_1).
\end{align*}
Note that in the first equality above we used the assumption that the sets of free variables are disjoint to decompose the context spaces as tensor products. This completes the proof.
\end{proof}

\section{The Isometry Judgment}\label{app:isometry_judgment}

The isometry judgment in \Cref{fig:iso} allows the Qunity typechecker to determine if the given expression or program is guaranteed to have isometric (norm-preserving) semantics, and thus never raise an error. This information can be used by the compiler to perform certain optimizations, such as deleting certain flag qubits that are known to always be in the $\ket{0}$ state due to the fact that the corresponding Qunity expression is isometric.

\begin{figure}[ht]
\[
	\inference{}{\iso(\unit)}[\textsc{I-Unit}]
	\quad
	\inference{}{\iso(x)}[\textsc{I-Var}]
	\quad
	\inference{\iso(e_0) \qquad \iso(e_1)}{\iso(\pair{e_0}{e_1})}[\textsc{I-Pair}]
\]
\vspace{2mm}
\[
\inference{\classical(e) \qquad \iso(e) \qquad \spanning{T}{e_1, \ldots, e_n} \qquad \iso(e_j') \; \forall j}{\iso\left(\cntrl{e}{T}{e_1 &\mapsto e_1' \\ &\cdots \\ e_n &\mapsto e_n'}{T'}\right)}[\textsc{I-Ctrl}]
\]
\vspace{2mm}
\[
\inference{\classical(e) \qquad \iso(e) \qquad \spanning{T}{e_1, \ldots, e_n} \qquad \iso(e_j') \; \forall j}{\iso\left(\match{e}{T}{e_1 &\mapsto e_1' \\ &\cdots \\ e_n &\mapsto e_n'}{T'}\right)}[\textsc{I-Match}]
\]
\vspace{2mm}
\[
	\inference{\iso(e_0)}{\iso(\trycatch{e_0}{e_1})}[\textsc{I-Try}]
	\quad
	\inference{\iso(e_0)}{\iso(\trycatch{e_0}{e_1})}[\textsc{I-Catch}]
	\quad
	\inference{\iso(f) \qquad \iso(e)}{\iso(f \: e)}[\textsc{I-App}]
\]
\vspace{2mm}
\[
	\inference{}{\iso(\uthree{r_\theta}{r_\phi}{r_\lambda})}[\textsc{I-Gate}]
	\quad
	\inference{}{\iso(\lef{T_0}{T_1})}[\textsc{I-Left}]
	\quad
	\inference{}{\iso(\rit{T_0}{T_1})}[\textsc{I-Right}]
\]
\vspace{2mm}
\[
	\inference{\spanning{T}{e} \qquad \iso(e')}{\iso(\lambda e \xmapsto{{\color{gray}T}} e')}[\textsc{I-Abs}]
	\quad
	\inference{\iso(e)}{\iso\left(\rphase{T}{e}{r}{r'}\right)}[\textsc{I-Rphase}]
\]
\vspace{2mm}
\[
\inference{\spanning{T}{e_1, \ldots, e_n} \qquad \iso(e_1') \quad \cdots \quad \iso(e_n')}{\iso\left(\pmatch{T}{e_1 &\mapsto e_1' \\ &\cdots \\ e_n &\mapsto e_n'}{T'}\right)}[\textsc{I-Pmatch}]
\]
\caption{Isometry inference rules.}
\Description{}
\label{fig:iso}
\end{figure}

\newpage

\begin{lemma}\label{lem:iso_judgment}\leavevmode
\begin{itemize}
\item If $\Gamma \partition \Delta \vdash e : T$ holds and $\iso(e)$, then for all $\ket{\psi} \in \cH(\Delta)$, we have
\[
\norm{\msem{\Gamma \partition \Delta \vdash e : T} \ket{\psi}} = \norm{\ket{\psi}}.
\]
\item If $\Gamma \partition \Delta \Vdash e : T$ holds and $\iso(e)$, then for all $\rho \in \cL(\cH(\Delta))$, we have
\[
\tr{\msem{\Gamma \partition \Delta \Vdash e : T} (\rho)} = \tr(\rho).
\]
\item If $\vdash e : T \rightsquigarrow T'$ holds and $\iso(e)$, then for all $\ket{\psi} \in \cH(\Delta)$, we have
\[
\norm{\msem{\vdash e : T \rightsquigarrow T'} \ket{\psi}} = \norm{\ket{\psi}}.
\]
\item If $\vdash e : T \Rrightarrow T'$ holds and $\iso(e)$, then for all $\rho \in \cL(\cH(T))$, we have
\[
\tr{\msem{\vdash e : T \Rrightarrow T'} (\ket{\tau}\bra{\tau})} = \tr(\rho).
\]
\end{itemize}
\end{lemma}

\begin{proof}
\textsc{I-Unit}:
\[
\norm{\msem{\sigma : \Gamma \partition \Delta \vdash \unit : \Unit} \ket{\psi}} = \abs{\braket{\varnothing}{\psi}} \norm{\ket{\unit}} = \norm{\ket{\psi}}.
\]
\textsc{I-Var}:
\begin{align*}
& \norm{\msem{\sigma : \Gamma \partition \Delta \vdash x : T} \ket{\psi}}^2 = \\
={}& \norm{\sum_{v \in \VV(T)} \braket{v}{\psi} \msem{\sigma : \Gamma \partition \Delta \vdash x : T} \ket{x \mapsto v}}^2 =
\norm{\sum_{v \in \VV(T)} \braket{v}{\psi} \ket{v}}^2 =
\sum_{v \in \VV(T)} \abs{\braket{v}{\psi}}^2 = \norm{\ket{\psi}}^2.
\end{align*}
\textsc{I-Pair}:
\begin{align*}
& \norm{\msem{\sigma : \Gamma \partition \Delta, \Delta_0, \Delta_1 \vdash (e_0, e_1) : T_0 \otimes T_1} \ket{\psi}}^2 = \\
={}& \bra{\psi} \msem{\sigma : \Gamma \partition \Delta, \Delta_0, \Delta_1 \vdash (e_0, e_1) : T_0 \otimes T_1}\adj \msem{\sigma : \Gamma \partition \Delta, \Delta_0, \Delta_1 \vdash (e_0, e_1) : T_0 \otimes T_1} \ket{\psi} = \\
={}& \bra{\psi} \msem{\sigma : \Gamma \partition \Delta, \Delta_0, \Delta_1 \vdash (e_0, e_1) : T_0 \otimes T_1}\adj \sum_{(\tau, \tau_0, \tau_1) \in \VV(\Delta, \Delta_0, \Delta_1)} \braket{\tau, \tau_0, \tau_1}{\psi} (\msem{\sigma : \Gamma \partition \Delta, \Delta_0 \vdash e_0 : T_0}\ket{\tau, \tau_0} \otimes \\
    & \hspace{22.5em} \otimes \msem{\sigma : \Gamma \partition \Delta, \Delta_1 \vdash e_1 : T_1}\ket{\tau, \tau_1}) \\
={}& \bra{\psi} \sum_{(\tau', \tau_0', \tau_1') \in \VV(\Delta, \Delta_0, \Delta_1)} \sum_{(\tau, \tau_0, \tau_1) \in \VV(\Delta, \Delta_0, \Delta_1)} \braket{\tau, \tau_0, \tau_1}{\psi} \ket{\tau', \tau_0', \tau_1'} \cdot \\
& \hspace{15em} \cdot \bra{\tau', \tau_0'} \msem{\sigma : \Gamma \partition \Delta, \Delta_0 \vdash e_0 : T_0}\adj \msem{\sigma : \Gamma \partition \Delta, \Delta_0 \vdash e_0 : T_0}\ket{\tau, \tau_0} \cdot \\
& \hspace{15em} \cdot \bra{\tau', \tau_1'} \msem{\sigma : \Gamma \partition \Delta, \Delta_1 \vdash e_1 : T_1}\adj \msem{\sigma : \Gamma \partition \Delta, \Delta_1 \vdash e_1 : T_1}\ket{\tau, \tau_1} = \\
={}& \bra{\psi} \sum_{(\tau', \tau_0', \tau_1') \in \VV(\Delta, \Delta_0, \Delta_1)} \sum_{(\tau, \tau_0, \tau_1) \in \VV(\Delta, \Delta_0, \Delta_1)} \braket{\tau, \tau_0, \tau_1}{\psi} \braket{\tau', \tau_0'}{\tau, \tau_0} \braket{\tau', \tau_1'}{\tau, \tau_1} \ket{\tau', \tau_0', \tau_1'} = \\
={}& \bra{\psi} \sum_{(\tau, \tau_0, \tau_1) \in \VV(\Delta, \Delta_0, \Delta_1)} \braket{\tau, \tau_0, \tau_1}{\psi} = \norm{\ket{\psi}}^2.
\end{align*}

\textsc{I-Ctrl}:
\begin{align*}
& \norm{\msem{\sigma : \Gamma \partition \Delta, \Delta' \vdash \cntrl{e\hspace{-2mm}}{T}{e_1 &\mapsto e_1' \\ &\cdots \\ e_n &\mapsto e_n'}{T'} \hspace{-2mm}: T'} \ket{\psi}} = \\
={}& \left\| \sum_{(\tau, \tau') \in \VV(\Delta, \Delta')} \braket{\tau, \tau'}{\psi} \sum_{v \in \VV(T)} \bra{v} \msem{\Gamma \partition \Delta \Vdash e : T}\left( \op{\sigma, \tau}{\sigma, \tau} \right) \ket{v} \right. \\
         & \left. \cdot \sum_{j=1}^n \sum_{\sigma_j \in \VV(\Gamma_j)} \bra{\sigma_j} \msem{\varnothing : \varnothing \partition \Gamma_j \vdash e_j : T}^\dagger \ket{v}
         \cdot \msem{\sigma, \sigma_j : \Gamma, \Gamma_j \partition \Delta,\Delta' \vdash e_j' : T'} \ket{\tau, \tau'} \right\|
\end{align*}

Since $e$ is classical and $\tr\parens{\msem{\Gamma \partition \Delta \Vdash e : T}\left( \op{\sigma, \tau}{\sigma, \tau} \right)} = 1$ by the isometry assumption for the scrutinee, we must have that
\[
\msem{\Gamma \partition \Delta \Vdash e : T}\left( \op{\sigma, \tau}{\sigma, \tau} \right) = \ket{v}\bra{v}
\]
for some particular $v \in \VV(T)$. Then, the entire expression becomes
\begin{align*}
\norm{\sum_{(\tau, \tau') \in \VV(\Delta, \Delta')} \braket{\tau, \tau'}{\psi} \sum_{j=1}^n \sum_{\sigma_j \in \VV(\Gamma_j)} \bra{\sigma_j} \msem{\varnothing : \varnothing \partition \Gamma_j \vdash e_j : T}^\dagger \ket{v}
\msem{\sigma, \sigma_j : \Gamma, \Gamma_j \partition \Delta,\Delta' \vdash e_j' : T'} \ket{\tau, \tau'}}.
\end{align*}
Now, since the $e_j$ satisfy the spanning judgment, the direct sum of their images in $\cH(T)$ must contain all of $\ket{v'} \in \VV(T)$. Furthermore, since we assume the $e_j$ are all classical (from \textsc{T-Ctrl}), each must map each $\sigma_j$ to some $v' \in \VV(T)$, and there must be a one-to-one correspondence between $v'$ and $(j, \sigma_j)$. Thus, exactly one value of $j$ makes the value of summand nonzero. So, the expression becomes
\begin{align*}
& \norm{\sum_{(\tau, \tau') \in \VV(\Delta, \Delta')} \braket{\tau, \tau'}{\psi} \msem{\sigma, \sigma_j : \Gamma, \Gamma_j \partition \Delta,\Delta' \vdash e_j' : T'} \ket{\tau, \tau'}} = \\
={}& \norm{\msem{\sigma, \sigma_j : \Gamma, \Gamma_j \partition \Delta,\Delta' \vdash e_j' : T'} \ket{\psi}} = \norm{\ket{\psi}},
\end{align*}
which follows from the isometry assumption on the RHS.

\textsc{I-Match}:
\begin{align*}
& \tr\parens{\msem{\sigma : \Gamma \partition \Delta, \Delta_0, \Delta_1 \vdash \match{e\hspace{-2mm}}{T}{e_1 &\mapsto e_1' \\ &\cdots \\ e_n &\mapsto e_n'}{T'} \hspace{-2mm}: T'} (\rho)} = \\
={}& \sum_{(\tau, \tau_0, \tau_1), (\tau', \tau_0', \tau_1') \in \VV(\Delta, \Delta_0, \Delta_1)} \bra{\tau, \tau_0, \tau_1} \rho \ket{\tau', \tau_0', \tau_1'} \sum_{v \in \VV(T)} \bra{v} \parens{\msem{\sigma : \Gamma \partition \Delta, \Delta_0 \Vdash e : T}\parens{\ket{\tau, \tau_0}\bra{\tau', \tau_0'}}} \ket{v} \cdot \\
& \cdot \sum_{j=1}^n \sum_{\sigma_j \in \VV(\Gamma_j)}
\bra{\sigma_j} \msem{\varnothing : \varnothing \partition \Gamma_j \vdash e_j : T}\adj \ket{v} \cdot
\tr\parens{\msem{\sigma, \sigma_j : \Gamma, \Gamma_j \partition \Delta,\Delta_1 \Vdash e_j' : T'}\parens{\ket{\tau, \tau_1}\bra{\tau', \tau_1'}}}
\end{align*}
Then, by the isometry assumption for the scrutinee, we have that $\tr\parens{\msem{\Gamma \partition \Delta, \Delta_0 \Vdash e : T}\left( \op{\tau, \tau_0}{\tau', \tau_0'} \right)} = \tr{\op{\tau, \tau_0}{\tau', \tau_0'}} = \delta_{\tau, \tau'} \delta_{\tau_0, \tau_0'}$. So then, by the same argument as for \textsc{T-Ctrl}, we can eliminate the inner sums and obtaining
\begin{align*}
& \sum_{(\tau, \tau_0, \tau_1) \in \VV(\Delta, \Delta_0, \Delta_1), \tau_1' \in \VV(\Delta_1)} \bra{\tau, \tau_0, \tau_1} \rho \ket{\tau, \tau_0, \tau_1'} \tr\parens{\msem{\sigma, \sigma_j : \Gamma, \Gamma_j \partition \Delta,\Delta_1 \Vdash e_j' : T'}\parens{\ket{\tau, \tau_1}\bra{\tau, \tau_1'}}} = \\
={}& \sum_{(\tau, \tau_0, \tau_1) \in \VV(\Delta, \Delta_0, \Delta_1), \tau_1' \in \VV(\Delta_1)} \bra{\tau, \tau_0, \tau_1} \rho \ket{\tau, \tau_0, \tau_1'} \tr\parens{\ket{\tau, \tau_1}\bra{\tau, \tau_1'}} = \\
={}& \sum_{(\tau, \tau_0, \tau_1) \in \VV(\Delta, \Delta_0, \Delta_1), \tau_1' \in \VV(\Delta_1)} \bra{\tau, \tau_0, \tau_1} \rho \ket{\tau, \tau_0, \tau_1'} \tr\parens{\ket{\tau, \tau_1}\bra{\tau, \tau_1'}} = \\
={}& \sum_{(\tau, \tau_0, \tau_1) \in \VV(\Delta, \Delta_0, \Delta_1)} \bra{\tau, \tau_0, \tau_1} \rho \ket{\tau, \tau_0, \tau_1} = \tr(\rho).
\end{align*}

\textsc{I-Try}:
Assuming $\iso(e_0)$,
\begin{align*}
& \tr\parens{\msem{\sigma : \Gamma \partition \Delta_0,\Delta_1 \Vdash \trycatch{e_0}{e_1} : T}(\rho)} = \\
={}& \sum_{(\tau_0, \tau_1), (\tau_0', \tau_1') \in \VV(\Delta_0, \Delta_1)} \bra{\tau_0, \tau_1} \rho \ket{\tau_0', \tau_1'} \left[ \delta_{\tau_1, \tau_1'} \tr\parens{\msem{\sigma : \Gamma \partition \Delta_0 \Vdash e_0 : T}\left( \op{\tau_0}{\tau_0'} \right)} + \right. \\
& \left. + \delta_{\tau_0, \tau_0'} (1 - \tr(\msem{\sigma : \Gamma \partition \Delta_0 \Vdash e_0 : T}\left( \op{\tau_0}{\tau_0'} \right))) \cdot \tr\parens{\msem{\sigma : \Gamma \partition \Delta_1 \Vdash e_1 : T}\left( \op{\tau_1}{\tau_1'} \right)} \right] = \\
={}& \sum_{\tau_0 \in \VV(\Delta_0), \tau_1 \in \VV(\Delta_1)} \bra{\tau_0, \tau_1} \rho \ket{\tau_0, \tau_1} = \tr(\rho).
\end{align*}

\textsc{I-Catch}:
Assuming $\iso(e_1)$,
\begin{align*}
& \tr\parens{\msem{\sigma : \Gamma \partition \Delta_0,\Delta_1 \Vdash \trycatch{e_0}{e_1} : T}(\rho)} = \\
={}& \sum_{(\tau_0, \tau_1), (\tau_0', \tau_1') \in \VV(\Delta_0, \Delta_1)} \bra{\tau_0, \tau_1} \rho \ket{\tau_0', \tau_1'} \left[ \delta_{\tau_1, \tau_1'} \tr\parens{\msem{\sigma : \Gamma \partition \Delta_0 \Vdash e_0 : T}\left( \op{\tau_0}{\tau_0'} \right)} + \right. \\
& \left. + \delta_{\tau_0, \tau_0'} (1 - \tr(\msem{\sigma : \Gamma \partition \Delta_0 \Vdash e_0 : T}\left( \op{\tau_0}{\tau_0'} \right))) \cdot \tr\parens{\msem{\sigma : \Gamma \partition \Delta_1 \Vdash e_1 : T}\left( \op{\tau_1}{\tau_1'} \right)} \right] = \\
={}& \sum_{(\tau_0, \tau_1), (\tau_0', \tau_1') \in \VV(\Delta_0, \Delta_1)} \bra{\tau_0, \tau_1} \rho \ket{\tau_0', \tau_1'} \left[ \delta_{\tau_1, \tau_1'} \tr\parens{\msem{\sigma : \Gamma \partition \Delta_0 \Vdash e_0 : T}\left( \op{\tau_0}{\tau_0'} \right)} + \right. \\
& \left. + \delta_{\tau_0, \tau_0'} (1 - \tr(\msem{\sigma : \Gamma \partition \Delta_0 \Vdash e_0 : T}\left( \op{\tau_0}{\tau_0} \right))) \delta_{\tau_1, \tau_1'} \right] = \\
={}& \sum_{\tau_0 \in \VV(\Delta_0), \tau_1 \in \VV(\Delta_1)} \bra{\tau_0, \tau_1} \rho \ket{\tau_0, \tau_1} \left[ \tr\parens{\msem{\sigma : \Gamma \partition \Delta_0 \Vdash e_0 : T}\left( \op{\tau_0}{\tau_0} \right)} + \right. \\
& \left. + (1 - \tr(\msem{\sigma : \Gamma \partition \Delta_0 \Vdash e_0 : T}\left( \op{\tau_0}{\tau_0} \right))) \right] = \\
={}& \sum_{\tau_0 \in \VV(\Delta_0), \tau_1 \in \VV(\Delta_1)} \bra{\tau_0, \tau_1} \rho \ket{\tau_0, \tau_1} = \tr(\rho).
\end{align*}
Note that in the above, we used the trace non-increasing property to say that
\[
\tr\parens{\msem{\sigma : \Gamma \partition \Delta_0 \Vdash e_0 : T}\left( \op{\tau_0}{\tau_0'} \right)} \leq \tr\parens{\op{\tau_0}{\tau_0'}} = \delta_{\tau_0, \tau_0'}
\]
so it is zero when $\tau_0 \neq \tau_0'$: this can be seen as a consequence of the correctness of the compilation procedure (\Cref{app:typing_judgment_compilation}).

\textsc{I-App}:
Here, we must consider both pure and mixed typing. If $f$ can be typed as a pure program, we have that
\begin{align*}
& \norm{\msem{\sigma : \Gamma \partition \Delta \vdash f e : T'} \ket{\psi}} =
\norm{\msem{\vdash f : T \rightsquigarrow T'} \msem{\sigma : \Gamma \partition \Delta \vdash e : T} \ket{\psi}} = \\
={}& \norm{\msem{\sigma : \Gamma \partition \Delta \vdash e : T} \ket{\psi}} =
\norm{\ket{\psi}}.
\end{align*}

Similarly, for mixed typing, we have that
\begin{align*}
& \tr\parens{\msem{\sigma : \Gamma \partition \Delta \Vdash f e : T'} (\rho)} = \\
={}& \tr\parens{\msem{\vdash f : T \Rrightarrow T'} \parens{\msem{\sigma : \Gamma \partition \Delta \vdash e : T} (\rho)}} = \\
={}& \tr\parens{\msem{\sigma : \Gamma \partition \Delta \vdash e : T} (\rho)} = \tr(\rho).
\end{align*}

\textsc{I-Gate}:
This is clear since single-qubit unitary gates are isometries.

\textsc{I-Left}, \textsc{I-Right}:
The direct sum injections map between $\cH(T_0)$ (resp. $\cH(T_1)$) and the corresponding isomorphic subspace in $\cH(T_0 \oplus T_1)$, and hence they are clearly isometries.

\textsc{I-Abs}:
Here, we again consider both pure and mixed typing for $e'$. If $e'$ can be typed as a pure expression, then
\begin{align*}
& \norm{\msem{\vdash \lambda e \xmapsto{{\color{gray}T}} e' : T \rightsquigarrow T'}\ket{\psi}} = \\
={}& \norm{\msem{\varnothing : \varnothing \partition \Delta \vdash e' : T'} \msem{\varnothing : \varnothing \partition \Delta \vdash e : T}^\dagger \ket{\psi}} = \\
={}& \norm{\msem{\varnothing : \varnothing \partition \Delta \vdash e : T}^\dagger \ket{\psi}}
\end{align*}
Now, consider the circumstances under which a single expression $e$ can satisfy the spanning judgment. It is only possible for \textsc{S-Sum} to have been applied if one of the two types summed is $\Void$. So, in all cases, the dimension of $\cH(T)$ must be equal to the dimension of $\cH(\Delta)$ (for the case of \textsc{S-Pair} this is enforced by the ``no shared free variables'' restriction). Since $e$ is formed from just variables, pairs, and applications of isometric programs, it has isometric, and thus unitary, semantics. This means that the norm of the expression above must be equal to $\norm{\ket{\psi}}$.

Now, for mixed typing:
\begin{align*}
& \tr\parens{\msem{\vdash \lambda e \xmapsto{{\color{gray}T}} e' : T \Rrightarrow T'}(\rho)} = \\
={}& \sum_{v, v' \in \VV(T)} \bra{v}\rho\ket{v'} \tr\parens{\msem{\varnothing : \varnothing \partition \Delta \Vdash e' : T'}\left(\msem{\varnothing : \varnothing \partition \Delta \vdash e : T}^\dagger \op{v}{v'} \msem{\varnothing : \varnothing \partition \Delta \vdash e : T}\right)} = \\
={}& \sum_{v, v' \in \VV(T)} \bra{v}\rho\ket{v'} \tr\parens{\msem{\varnothing : \varnothing \partition \Delta \vdash e : T}^\dagger \op{v}{v'} \msem{\varnothing : \varnothing \partition \Delta \vdash e : T}} = \\
={}& \sum_{v, v' \in \VV(T)} \bra{v}\rho\ket{v'} \tr\parens{\bra{v'}\msem{\varnothing : \varnothing \partition \Delta \vdash e : T} \msem{\varnothing : \varnothing \partition \Delta \vdash e : T}^\dagger \ket{v}}.
\end{align*}

Now, by the same argument as above, $\msem{\varnothing : \varnothing \partition \Delta \vdash e : T}$ must be unitary, so the expression simplifies to $\tr(\rho)$.

\textsc{I-Rphase}:
\begin{align*}
& \msem{\vdash \rphase{T}{e}{r}{r'} : T \rightsquigarrow T} \ket{\psi} = \\
={}& e^{i r} \msem{\varnothing \partition \Delta \vdash e : T} \msem{\varnothing : \varnothing \partition \Delta \vdash e : T}^\dagger \ket{\psi}
+ e^{i r'} \left(\mathbb{I} - \msem{\varnothing : \varnothing \partition \Delta \vdash e : T} \msem{\varnothing : \varnothing \partition \Delta \vdash e : T}^\dagger\right) \ket{\psi}
\end{align*}
Letting $P = \msem{\varnothing : \varnothing \partition \Delta \vdash e : T} \msem{\varnothing : \varnothing \partition \Delta \vdash e : T}^\dagger$, and observe that by the isometry assumption, $P^2 = P$ and $P\adj = P$. So, $P (\II - P) = P - P^2 = 0$ and $(\II - P)^2 = \II - 2P + P^2 = \II - P$. Then,
\begin{align*}
& \norm{\msem{\vdash \rphase{T}{e}{r}{r'} : T \rightsquigarrow T} \ket{\psi}}^2 = \\
={}& \norm{e^{i r} P \ket{\psi} + e^{i r'} (\II - P) \ket{\psi}} = \\
={}& \parens{e^{-ir} \bra{\psi} P + e^{-i r'} \bra{\psi} (\II - P)} \parens{e^{i r} P \ket{\psi} + e^{i r'} (\II - P) \ket{\psi}} = \\
={}& \bra{\psi} P \ket{\psi} + \bra{\psi} (\II - P) \ket{\psi} = \norm{\ket{\psi}}^2.
\end{align*}

\textsc{I-Pmatch}:
\begin{align*}
& \norm{\msem{\vdash \pmatch{T}{e_1 &\mapsto e_1' \\ &\cdots \\ e_n &\mapsto e_n'}{T'} : T \rightsquigarrow T'}\ket{\psi}}^2 = \\
={}& \norm{\sum_{j=1}^n \msem{\varnothing : \varnothing \partition \Delta_j \vdash e_j' : T'}
\msem{\varnothing : \varnothing \partition \Delta_j \vdash e_j : T}\adj \ket{\psi}}^2 = \\
={}& \sum_{j=1}^n \sum_{k=1}^n \bra{\psi} \msem{\varnothing : \varnothing \partition \Delta_j \vdash e_j : T} \msem{\varnothing : \varnothing \partition \Delta_j \vdash e_j' : T'}\adj \msem{\varnothing : \varnothing \partition \Delta_k \vdash e_k' : T'}
\msem{\varnothing : \varnothing \partition \Delta_k \vdash e_k : T}\adj \ket{\psi}
\end{align*}
By the orthogonality assumption for the RHS in \textsc{T-Pmatch}, we must have $\msem{\varnothing : \varnothing \partition \Delta_j \vdash e_j' : T'}\adj \msem{\varnothing : \varnothing \partition \Delta_j \vdash e_k' : T'} = 0$ unless $j = k$, in which case it must be the identity by the isometry assumption for the RHS, so the expression becomes
\begin{align*}
& \sum_{j=1}^n \bra{\psi} \msem{\varnothing : \varnothing \partition \Delta_j \vdash e_j : T}
\msem{\varnothing : \varnothing \partition \Delta_j \vdash e_j : T}\adj \ket{\psi} = \\
={}& \sum_{v in \VV(T)} \abs{\braket{v}{\psi}} \sum_{j=1}^n \bra{v} \msem{\varnothing : \varnothing \partition \Delta_j \vdash e_j : T}
\msem{\varnothing : \varnothing \partition \Delta_j \vdash e_j : T}\adj \ket{v} = \norm{\ket{\psi}}^2,
\end{align*}
which follows by the spanning assumption on the LHS expressions. This completes the analysis of all the cases for the isometry judgment.

\end{proof}

\section{The Erasure Judgment}\label{app:erasure_judgment}

The erasure judgment (\Cref{fig:erasure}) is used for typing Qunity's \lstinline|ctrl| construct, ensuring that it is possible to correctly perform the necessary uncomputation. This requires that all quantum variables in the scrutinee must be present ``in the same way'' in all the RHS expressions. If the purpose of the control is to apply a controlled phase, this is easily satisfied. If it needs to output some additional data, it needs to be paired with the original variables, possibly inside a nested \lstinline|ctrl|. This condition is necessary to ensure that \lstinline|ctrl| does not discard any quantum information and can be typed as a pure expression.

\begin{figure}[ht]
\[
	\inference{\erases{T}(x; e_1, \ldots, e_n)}{\erases{T}(x; e_1, \ldots, e_{j-1}, e_j \triangleright \gphase{T}{r}, e_{j+1}, \ldots, e_n)}[\textsc{E-Gphase}]
\]
\vspace{2mm}
\[
	\inference{\erases{T}(x; e_1, \ldots, e_{j-1}, e_{j,1}, \ldots, e_{j,m}, e_{j+1}, \ldots, e_n)}{\erases{T}\left(x; e_1, \ldots, e_{j-1}, \cntrl{e}{T'}{e_1' &\mapsto e_{j,1} \\ &\cdots \\ e_m' &\mapsto e_{j,m}}{T}, e_{j+1}, \ldots, e_n\right)}[\textsc{E-Ctrl}]
\]
\vspace{2mm}
\[
	\inference{}{\erases{T}(x; x, x, \ldots, x)}[\textsc{E-Var}]
	\quad
	\inference{\erases{T_0}(x; e_{0,1}, \ldots, e_{0,n})}{\erases{T_0 \otimes T_1}(x; \pair{e_{0,1}}{e_{1,1}}, \ldots, \pair{e_{0,n}}{e_{1,n}})}[\textsc{E-Pair0}]
\]
\vspace{2mm}
\[
	\inference{\erases{T_1}(x; e_{1,1}, \ldots, e_{1,n})}{\erases{T_0 \otimes T_1}(x; \pair{e_{0,1}}{e_{1,1}}, \ldots, \pair{e_{0,n}}{e_{1,n}})}[\textsc{E-Pair1}]
\]
\caption{Erasure inference rules.}
\Description{}
\label{fig:erasure}
\end{figure}

\newpage

\section{Direct Sum Circuit Correctness Proof}\label{app:dirsum_proof}

\begin{figure}[ht]
\centering
\begin{subfigure}{0.4\textwidth}
    \centering
    \begin{quantikz}
    \lstick{$1$} & \octrl{1} & \ctrl{3} & \ctrl{1} & \rstick{1} \\
    \lstick{$s_1$} & \gate[4]{U_0} & \gate{U_1} & \gate[5]{P^{(1)}} & \rstick{$s_1'$} \\
    \lstick{$s_0 - s_1$} &&&& \rstick{$s_0' - s_1'$} \\
    \lstick{$p_1$} && \gate{U_1} & \\
    \lstick{$p_0 - p_1$} &&& \\
    \lstick{$k$} &&&& \rstick{$f$}
    \end{quantikz}
\end{subfigure}
\hfil
\begin{subfigure}{0.4\textwidth}
    \centering
    \begin{quantikz}
    \lstick{$s_1$} && \midstick[2]{} & \qwbundle{s_1'} && \\
    \lstick{$s_0 - s_1$} & \permute{2,1} && \qwbundle{f_1} & \permute{4,1,2,3} & \\
    \lstick{$p_1$} &&&&& \\
    \lstick{$p_0 - p_1$} &&&&&  \\
    \lstick{$k$} &&&&&
    \end{quantikz}
\end{subfigure}
\caption{Improved direct sum circuit for Case 1: $s_0 \geq s_1, s_0' \geq s_1', p_0 \geq p_1$. For this circuit, let $k = \max\{0, f_1 - f_0\}$, and let $f = \max\{f_0, f_1\} = f_0 + k$. The circuit on the left implements the direct sum, while the one on the right is the implementation of the component $P^{(1)}$. Note that where $U_1$ is drawn as two separate boxes, these are not two separate gates, but rather a single gate that only acts on the $s_1$ and $p_1$ registers and not on the one in between them. The curly braces in the circuit on the right indicate a repartitioning of a set of registers taken together into registers of different sizes, preserving the order of qubits.}
\Description{}
\label{fig:dirsum_improved_case1}
\end{figure}
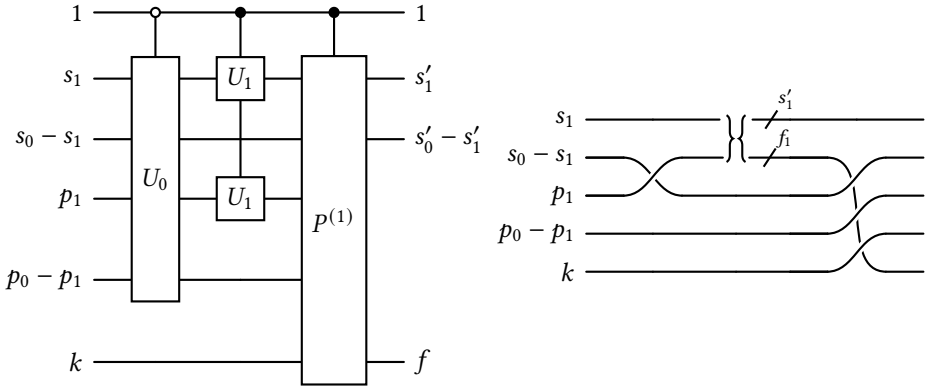

\begin{figure}[ht]
\centering
\begin{subfigure}{0.4\textwidth}
    \centering
    \begin{quantikz}
    \lstick{$1$} & \octrl{1} & \ctrl{3} & \ctrl{1} & \rstick{$1$} \\
    \lstick{$s_1$} & \gate[3]{U_0} & \gate{U_1} & \gate[5]{P^{(2)}} & \rstick{$s_1'$} \\
    \lstick{$s_0 - s_1$} &&&& \rstick{$s_0' - s_1'$} \\
    \lstick{$p_0$} && \gate[2]{U_1} & \\
    \lstick{$p_1 - p_0$} &&& \\
    \lstick{$k$} &&&& \rstick{$f$}
    \end{quantikz}
\end{subfigure}
\hfil
\begin{subfigure}{0.4\textwidth}
    \centering
    \begin{quantikz}
    \lstick{$s_1$} && \midstick[3]{} \\
    \lstick{$s_0 - s_1$} & \permute{3,1,2} && \qwbundle{s_1'} && \\
    \lstick{$p_0$} &&& \qwbundle{f_1} & \permute{3,1,2} & \\
    \lstick{$p_1 - p_0$} &&&&& \\
    \lstick{$k$} &&&&&
    \end{quantikz}
\end{subfigure}
\caption{Improved direct sum circuit for Case 2: $s_0 \geq s_1, s_0' \geq s_1', p_0 \leq p_1$. For this circuit, let \\ $k = \max\{0, f_1 - f_0 + p_0 - p_1\}$, and let $f = \max\{f_0 + p_1 - p_0, f_1\}$. The circuit on the left implements the direct sum, while the one on the right is the implementation of the component $P^{(2)}$.}
\Description{}
\label{fig:dirsum_improved_case2}
\end{figure}
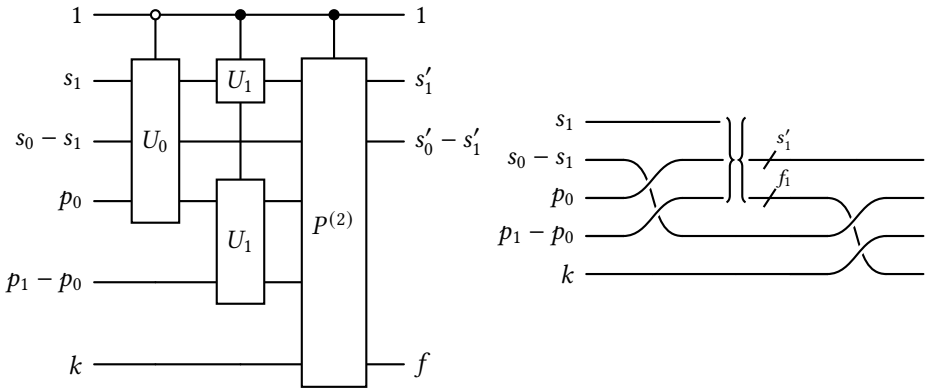

\begin{figure}[ht]
\centering
\begin{subfigure}{0.5\textwidth}
    \centering
    \begin{adjustbox}{width=\textwidth}
    \begin{quantikz}
    \lstick{$1$} & \octrl{1} & \ctrl{3} & \octrl{1} & \ctrl{1} & \rstick{$1$} \\
    \lstick{$s_1$} & \gate[4]{U_0} & \gate{U_1} & \gate[5]{P^{(3)}_0} & \gate[5]{P^{(3)}_1} & \rstick{$s_0'$} \\
    \lstick{$s_0 - s_1$} &&&&& \rstick{$s_1' - s_0'$} \\
    \lstick{$p_1$} && \gate{U_1} && \\
    \lstick{$p_0 - p_1$} &&&&  \\
    \lstick{$k$} &&&&& \rstick{$f$}
    \end{quantikz}
    \end{adjustbox}
\end{subfigure}
\hfil
\begin{subfigure}[valign=c]{0.25\textwidth}
    \begin{subfigure}{\textwidth}
        \centering
        \begin{adjustbox}{width=\textwidth}
        \begin{quantikz}
        \lstick{$s_1$} & \midstick[4]{} \\
        \lstick{$s_0 - s_1$} & \\
        \lstick{$p_1$} && \qwbundle{s_0'} && \\
        \lstick{$p_0 - p_1$} && \qwbundle{f_0} & \permute{2,1} & \\
        \lstick{$k$} &&&&
        \end{quantikz}
        \end{adjustbox}
    \end{subfigure}
    \par\bigskip\bigskip\bigskip
    \begin{subfigure}{\textwidth}
        \centering
        \begin{adjustbox}{width=\textwidth}
        \begin{quantikz}
        \lstick{$s_1$} && \midstick[2]{} & \qwbundle{s_1'} \\
        \lstick{$s_0 - s_1$} & \permute{2,1} && \qwbundle{f_1} \\
        \lstick{$p_1$} &&& \\
        \lstick{$p_0 - p_1$} &&& \\
        \lstick{$k$} &&&
        \end{quantikz}
        \end{adjustbox}
    \end{subfigure}
\end{subfigure}
\caption{Improved direct sum circuit for Case 3: $s_0 \geq s_1, s_0' \leq s_1', p_0 \geq p_1$. For this circuit, let $k = s_1' - s_0'$, and let $f = f_0$. The circuit on the left implements the direct sum, while the ones on the right are the implementations of the components $P^{(3)}_0$ and $P^{(3)}_1$.}
\Description{}
\label{fig:dirsum_improved_case3}
\end{figure}

\begin{figure}[ht]
\centering
\begin{subfigure}{0.5\textwidth}
    \centering
    \begin{adjustbox}{width=\textwidth}
    \begin{quantikz}
    \lstick{$1$} & \octrl{1} & \ctrl{3} & \octrl{1} & \ctrl{1} & \rstick{$1$} \\
    \lstick{$s_1$} & \gate[3]{U_0} & \gate{U_1} & \gate[5]{P^{(4)}_0} & \gate[5]{P^{(4)}_1} & \rstick{$s_0'$} \\
    \lstick{$s_0 - s_1$} &&&&& \rstick{$s_1' - s_0'$} \\
    \lstick{$p_0$} && \gate[2]{U_1} && \\
    \lstick{$p_1 - p_0$} &&&& \\
    \lstick{$k$} &&&&& \rstick{$f$}
    \end{quantikz}
    \end{adjustbox}
\end{subfigure}
\hfil
\begin{subfigure}[valign=c]{0.25\textwidth}
    \begin{subfigure}{\textwidth}
        \centering
        \begin{adjustbox}{width=\textwidth}
        \begin{quantikz}
        \lstick{$s_1$} & \midstick[3]{} \\
        \lstick{$s_0 - s_1$} && \qwbundle{s_0'} && \\
        \lstick{$p_0$} && \qwbundle{f_0} & \permute{3,1,2} & \\
        \lstick{$p_1 - p_0$} &&&& \\
        \lstick{$k$} &&&&
        \end{quantikz}
        \end{adjustbox}
    \end{subfigure}
    \par\bigskip\bigskip\bigskip
    \begin{subfigure}{\textwidth}
        \centering
        \begin{adjustbox}{width=\textwidth}
        \begin{quantikz}
        \lstick{$s_1$} && \midstick[3]{} \\
        \lstick{$s_0 - s_1$} & \permute{3,1,2} && \qwbundle{s_1'} \\
        \lstick{$p_0$} &&& \qwbundle{f_1} \\
        \lstick{$p_1 - p_0$} &&& \\
        \lstick{$k$} &&&
        \end{quantikz}
        \end{adjustbox}
    \end{subfigure}
\end{subfigure}
\caption{Improved direct sum circuit for Case 4: $s_0 \geq s_1, s_0' \leq s_1', p_0 \leq p_1$. For this circuit, let \\ $k = \max\{0, s_1' - s_0' + p_0 - p_1\}$, and let $f = s_0 + p_1 - s_1' + k$. The circuit on the left implements the direct sum, while the ones on the right are the implementations of the components $P^{(4)}_0$ and $P^{(4)}_1$.}
\Description{}
\label{fig:dirsum_improved_case4}
\end{figure}

\clearpage

\begin{lemma}
\label{lem:dirsum_possible}
Suppose it is possible to implement (as specified in Definition \ref{def:possible_to_implement_op}) the operators \\ $E_0 : \cH(T_0) \rarr \cH(T_0')$ and $E_1 : \cH(T_1) \rarr \cH(T_1')$. Then it is possible to implement the operator $E_0 \oplus E_1 : \cH(T_0 \oplus T_1) \rarr \cH(T_0' \oplus T_1')$, using the circuits shown in Figures \ref{fig:dirsum_improved_case1}, \ref{fig:dirsum_improved_case2}, \ref{fig:dirsum_improved_case3}, and \ref{fig:dirsum_improved_case4}.
\end{lemma}

\begin{proof}

Let $s_0 = \size(T_0)$, $s_0' = \size(T_0')$, $s_1 = \size(T_1)$, $s_1' = \size(T_1')$. Suppose that $E_0$ is implemented by a unitary $U_0$ using $p_0$ prep qubits and $p_1$ flag qubits, and $E_1$ is implemented by a unitary $U_1$ using $p_1$ prep qubits and $f_1$ flag qubits. We have that
\[
s_0 + p_0 = s_0' + f_0
\]
\[
s_1 + p_1 = s_1' + f_1
\]
We introduce the sets of computational basis states for the qubits in the corresponding registers: we label these as $\ket{\xi_0}, \ket{\pi_0}, \ket{\xi_0'}, \ket{\phi_0}$ for the registers of size $s_0, p_0, s_0', f_0$ respectively, and $\ket{\xi_1}, \ket{\pi_1}, \ket{\xi_1'}, \ket{\phi_1}$ for $s_1, p_1, s_1', f_1$. 

Now, without loss of generality, we can assume that $s_0 \geq s_1$, because if this were not the case, then we can implement $E_1 \oplus E_0$ and apply the direct sum's commutativity isomorphism by simply surrounding the circuit with Pauli $X$ gates on the first ``signal'' qubit. Now, we have four cases to analyze.

\textbf{Case 1}: Suppose that $s_0' \geq s_1'$ and $p_0 \geq p_1$. For this case, let $k = \max\{0, f_1 - f_0\}$. Let \\ $f = \max\{f_0, f_1\} = f_0 + k$. Then, consider the circuit $U$ shown in \Cref{fig:dirsum_improved_case1}. First, to verify that the wire counts are correct:
\[
1 + s_0 + p_0 + k =
1 + s_0 + p_0 + \max\{0, f_1 - f_0\} =
1 + s_0' + f_0 + \max\{0, f_1 - f_0\} =
1 + s_0' + \max\{f_0, f_1\} =
1 + s_0' + f
\]

Now, to verify the correctness of the circuit:
\begin{align*}
& \bra{0, \enc(v_0'), 0}U\ket{0, \enc(v_0), 0, 0} =
\bra{\enc(v_0'), 0}U_0\ket{\enc(v_0), 0} = \\
={}& \bra{v_0'}E_0\ket{v_0} =
(\bra{v_0'} \oplus 0)(E_0 \oplus E_1)(\ket{v_0} \oplus 0) \\
\\
& \bra{1, \enc(v_1'), 0, 0}U\ket{0, \enc(v_0), 0, 0} = 0 =
(0 \oplus \bra{v_1'})(E_0 \oplus E_1)(\ket{v_0} \oplus 0) \\
\\
& \bra{0, \enc(v_0'), 0}U\ket{1, \enc(v_1), 0, 0, 0, 0} = 0 =
(\bra{v_0'} \oplus 0)(E_0 \oplus E_1)(0 \oplus \ket{v_1}) \\
\\
& \bra{1, \enc(v_1'), 0, 0}U\ket{1, \enc(v_1), 0, 0, 0, 0} = \\
={}& \bra{\enc(v_1'), 0, 0} P^{(1)} \sum_{\xi_1, \pi_1} \ket{\xi_1, 0, \pi_1, 0, 0}\bra{\xi_1, \pi_1}U_1\ket{\enc(v_1), 0} = \\
={}& \bra{\enc(v_1'), 0, 0} \sum_{\xi_1, \pi_1} \sum_{\xi_1', \phi_1} \ket{\xi_1', 0, 0, 0, \phi_1}\braket{\xi_1', \phi_1}{\xi_1, \pi_1}\bra{\xi_1, \pi_1}U_1\ket{\enc(v_1), 0} = \\
={}& \bra{\enc(v_1'), 0, 0} \sum_{\xi_1'} \ket{\xi_1', 0, 0, 0, 0}\bra{\xi_1', 0}U_1\ket{\enc(v_1), 0} = \\
={}& \bra{\enc(v_1'), 0}U_1\ket{\enc(v_1), 0} = \\
={}& \bra{\enc(v_1'), 0}U_1\ket{\enc(v_1), 0} = \bra{v_1'}E_1\ket{v_1} =
(0 \oplus \bra{v_1'})(E_0 \oplus E_1)(0 \oplus \ket{v_1})
\end{align*}

And now, to verify the preservation of encoding validity:

If a valid encoding of $T_0$ is given as input, the output will always be a valid encoding of $T_0'$ since $s_0' \geq s_1'$. If a valid encoding of $T_1$ is given as input, observe that the register of size $s_0 - s_1$ must be in the $\ket{0}$ state. It is not affected by the application of $U_1$. Then, the output register of size $s_0' - s_1'$ can only be formed from qubits in the input registers of sizes $s_0 - s_1$, $p_0 - p_1$, and $k$, and the latter two must also be $\ket{0}$ since they are part of the prep register. Thus $U$ preserves encoding validity.

Now, for the adjoint $U^\dagger$, we consider the version of this circuit run in reverse, where the flag registers are now prep registers and vice versa. If a valid encoding of $T_0'$ is given as input, the output will always be a valid encoding of $T_0$ since $s_0 \geq s_1$. If a valid encoding of $T_1$ is given as input, the register of size $s_0' - s_1'$ must be in the $\ket{0}$ state. The output register of size $s_0 - s_1$ can only be formed from qubits in that register and from the input register of size $f$, which is a prep register. Thus, $U^\dagger$ also preserves encoding validity.

\textbf{Case 2}: Suppose that $s_0' \geq s_1'$ and $p_0 \leq p_1$. For this case, let $k = \max\{0, f_1 - f_0 + p_0 - p_1\}$, and let $f = \max\{f_0 + p_1 - p_0, f_1\}$. Then, consider the circuit $U$ shown in \Cref{fig:dirsum_improved_case2}. First, to verify that the wire counts are correct:
\begin{gather*}
1 + s_0 + p_1 + k =
1 + s_0 + p_1 + \max\{0, f_1 - f_0 + p_0 - p_1\} = \\
= 1 + s_0' + f_0 - p_0 + p_1 + \max\{0, f_1 - f_0 + p_0 - p_1\} =
1 + s_0' + \max\{f_0 + p_1 - p_0, f_1\} =
1 + s_0' + f
\end{gather*}

Now, to verify the correctness of the circuit (we will now omit the obvious cases where the state of the first qubit does not match):

\begin{align*}
& \bra{0, \enc(v_0'), 0}U\ket{0, \enc(v_0), 0, 0, 0} =
\bra{\enc(v_0'), 0}U_0\ket{\enc(v_0), 0} = \\
={}& \bra{v_0'}E_0\ket{v_0} =
(\bra{v_0'} \oplus 0)(E_0 \oplus E_1)(\ket{v_0} \oplus 0) \\
& \bra{1, \enc(v_1'), 0, 0}U\ket{1, \enc(v_1), 0, 0, 0} = \\
={}& \bra{\enc(v_1'), 0, 0} P^{(2)} \sum_{\xi_1, \pi_1} \ket{\xi_1, 0, \pi_1, 0}\bra{\xi_1, 0, \pi_1, 0}U_1\ket{\enc(v_1), 0} = \\
={}& \bra{\enc(v_1'), 0, 0} \sum_{\xi_1, \pi_1} \sum_{\xi_1', \phi_1} \ket{\xi_1', 0, 0, \phi_1} \braket{\xi_1, \phi_1}{\xi_1, \pi_1}\bra{\xi_1, 0, \pi_1, 0}U_1\ket{\enc(v_1), 0} = \\
={}& \bra{\enc(v_1'), 0, 0} \sum_{\xi_1'} \ket{\xi_1', 0, 0, 0} \bra{\xi_1', 0}U_1\ket{\enc(v_1), 0} = \\
={}& \bra{\enc(v_1'), 0}U_1\ket{\enc(v_1), 0} = \bra{v_1'}E_1\ket{v_1} =
(0 \oplus \bra{v_1'})(E_0 \oplus E_1)(0 \oplus \ket{v_1})
\end{align*}

And now, to verify the preservation of encoding validity:

The encoding validity for $T_0$ and $T_0'$ is clear as in the previous case. If a valid encoding of $T_1$ is given as input, the register of size $s_0' - s_1'$ is again formed only from the register of size $s_0 - s_1$ and prep registers (not the flag register since $f \geq f_1$). Similarly, for the adjoint circuit, the register of size $s_0 - s_1$ is only formed from the registers of size $s_0' - s_1'$ and $f$. So, the encoding validity is preserved.

\textbf{Case 3}: Suppose that $s_0' \leq s_1'$ and $p_0 \geq p_1$. For this case, let $k = s_1' - s_0'$, and let $f = f_0$. Then, consider the circuit $U$ shown in \Cref{fig:dirsum_improved_case3}. First, to verify that the wire counts are correct:
\[
1 + s_0 + p_0 + k = 1 + s_0' + f_0 + s_1' - s_0' = 1 + s_1' + f
\]

Now, to verify the correctness of the circuit:
\begin{align*}
& \bra{0, \enc(v_0'), 0, 0}U\ket{0, \enc(v_0), 0, 0} = \\
={}& \bra{\enc(v_0'), 0, 0} P^{(3)}_0 \sum_{\xi_0, \pi_0} \ket{\xi_0, \pi_0, 0} \bra{\xi_0, \pi_0} U_0 \ket{\enc(v_0), 0} = \\
={}& \bra{\enc(v_0'), 0, 0} \sum_{\xi_0, \pi_0} \sum_{\xi_0', \phi_0} \ket{\xi_0', 0, \phi_0} \braket{\xi_0', \phi_0}{\xi_0, \pi_0} \bra{\xi_0, \pi_0} U_0 \ket{\enc(v_0), 0} = \\
={}& \bra{\enc(v_0'), 0, 0} \sum_{\xi_0'} \ket{\xi_0', 0, 0} \bra{\xi_0', 0} U_0 \ket{\enc(v_0), 0} = \\
={}& \bra{\enc(v_0'), 0}U_0\ket{\enc(v_0), 0} = \bra{v_0'}E_0\ket{v_0} =
(\bra{v_0'} \oplus 0)(E_0 \oplus E_1)(\ket{v_0} \oplus 0) \\
& \bra{1, \enc(v_1'), 0}U\ket{1, \enc(v_1), 0, 0, 0} = \\
={}& \bra{\enc(v_1'), 0} P^{(3)}_1 \sum_{\xi_1, \pi_1} \ket{\xi_1, 0, \pi_1, 0, 0} \bra{\xi_1, \pi_1} U_1 \ket{\enc(v_1), 0} = \\
={}& \bra{\enc(v_1'), 0} \sum_{\xi_1, \pi_1} \sum_{\xi_1', \phi_1} \ket{\xi_1', \phi_1, 0, 0, 0} \braket{\xi_1', \phi_1}{\xi_1, \pi_1} \bra{\xi_1, \pi_1} U_1 \ket{\enc(v_1), 0} = \\
={}& \bra{\enc(v_1'), 0} \sum_{\xi_1'} \ket{\xi_1', 0, 0, 0, 0} \bra{\xi_1', \phi_1} U_1 \ket{\enc(v_1), 0} = \\
={}& \bra{\enc(v_1'), 0}U_1\ket{\enc(v_1), 0} = \bra{v_1'}E_1\ket{v_1} =
(0 \oplus \bra{v_1'})(E_0 \oplus E_1)(0 \oplus \ket{v_1})
\end{align*}

The encoding validity preserving property follows by a similar argument as before. If a valid encoding of $T_0$ is given as input, the register of size $s_1' - s_0'$ exactly corresponds to the prep register of size $k$, which must be in the $\ket{0}$ state. If a valid encoding of $T_1$ is given as input, the encoding validity of the output is clear. For the adjoint circuit, if $T_0'$ is given, the encoding validity is clear since $s_0 \geq s_1$. If $T_1'$ is given the output register of size $s_0 - s_1$ is only formed from the prep registers. So, the encoding validity is preserved.

\textbf{Case 4}: Suppose that $s_0' \leq s_1'$ and $p_0 \leq p_1$. For this case, let $k = \max\{0, s_1' - s_0' + p_0 - p_1\}$, and let $f = s_0 + p_1 - s_1' + k$. Then, consider the circuit $U$ shown in \Cref{fig:dirsum_improved_case4}. First, to verify that the wire counts are correct:
\[
1 + s_0 + p_1 + k =
1 + s_0 + p_1 + f - s_0 - p_1 + s_1' =
1 + s_1' + f
\]
Now, to verify the correctness of the circuit:
\begin{align*}
& \bra{0, \enc(v_0'), 0, 0}U\ket{0, \enc(v_0), 0, 0, 0} = \\
={}& \bra{\enc(v_0'), 0, 0} P^{(4)}_0 \sum_{\xi_0, \pi_0} \ket{\xi_0, \pi_0, 0, 0} \bra{\xi_0, \pi_0} U_0 \ket{\enc(v_0), 0} = \\
={}& \bra{\enc(v_0'), 0, 0} \sum_{\xi_0, \pi_0} \sum_{\xi_0', \phi_0} \ket{\xi_0', 0, 0, \phi_0} \braket{\xi_0', \phi_0}{\xi_0, \pi_0} \bra{\xi_0, \pi_0} U_0 \ket{\enc(v_0), 0} = \\
={}& \bra{\enc(v_0'), 0, 0} \sum_{\xi_0'} \ket{\xi_0', 0, 0, 0} \bra{\xi_0', 0} U_0 \ket{\enc(v_0), 0} = \\
={}& \bra{\enc(v_0'), 0}U_0\ket{\enc(v_0), 0} = \bra{v_0'}E_0\ket{v_0} =
(\bra{v_0'} \oplus 0)(E_0 \oplus E_1)(\ket{v_0} \oplus 0) \\
& \bra{1, \enc(v_1'), 0}U\ket{1, \enc(v_1), 0, 0} = \\
={}& \bra{\enc(v_1'), 0} P^{(4)}_1 \sum_{\xi_1, \pi_1} \ket{\xi_1, 0, \pi_1, 0} \bra{\xi_1, \pi_1} U_1 \ket{\enc(v_1), 0} = \\
={}& \bra{\enc(v_1'), 0} \sum_{\xi_1, \pi_1} \sum_{\xi_1', \phi_1} \ket{\xi_1', \phi_1, 0, 0} \braket{\xi_1', \phi_1}{\xi_1, \pi_1} \bra{\xi_1, \pi_1} U_1 \ket{\enc(v_1), 0} = \\
={}& \bra{\enc(v_1'), 0} \sum_{\xi_1'} \ket{\xi_1', 0, 0, 0} \bra{\xi_1', \phi_1} U_1 \ket{\enc(v_1), 0} = \\
={}& \bra{\enc(v_1'), 0}U_1\ket{\enc(v_1), 0} = \bra{v_1'}E_1\ket{v_1} =
(0 \oplus \bra{v_1'})(E_0 \oplus E_1)(0 \oplus \ket{v_1})
\end{align*}

For the encoding validity: if a valid encoding of $T_0$ is given as input, the register of size $s_1' - s_0'$ is only formed from the prep registers of size $p_1 - p_0$ and $k$: we can confirm that
\begin{gather*}
f = s_0 + p_1 - s_1' + \max\{0, s_1' - s_0' + p_0 - p_1\} =
s_0 + \max\{p_1 - s_1', - s_0' + p_0\} = \\
= s_0' + f_0 - p_0 + \max\{p_1 - s_1', - s_0' + p_0\} =
f_0 + \max\{p_1 - s_1' - p_0 + s_0', 0\} \geq f_0,
\end{gather*}
so the flag register does not overlap with the $s_1' - s_0'$ register. If a valid encoding of $T_1$ is given as input, the encoding validity of the output is clear. For the adjoint: if a valid encoding of $T_0'$ is given as input, the encoding validity is clear. If a valid encoding of $T_1'$ is given as input, the register of size $s_0 - s_1$ is only formed from the prep registers, so the encoding validity is preserved.

This completes the analysis of all four cases, and thus completes the proof for the direct sum circuit construction.
\end{proof}

\section{Notation and Definitions for Binary Trees}\label{app:binary_trees}

First, we can define a binary tree $\cR$ as either $\Leaf$ or $(\cR_0, \cR_1)$, where $\cR_0$ and $\cR_1$ are binary trees, the left and right children of $\cR$.

We can define size and height functions as:
\begin{definition}
\begin{gather*}
\size(\Leaf) := 1 \\
\size((\cR_0, \cR_1)) := \size(\cR_0) + \size(\cR_1) \\
\height(\Leaf) := 0 \\
\height((\cR_0, \cR_1)) := 1 + \max(\height(\cR_0), \height(\cR_1))
\end{gather*}
\end{definition}

Then, we can define the following:
\begin{definition}[Direct sum of operators (or states or spaces) along a binary tree]
\begin{gather*}
\bigoplus_{j : \Leaf} E_j := E_1 \\
\bigoplus_{j : (\cR_0, \cR_1)} E_j :=
\parens{\bigoplus_{j : \cR_0} E_j} \oplus
\parens{\bigoplus_{j : \cR_1} E_{j + \size(\cR_0)}}
\end{gather*}
As a shorthand, we can write
\[
E^{\oplus \cR} = \bigoplus_{j : \cR} E.
\]
\end{definition}

\begin{definition}[Direct sum of operators (or states or spaces) along a binary tree, leveled to height $h$]\label{def:leveled_tree_sum}
\begin{gather*}
\bigoplus_{j : L(\Leaf, 0)} E_j := E_1 \\
\bigoplus_{j : L(\Leaf, h)} E_j := \parens{\bigoplus_{j : L(\Leaf, h - 1)} E_j} \oplus 0 \qquad \text{(where $h \in \NN$, $h > 0$)} \\
\bigoplus_{j : L((\cR_0, \cR_1), h)} E_j :=
\parens{\bigoplus_{j : L(\cR_0, h - 1)} E_j} \oplus
\parens{\bigoplus_{j : L(\cR_1, h - 1)} E_{\size(\cR_0)}}
\end{gather*}
Note that the $0$ in the above definition exists in the space $\{0\}$: thus, leveling operation is effectively an extension of the additive unit isomorphism. As a shorthand, we can write
\[
\bigoplus_{j : L(\cR)} E_j := \bigoplus_{j : L(\cR, \height(\cR))} E_j,
\]
and
\[
E^{\oplus L(\cR)} = \bigoplus_{j : L(\cR)} E.
\]
\end{definition}

Summing along a leveled tree, as in \Cref{def:leveled_tree_sum}, can be thought of as adding left children to all nodes in the tree, until they reach a specified height (which must be at least the height of the tree). This construction will be useful when constructing the new spanning circuit.

\begin{definition}[Indexed injection into a tree]\label{def:indexed_injection}
\begin{gather*}
\inj_1^\Leaf \ket{v} = \ket{v} \\
\inj_k^{(\cR_0, \cR_1)} \ket{v} = \begin{cases}
\inj_k^{\cR_0} \ket{v} \oplus 0^{\oplus \cR_1} & \text{if}\; 1 \leq k \leq \size(\cR_1) \\
0^{\oplus \cR_1} \oplus \inj_{k - \size(\cR_0)}^{\cR_1} \ket{v} & \text{otherwise}.
\end{cases}
\end{gather*}
\end{definition}

\begin{definition}[Indexed injection into a leveled tree]\label{def:leveled_indexed_injection}
\begin{gather*}
\inj_1^{L(\Leaf, 0)} \ket{v} = \ket{v} \\
\inj_1^{L(\Leaf, h)} \ket{v} = \inj_1^{L(\Leaf, h-1)} \ket{v} \oplus 0 \qquad \text{(where $h \in \NN$, $h > 0$)} \\
\inj_k^{L((\cR_0, \cR_1), h)} \ket{v} = \begin{cases}
\inj_k^{L(\cR_0, h-1)} \ket{v} \oplus 0^{\oplus L(\cR_1,h-1)} & \text{if}\; 1 \leq k \leq \size(\cR_1) \\
0^{\oplus L(\cR_1, h-1)} \oplus \inj_{k - \size(\cR_0)}^{L(\cR_1,h-1)} \ket{v} & \text{otherwise}.
\end{cases}
\end{gather*}
\end{definition}

\section{Proofs of Correctness for Low-Level Compilation}\label{app:low_level_compilation}

In this section, we reproduce the proofs from Appendix H.1 in \citeauthor{Voichick_2023}, with appropriate modifications. The constructions in this section correspond to primitive operations in Qunity's intermediate representation. Note, however, that this does not constitute a proof of correctness of the entire Qunity compilation procedure, as mathematically formalizing the circuit specification data structures, the process of circuit instantiation, and qubit allocation and recycling in registers is a complex task beyond the scope of this work. Note that we do not need the ``validation'' circuits introduced in \citeauthor{Voichick_2023} since our construction guarantees the preservation of valid encodings: this property is not stated explicitly in the following, but it is is clear for the circuits in this section as the direct sum (\Cref{app:dirsum_proof}) is the only truly nontrivial construction in terms of encoding validity.

\begin{lemma}
	\label{lem:compile-inj}
	It is always possible to implement the direct sum injections $\msem{\lef{T_0}{T_1}}$ and $\msem{\rit{T_0}{T_1}}$.
\end{lemma}
\begin{proof}
First, one can implement the operator $\lef{T_0}{T_1} : \cH(T_0) \to \cH(T_0) \oplus \cH(T_1)$ using $1 + \max\{\size(T_0), \size(T_1)\} - \size(T_0)$ prep wires and $0$ flag wires with this circuit:
\[
\begin{quantikz}
	\lstick{$\size(T_0)$} & \permute{2,1} & \qw & \rstick{1} \qw \\
	\lstick{1} & & \qwbundle{\size(T_0)} & \rstick[wires=2]{$\max\{\size(T_0),\size(T_1)\}$} \qw \\
	\lstick{$\max\{\size(T_0),\size(T_1)\} - \size(T_0)$} & \qw & \qw & \qw
\end{quantikz}
\]
Using $U$ to represent the qubit-based unitary implemented by this circuit, we can show that it meets the needed criteria:
	\begin{alignat*}{2}
		&&\;& \bra{\enc(\lef{T_0}{T_1} v'), 0} U \ket{\enc(v), 0, 0} \\
		&=&& \bra{0, \enc(v'), 0} U \ket{\enc(v), 0, 0} \\
		&=&& \ip{0, \enc(v'), 0}{0, \enc(v), 0} \\
		&=&& \ip{\enc(v')}{\enc(v)} \\
		&=&& \ip{v'}{v} \\
		&=&& \left(\bra{v'} \oplus 0\right)\left(\ket{v} \oplus 0\right) \\
		&=&& \ip{\lef{T_0}{T_1} v'}{\lef{T_0}{T_1} v} \\
		&&\;& \bra{\enc(\rit{T_0}{T_1} v'), 0} U \ket{\enc(v), 0, 0} \\
		&=&& \bra{1, \enc(v'), 0} U \ket{\enc(v), 0, 0} \\
		&=&& \ip{1, \enc(v'), 0}{0, \enc(v), 0} \\
		&=&& 0 \\
		&=&& \left(0 \oplus \bra{v'}\right)\left(\ket{v} \oplus 0\right) \\
		&=&& \ip{\rit{T_0}{T_1} v'}{\lef{T_0}{T_1} v} \\
	\end{alignat*}

Similarly, one can implement the operator $\rit{T_0}{T_1} : \cH(T_1) \to \cH(T_0) \oplus \cH(T_1)$ using $1 + \max\{\size(T_0), \size(T_1)\} - \size(T_1)$ prep wires and $0$ flag wires like this:
\[
\begin{quantikz}
	\lstick{$\size(T_1)$} & \permute{2,1} & \gate{X} & \rstick{1} \qw \\
	\lstick{1} & & \qwbundle{\size(T_1)} & \rstick[wires=2]{$\max\{\size(T_0),\size(T_1)\}$} \qw \\
	\lstick{$\max\{\size(T_0),\size(T_1)\} - \size(T_1)$} & \qw & \qw & \qw
\end{quantikz}
\]
	\begin{alignat*}{2}
		&&\;& \bra{\enc(\lef{T_0}{T_1} v'), 0} U \ket{\enc(v), 0, 0} \\
		&=&& \bra{0, \enc(v'), 0} U \ket{\enc(v), 0, 0} \\
		&=&& \ip{0, \enc(v'), 0}{1, \enc(v), 0} \\
		&=&& \ip{\enc(v')}{\enc(v)} \\
		&=&& 0 \\
		&=&& \left(\bra{v'} \oplus 0\right)\left(0 \oplus \ket{v}\right) \\
		&=&& \ip{\lef{T_0}{T_1} v'}{\rit{T_0}{T_1} v} \\
		&&\;& \bra{\enc(\rit{T_0}{T_1} v'), 0} U \ket{\enc(v), 0, 0} \\
		&=&& \bra{1, \enc(v'), 0} U \ket{\enc(v), 0, 0} \\
		&=&& \ip{1, \enc(v'), 0}{1, \enc(v), 0} \\
		&=&& \ip{\enc(v')}{\enc(v)} \\
		&=&& \ip{v'}{v} \\
		&=&& \left(0 \oplus \bra{v'}\right)\left(0 \oplus \ket{v}\right) \\
		&=&& \ip{\rit{T_0}{T_1} v'}{\rit{T_0}{T_1} v}
	\end{alignat*}
\end{proof}

When we construct circuits of these operators, we are implicitly using three things: identity operators (represented by bare wires), tensor products (represented by vertically-stacked gates), and operator composition (represented by horizontally-stacked gates).
To justify this form of circuit diagram, we must show that these constructions are always possible.
It is always possible to implement the identity operator $\mathbb{I} : \cH(T) \to \cH(T)$, by using an identity circuit with no prep or flag wires. The next two lemmas demonstrate that tensor products and compositions are possible.

\begin{lemma}
	Suppose it is possible to implement the operators $E_0 : \cH_0 \to \cH_0'$ and $E_1 : \cH_1 \to \cH_1'$.
	Then it is possible to implement the operator $E_0 \otimes E_1 : \cH_0 \otimes \cH_1 \to \cH_0' \otimes \cH_1'$.
\end{lemma}
\begin{proof}
	Assume $E_0$ is implemented by $U_0$ with $p_0$ prep wires and $f_0$ flag wires, and assume that $E_1$ is implemented by $U_1$ with $p_1$ prep wires and $f_1$ flag wires.
	The following qubit circuit $U$ then implements $E_0 \otimes E_1$ with $p_0 + p_1$ prep wires and $f_0 + f_1$ flag wires: 

	\[
	\begin{quantikz}
		\lstick{$\size(\cH_0)$} & \qw & \gate[2]{U_0} & \qw & \rstick{$\size(\cH_0')$} \qw \\
		\lstick{$\size(\cH_1)$} & \permute{2,1} & & \permute{2,1} & \rstick{$\size(\cH_1')$} \qw \\
		\lstick{$p_0$} & & \permute{2,1} & & \rstick{$f_0$} \qw \\
		\lstick{$p_1$} & \qw & & \qw & \rstick{$f_1$} \qw \\
	\end{quantikz}
	\]

	\begin{alignat*}{2}
		&&\;& \bra{\enc(v_0'), \enc(v_1'), 0, 0} U \ket{\enc(v_0), \enc(v_1), 0, 0} \\
		&=&& \bra{\enc(v_0'), 0, \enc(v_1'), 0} (U_0 \otimes U_1) \ket{\enc(v_0), 0, \enc(v_1), 0} \\
		&=&& \bra{\enc(v_0'), 0} U_0 \ket{\enc(v_0), 0} \cdot \bra{\enc(v_1'), 0} U_1 \ket{\enc(v_1), 0} \\
		&=&& \bra{v_0'} E_0 \ket{v_0} \cdot \bra{v_1'} E_1 \ket{v_1} \\
		&=&& \bra{v_0', v_1'} (E_0 \otimes E_1) \ket{v_0, v_1}
	\end{alignat*}
\end{proof}

\begin{lemma}
	\label{lem:compose}
	Suppose it is possible to implement the operators $E_0 : \cH_0 \to \cH'$ and $E_1 : \cH' \to \cH_1$.
	Then it is possible to implement the operator $E_1 E_0 : \cH_0 \to \cH_1$.
\end{lemma}
\begin{proof}
	Assume $E_0$ is implemented by $C(E_0) = U_0$ with $p_0$ prep wires and $f_0$ flag wires, and assume that $E_1$ is implemented by $C(E_1) = U_1$ with $p_1$ prep wires and $f_1$ flag wires.
	Let $B'$ be a basis for $\cH'$.
	The following qubit circuit $U$ then implements $E_1 E_0$ with $p_0 + p_1$ prep wires and $f_0 + f_1$ flag wires: 

	\[
	\begin{quantikz}
		\lstick{$\size(\cH_0)$} & \gate[2]{U_0} & \qwbundle{\size(\cH')} & \gate[2]{U_1} & \rstick{$\size(\cH_1)$} \qw \\
		\lstick{$p_0$} & & \permute{2,1} & & \rstick{$f_1$} \qw \\
		\lstick{$p_1$} & \qw & & \qw & \rstick{$f_0$} \qw
	\end{quantikz}
	\]

	\begin{alignat*}{2}
		&&\;& \bra{\enc(v_1), 0, 0} U \ket{\enc(v_0), 0, 0} \\
		&=&& \bra{\enc(v_1), 0, 0} (U_1 \otimes \mathbb{I}) (\mathbb{I} \otimes \textsc{swap}) (U_0 \otimes \mathbb{I}) \ket{\enc(v_0), 0, 0} \\
		&=&& \left(\bra{\enc(v_1), 0} U_1 \otimes \bra{0}\right) (\mathbb{I} \otimes \textsc{swap}) \left(U_0\ket{\enc(v_0),0} \otimes \ket{0}\right) \\
		&=&& \left(\bra{\enc(v_1), 0} U_1\right) (\mathbb{I} \otimes \op{0}{0}) \left(U_0\ket{\enc(v_0),0}\right) \\
		&=&& \sum_{b \in \{\zero, \one\}^{\size(\cH')}} \bra{\enc(v_1), 0} U_1 \op{b,0}{b,0} U_0\ket{\enc(v_0), 0} \\
		&=&& \sum_{\ket{v'} \in B'} \bra{\enc(v_1), 0} C(E_1) \op{\enc(v'), 0}{\enc(v'), 0} C(E_0) \ket{\enc(v_0), 0} \\
		&=&& \sum_{\ket{v'} \in B'} \bra{v_1} E_1 \op{v'}{v'} E_0 \ket{v_0} \\
		&=&& \bra{v_1} E_1 \left( \sum_{v' \in B'} \op{v'}{v'} \right) E_0 \ket{v_0} \\
		&=&& \bra{v_1} E_1 \mathbb{I} E_0 \ket{v_0} \\
		&=&& \bra{v_1} E_1 E_0 \ket{v_0} \\
	\end{alignat*}

\end{proof}

The adjoint is another useful construction that we will need in~\Cref{app:high_level_compilation}, motivating the following lemma:
\begin{lemma}
	\label{lem:compile-adj}
	Suppose it is possible to implement the operator $E : \cH \to \cH'$.
	Then it is possible to implement the operator $E^\dagger : \cH' \to \cH$.
\end{lemma}
\begin{proof}
	Assume $E$ is implemented by $U$ with $p$ prep wires and $f$ flag wires.
	Given the circuit for $U$, one can construct a circuit for $U^\dagger$ in the standard way: by taking the adjoint of each gate in the circuit and reversing the order.
	This circuit $U^\dagger$ then implements $E^\dagger$ with $f$ prep wires and $p$ flag wires.
	\[
	\begin{quantikz}
		\lstick{$\size(\cH')$} & \gate[2]{U^\dagger} & \rstick{$\size(\cH)$} \qw \\
		\lstick{$f$} & & \rstick{$p$} \qw \\
	\end{quantikz}
	\]
	\begin{alignat*}{2}
		&&\;& \bra{\enc(v'), 0} U^\dagger \ket{\enc(v), 0} \\
		&=&& \left(\bra{\enc(v), 0} U \ket{\enc(v'), 0}\right)^{*} \\
		&=&& \left(\bra{v} E \ket{v'}\right)^{*} \\
		&=&& \bra{v'} E^\dagger \ket{v}
	\end{alignat*}
\end{proof}

\begin{lemma}
	Suppose it is possible to implement the operators $E_0 : \cH_0 \to \cH_0'$ and $E_1 : \cH_1 \to \cH_1'$.
	Then it is possible to implement the operator $E_0 \oplus E_1 : \cH_0 \oplus \cH_1 \to \cH_0' \oplus \cH_1'$.
\end{lemma}
\begin{proof}
	See \Cref{lem:dirsum_possible}.
\end{proof}

It is standard to use the tensor product monoidally, implicitly using vector space isomorphisms $(\cH_1 \otimes \cH_2) \otimes \cH_3 \approx \cH_1 \otimes (\cH_2 \otimes \cH_3)$ and $\CC \otimes \cH \approx \cH \approx \cH \otimes \CC$.
Under the most general definition of the tensor product, these are isomorphisms rather than strict equality \cite[Chapter 14]{advanced-linear-algebra}, but these isomorphisms can typically be used implicitly.
Our encoding ensures that these implicit isomorphisms do not require any low-level gates because $\enc(\pair{\pair{v_1}{v_2}}{v_3}) = \enc(\pair{v_1}{\pair{v_2}{v_3}})$ and $\enc(\pair{\unit}{v}) = \enc(v) = \enc(\pair{v}{\unit})$.

However, the associativity and distributivity isomorphisms for the direct sum \emph{do} correspond to different encodings, so we must show how to implement them. Note that the additive unit isomorphisms were already implemented in Lemma~\ref{lem:compile-inj}. They are $\lef{T}{\Void}$ and $\rit{\Void}{T}$, and their inverses are their adjoints as compiled in Lemma~\ref{lem:compile-adj}. The following lemma demonstrates how to compile the associativity isomorphism:

\begin{lemma}
	Let $T_1$, $T_2$, and $T_3$ be arbitrary types.
	Then it is possible to implement the direct sum's associativity isomorphism $\textsc{assoc} : (\cH(T_1) \oplus \cH(T_2)) \oplus \cH(T_3) \to \cH(T_1) \oplus (\cH(T_2) \oplus \cH(T_3))$.
\end{lemma}
\begin{proof}
	Define the following integers:
	\begin{align*}
		n\subcap{max} &= \max\{\size(T_1), \size(T_2), \size(T_3)\} \\
		p &= n\subcap{max} - \max\{\max\{\size(T_1), \size(T_2)\}, \size(T_3) - 1\} \\
		f &= n\subcap{max} - \max\{\size(T_1) - 1, \max\{\size(T_2), \size(T_3)\}
	\end{align*}
	We can implement ``shifting'' operations \textsc{rsh} and \textsc{lsh} via a series of swap gates that enacts a rotation permutation, so that $\textsc{rsh} \ket{\psi_1, \psi_2, \ldots, \psi_{n-1}, \psi_n} = \ket{\psi_n, \psi_1, \psi_2, \ldots, \psi_{n-1}}$ and $\textsc{lsh} = \textsc{rsh}^{-1}$.
	The following qubit circuit $U$ then implements the associator with the $p$ prep wires and $f$ flag wires.
	\[
	\begin{quantikz}
		\lstick{1} & \ctrl{1} & \qw & \ctrl{1} & \permute{2,1} & \qw & \octrl{1} & \rstick{1} \qw \\
		\lstick{$1 + n\subcap{max}$} & \gate[2]{\textsc{rsh}} & \qwbundle{1} & \targ{} & & \qwbundle{1} & \gate[2]{\textsc{lsh}} & \rstick{$1 + n\subcap{max}$} \qw \\
         & & \qwbundle{n\subcap{max}} & \qw & \qw & \qw & \qw & &
	\end{quantikz}
	\]
	We can show that this circuit has the desired behavior:
	\begin{alignat*}{2}
		&&\;& \ket{\enc\left(\lef{(T_1 \oplus T_2)}{T_3} (\lef{T_1}{T_2} v) \right), 0} \\
		&=&& \ket{0, 0, \enc(v), 0} \\
		&\mapsto^{*}&& \ket{0, \enc(v), 0, 0} \\
		&=&& \ket{\enc(\lef{T_1}{(T_2 \oplus T_3)} v), 0} \\
		&&\;& \ket{\enc\left(\lef{(T_1 \oplus T_2)}{T_3} (\rit{T_1}{T_2} v) \right), 0} \\
		&=&& \ket{0, 1, \enc(v), 0} \\
		&\mapsto^{*}&& \ket{1, 0, \enc(v), 0} \\
		&=&& \ket{\enc\left(\rit{T_1}{(T_2 \oplus T_3)} (\lef{T_2}{T_3} v)\right), 0} \\
		&&\;& \ket{\enc\left(\rit{(T_1 \oplus T_2)}{T_3} v\right), 0} \\
		&=&& \ket{1, \enc(v), 0} \\
		&\mapsto&& \ket{1, 0, \enc(v), 0} \\
		&\mapsto&& \ket{1, 1, \enc(v), 0} \\
		&=&& \ket{\enc\left(\rit{T_1}{(T_2 \oplus T_3)} (\rit{T_2}{T_3} v)\right), 0}
	\end{alignat*}
\end{proof}

We can now implement the monoidal isomorphisms for the tensor product and direct sum.
Both of these are \emph{symmetric} monoidal categories, and the swap maps are straightforward to implement, using swap gates to implement $T_0 \otimes T_1 \cong T_1 \otimes T_0$, and using a single Pauli-X gate on the indicator qubit to implement $T_0 \oplus T_1 \cong T_1 \oplus T_0$.
We will also need \emph{distributivity} of the tensor product over the direct sum, part of the definition of a \emph{bimonoidal} category \cite{symmetric-bimonoidal} (sometimes known as a ``rig category'' \cite{rig-programming}). 

\begin{lemma}
	Let $T$, $T_0$, and $T_1$ be arbitrary types.
	It is possible to implement the distributivity isomorphism $\textsc{distr} : \cH(T) \otimes (\cH(T_0) \oplus \cH(T_1)) \cong (\cH(T) \otimes \cH(T_0)) \oplus (\cH(T) \otimes \cH(T_1))$, which acts like $\ket{v} \otimes \left( \ket{v_0} \oplus \ket{v_1} \right) \mapsto \ket{v, v_0} \oplus \ket{v, v_1}$.
\end{lemma}
\begin{proof}
	This can be done with no prep or flag wires:
	\[
	\begin{quantikz}
		\lstick{$\size(T)$} & \permute{2,1} & \rstick{1} \qw \\
		\lstick{1} & & \rstick{$\size(T)$} \qw \\
		\lstick{$\max\{\size(T_0), \size(T_1)\}$} & \qw & \rstick{$\max\{\size(T_0), \size(T_1)\}$} \qw
	\end{quantikz}
	\]
	It should be clear from the value encoding that this circuit has the correct behavior.
\end{proof}
In the high-level circuits, we will often use some transformation of the \textsc{distr} construction, for example:
\begin{itemize}
	\item its adjoint $(\cH(T) \otimes \cH(T_0)) \oplus (\cH(T) \otimes \cH(T_1)) \cong \cH(T) \otimes (\cH(T_0) \oplus \cH(T_1))$,
	\item its composition with swaps $(\cH(T_0) \oplus \cH(T_1)) \otimes \cH(T) \cong (\cH(T_0) \otimes \cH(T)) \oplus (\cH(T_1) \otimes \cH(T))$,
	\item compositions of distributions $(\cH(T_{00}) \oplus \cH(T_{01})) \otimes (\cH(T_{10}) \oplus \cH(T_{11})) \cong (\cH(T_{00}) \otimes \cH(T_{10})) \oplus (\cH(T_{00}) \otimes \cH(T_{11})) \oplus (\cH(T_{01}) \otimes \cH(T_{10})) \oplus (\cH(T_{01}) \otimes \cH(T_{11}))$.
\end{itemize}
We will denote all of these as ``\textsc{distr},'' but the transformation being applied should always be clear from context.

Additional operators used in Qunity's intermediate representation include ``context partition'' and ``context merge'' operators. We maintain the invariant that when contexts are encoded into qubit registers, the encodings of the variables are stored in lexicographic order. Therefore, implementing the isomorphism $\cH(\Delta_0, \Delta_1) \rarr \cH(\Delta_0) \otimes \cH(\Delta_1)$ may require the use of SWAP gates. We use these operators implicitly in most of the constructions in \Cref{app:high_level_compilation}.

\begin{lemma}[pure error handling]
	\label{lem:pure-error}
	Suppose it is possible to implement a contraction $E : \cH \to \cH'$.
	Then it is possible to implement a norm-preserving operator $E\subcap{f} : \cH \to \cH' \oplus \cH\subcap{f}$ for some ``flag space'' $\cH\subcap{f}$ such that $\msem{\lef{T}{T'}}^\dagger E\subcap{f} = E$.
\end{lemma}
\begin{proof}
	Assume that $E$ is implemented by the unitary $U$ with $p$ prep wires and $f$ flag wires. Then, since we assume encoding validity to be preserved, we can write
	\[
		U\ket{\enc(v), \zero} = \left(\sum_{v' \in \VV(T')} \bra{v'}E\ket{v} \ket{\enc(v'), \zero}\right) + \ket{\psi_{v}}
	\]
	for some $\ket{\psi_v} \in \cH\subcap{f}$ with $(\mathbb{I}\otimes \op{\zero}{\zero})\ket{\psi_v} = 0$.
	We use $E\subcap{f}$ defined such that 
	\[
		E\subcap{f}\ket{v} = \left(\sum_{v' \in \VV(T')} \bra{v'}E\ket{v} \ket{v'}\right) \oplus \ket{\psi_{v}} = E\ket{v} \oplus \ket{\psi_{v}}.
	\]

	This definition ensures that $\msem{\lef{T}{T'}}^\dagger E\subcap{f}\ket{v} = \msem{\lef{T}{T'}}^\dagger \left(E\ket{v} \oplus \ket{\psi_{v}}\right) = E\ket{v}$.
	See that $\norm{E\ket{v}} = \norm{U\ket{\enc{v}, \zero}} = 1$, so $E\subcap{f}$ is norm-preserving.

	The following circuit implements $E\subcap{f}$ with $p+1$ prep wires and no flag wires.
	\[
	\begin{quantikz}
		\lstick{$\size(\cH)$} & \permute{2,1} & \qw & \hphantomgate{} & \targ{} & \rstick{1} \qw \\
		\lstick{1} & & \gate[2]{U} & \qwbundle{\size(\cH')} & \qw &  \rstick{$\size(\cH')$} \qw \\
		\lstick{$p$} & \qw & & \qwbundle{f} & \ctrl{-2} & \rstick{$f$} \qw
	\end{quantikz}
	\]
	\begin{align*}
		\ket{\enc(v), 0, 0} &\mapsto \ket{0, \enc(v), 0} \\
        &\mapsto \ket{0} \otimes U\ket{\enc(v), 0} \\
        &= \left(\sum_{v' \in \VV(T')} \bra{v'}E\ket{v} \ket{\zero, \enc{v'}, \zero}\right) + \ket{\zero, \psi_{v}} \\
        &\mapsto \left(\sum_{v' \in \VV(T')} \bra{v'}E\ket{v} \ket{\zero, \enc{v'}, \zero}\right) + \ket{\one, \psi_{v}} \\
		\bra{0, v', 0} C(E\subcap{f}) \ket{\enc(v), 0, 0} &= \bra{0, v', 0} \left(\sum_{v' \in \VV(T')} \bra{v'}E\ket{v} \ket{\zero, \enc{v'}, \zero}\right) \\
            &= (\bra{v'} \oplus 0) E\subcap{f}\ket{v} \\
		\bra{1, v'} C(E\subcap{f}) \ket{\enc(v), 0, 0} &= \ip{\one, v'}{\one, \psi_{v}} \\
         &= (0 \oplus \bra{v'}) E\subcap{f}\ket{v}
	\end{align*}
\end{proof}

\begin{lemma}[mixed error handling]
	\label{lem:mixed-error}
Suppose it is possible to implement the completely positive trace-non-increasing linear superoperator $\mathcal E \in \cL(\cL(\cH), \cL(\cH'))$.
Then, it is possible to implement a completely positive trace-preserving linear superoperator $\cptp(\mathcal{E}) \in \cL(\cL(\cH), \cL(\cH' \oplus \CC))$ such that for all $\rho \in \cL(\cH)$,
\[
	\cptp(\mathcal{E})(\rho) = \mathcal{E}(\rho) \oplus (\tr(\rho) - \tr(\mathcal{E}(\rho))).
\]
\end{lemma}
\begin{proof}
	Assuming $\mathcal{E}$ is implemented by the unitary $U$ with $p$ prep wires, $f$ flag wires, and $g$ garbage wires, the following circuit achieves the desired result with $p+\size(\cH')+1$ prep wires, no flag wires, and $f + \size(\cH') + g$ garbage wires: 
	\[
	\begin{quantikz}
		\lstick{$\size(\cH)$} & \permute{2,1} & \qw & \hphantomgate{} & \targ{} & \qw & \rstick{1} \qw \\
		\lstick{1} & & \gate[3]{U} & \qwbundle{\size(\cH')} & \qw &  \swap{1} & \rstick{$\size(\cH')$} \qw \\
		\lstick{$p$} & \qw & & \qwbundle{f} & \ctrl{-2} & \ctrl{2} & \rstick{$f$} \qw \\
             &&& \qw & \qw & \qw & \rstick{$g$} \qw \\
		\lstick{$\size(\cH')$} & \qw & \qw & \qw & \qw & \swap{-2} & \rstick{$\size(\cH')$} \qw 
	\end{quantikz}
	\]
	This circuit works by turning flag wires from $\mathcal{E}$ into garbage wires for $\cptp(\mathcal{E})$.
	The first qubit is used as an indicator of failure, and in the event of failure, the output of $U$ is treated as garbage and replaced with a fresh set of $\ket\zero$ qubits from the prep register, as required by the bitstring encoding of sum types.
	Here, the ``control'' on the flag wires should be understood to apply the gate conditioned on \emph{any} of the qubits on the control register being in the $\ket\one$ state, which can be implemented using a control construct conditioned on \emph{all} of the qubits on the control register being in the $\ket\zero$ state followed by an uncontrolled gate.

	To prove that this circuit superoperator $\Ccptp$ correctly implements $\cptp(\cE)$, we must show that
	\[
		\bra{v_1'}\Ccptp(\op{v_1}{v_2})\ket{v_2'} = \sum_b \bra{\enc(v_1'), b} \mathcal{E}(\op{\enc(v_1), \zero}{\enc(v_2), \zero}) \ket{\enc(v_2'), b}
	\]
	for all $v_1, v_2 \in \VV(T), v_1', v_2' \in \VV(T' \oplus \Unit)$.

	We will consider three cases for $v_1'$ and $v_2'$.
	\begin{itemize}
		\item First, suppose both are in the ``success'' (non-error) subspace, and see that
        \[
        (\bra{v_1'} \oplus 0)\cptp(\cE)(\op{v_1}{v_2})(\ket{v_2'} \oplus 0) = \bra{v_1'}\cE(\op{v_1}{v_2})\ket{v_2'}.
        \]
			To see that this equals
            \[
            \sum_b \bra{\zero, \enc(v_1'), b} \Ccptp(\op{\enc(v_1), \zero}{\enc(v_2), \zero}) \ket{\zero, \enc(v_2'), b},
            \]
            see from the circuit that a $\zero$ output on the first wire also implies a $\zero$ output on the $f$ segment of the output wires, as well as the final $\size(\cH')$ section, so we really care about
			\[
				\sum_{b \in \{\zero,\one\}^g} \bra{\zero, \enc(v_1'), \zero, b, \zero} \Ccptp(\op{\enc(v_1), \zero}{\enc(v_2), \zero}) \ket{\zero, \enc(v_2'), \zero, b, \zero}.
			\]
			The requirements placed on $U$ ensure that this equality holds.
		\item Consider the case where exactly one of the two values is in the error subspace, for example $(\bra{v_1'} \oplus 0)\cptp(\cE)(\op{v_1}{v_2})(0 \oplus \ket{v_2'}) = 0$.
			Our circuit works correctly in this case because $\bra{\zero, \enc(v_1'), b} \Ccptp(\op{\enc(v_1), \zero}{\enc(v_2), \zero}) \ket{\one, \enc(v_2'), b}$ is always zero, regardless of $b$.
			To see this, see that the $f$ garbage bits are all zero if and only if the first indicator output bit is zero, so any setting of $b$ would cause one of the two sides of the expression to vanish. In other words, there is no superposition between the error and non-error subspaces because the discarded garbage collapses the state into one of the two.
		\item In the final case, both $v_1'$ and $v_2'$ are in the error subspace, so we must show that
			\[
            \sum_b \bra{\one, \zero, b} \Ccptp(\op{\enc(v_1), \zero}{\enc(v_2), \zero}) \ket{\one, \zero, b} = \tr(\op{v_1}{v_2}) - \tr(\mathcal{E}(\op{v_1}{v_2})),
            \]
            as we are encoding an error value as $\texttt{"1"} \doubleplus \texttt{"0"}^{\size(\cH')}$.
			This case is constrained by the first two, as it is the only possible value that would ensure that $\cptp(\cE)$ is trace-preserving.
			To verify that $\cptp(\cE)$ is trace-preserving, see that there are no flag wires and $U$ cannot output invalid encodings without setting flag qubits, so there is no way for the trace to decrease.
	\end{itemize}
\end{proof}

\begin{lemma}[purification]
	\label{lem:purify}
	Suppose it is possible to implement the completely positive trace-non-increasing linear superoperator $\mathcal E \in \cL(\cL(\cH), \cL(\cH'))$.
	Then, it is possible to implement a contraction $E \in \cL(\cH, \cH' \otimes \cH_{\textsc{g}})$ for some ``garbage'' Hilbert space $\cH_{\textsc{g}}$ with the following property: for any $\rho \in \cL(\cH), \ket{\psi} \in \cH'$, there is some $\ket{g_{\rho,\psi}} \in \cH_{\textsc{g}}$ such that
	\[
		\bra{\psi} \mathcal{E}(\rho) \ket{\psi} = \bra{g_{\rho,\psi}, \psi} E\rho E^\dagger \ket{g_{\rho,\psi}, \psi}
	\]
\end{lemma}
\begin{proof}
	Assuming $\mathcal{E}$ is implemented by the unitary $U$ with $p$ prep wires, $f$ flag wires, and $g$ garbage wires, the following circuit achieves the desired result with $p$ prep wires and $f$ flag wires by setting $\cH\subcap{g} = \CC^{2^{g}}$:
	\[
	\begin{quantikz}
		\lstick{$\size(\cH)$} & \gate[3]{U} & \qw & \rstick{$\size(\cH')$} \qw \\
		\lstick{$p$} & & \permute{2,1} & \rstick{$g$} \qw \\
													 &&& \rstick{$f$} \qw
	\end{quantikz}
	\]
	This circuit simply feeds the existing garbage wires into an additional output.
\end{proof}

Next, we consider some results about how to implement trace-non-increasing superoperators with low-level qubit-based unitaries.

\begin{lemma}
	For any type $T$, it is possible to implement a superoperator $\mathcal{E} : \cL(\cH(T)) \to \cL(\cH(\Unit))$ that computes the trace of its input, effectively discarding it.
\end{lemma}
\begin{proof}
	This gate is implemented with an empty (identity) circuit by setting $p = f = 0$, $g = \size(T)$.
	\[
	\begin{quantikz}
		\lstick{$\size(T)$} & \rstick{$g$} \qw
	\end{quantikz}
	\]
	\[
	\sum_{g \in \{\zero, \one\}^{\size(T)}} \bra{\enc(\unit), g} \mathbb{I} \rho \mathbb{I}^\dagger \ket{\enc(\unit), g} = \tr(\rho)
	\]
\end{proof}

	In \Cref{app:high_level_compilation}, we will also be constructing circuits from these trace-non-increasing superoperators.
	Again, we must justify this by demonstrating that tensor products and function composition are possible.

\begin{lemma}
	Suppose it is possible to implement the superoperators $\mathcal{E}_0 : \cL(\cH_0) \to \cL(\cH_0')$ and $\mathcal{E}_1 : \cL(\cH_1) \to \cL(\cH_1')$.
	Then it is possible to implement the operator $\mathcal{E}_0 \otimes \mathcal{E}_1 : \cL(\cH_0 \otimes \cH_1) \to \cL(\cH_0' \otimes \cH_1')$.
\end{lemma}
\begin{proof}
	Assume $E_0$ is implemented by $U_0$ with $p_0$ prep wires, $f_0$ flag wires, and $g$ garbage wires.
	Assume $E_1$ is implemented by $U_1$ with $p_1$ prep wires, $f_1$ flag wires, and $g$ garbage wires.
	The following qubit circuit $U$ then implements $E_0 \otimes E_1$ with $p = (p_0 + p_1)$ prep wires, $f = (f_0 + f_1)$ flag wires, and $g = (g_0 + g_1)$ garbage wires:

	\[
	\begin{quantikz}
		\lstick{$\size(\cH_0)$} & \qw & \gate[3]{U_0} & \qw & \qw & \rstick{$\size(\cH_0')$} \qw \\
															&&& \qw & \permute{2,1} & \rstick{$\size(\cH_1')$} \qw \\
		\lstick{$\size(\cH_1)$} & \permute{2,1} & & \permute{2,1} & & \rstick{$f_0$} \qw \\
		\lstick{$p_0$} & & \gate[3]{U_1} & & \permute{2,1} & \rstick{$f_1$} \qw \\
        &&& \qw & & \rstick{$g_0$} \qw \\
		\lstick{$p_1$} & \qw & & \qw & \qw & \rstick{$g_1$} \qw
	\end{quantikz}
	\]

	\begin{alignat*}{2}
		&&\;& \sum_{g_0 \in \{\zero, \one\}^{g_0}} \sum_{g_1 \in \{\zero, \one\}^{g_1}} \bra{\enc(v_{0,1}'), \enc(v_{1,1}'), 0, g_0, g_1} U \left(\rho_0 \otimes \rho_1 \otimes \op{\zero}{\zero}\right) \\ &&&\qquad \cdot U^\dagger \ket{\enc(v_{0,2}'), \enc(v_{1,2}'), \zero, g_0, g_1} \\
		&=&& \sum_{g_0 \in \{\zero, \one\}^{g_0}} \sum_{g_1 \in \{\zero, \one\}^{g_1}} \bra{\enc(v_{0,1}'), \zero, g_0, \enc(v_{1,1}'), \zero, g_1} (U_0 \otimes U_1) \left(\rho_0 \otimes \op{\zero}{\zero} \otimes \rho_1 \otimes \op{\zero}{\zero}\right) \\ &&&\qquad \cdot (U_0 \otimes U_1)^\dagger \ket{\enc(v_{0,2}'), \zero, g_0, \enc(v_{1,2}'), \zero, g_1} \\
		&=&& \sum_{g_0 \in \{\zero, \one\}^{g_0}} \bra{\enc(v_{0,1}'), \zero, g_0} U_0 \left(\rho_0 \otimes \op{\zero}{\zero}\right) U_0^\dagger \ket{\enc(v_{0,2}'), \zero, g_0} \\ &&&\qquad \cdot \sum_{g_1 \in \{\zero, \one\}^{g_1}} \bra{\enc(v_{1,1}'), \zero, g_1} U_1 (\rho_1 \otimes \op{\zero}{\zero}) U_1^\dagger \ket{\enc(v_{1,2}'), \zero, g_1} \\
		&=&& \bra{v_{0,1}'} \mathcal{E}_0(\rho_0) \ket{v_{0,2}'} \cdot \bra{v_{1,1}'} \mathcal{E}_1(\rho_1) \ket{v_{1,2}'} \\
		&=&& \bra{v_{0,1}', v_{1,1}'} (\mathcal{E}_0 \otimes \mathcal{E}_1)(\rho_0 \otimes \rho_1) \ket{v_{0,2}', v_{1,2}'} \\
	\end{alignat*}
\end{proof}

\begin{lemma}
	Suppose it is possible to implement the superoperators $\mathcal{E}_0 : \cL(\cH_0) \to \cL(\cH')$ and $\mathcal{E}_1 : \cL(\cH') \to \cL(\cH_1)$.
	Then it is possible to implement the operator $\mathcal{E}_1 \circ \mathcal{E}_0 : \cL(\cH_0) \to \cL(\cH_1)$.
\end{lemma}
\begin{proof}
	Assume $E_0$ is implemented by $U_0$ with $p_0$ prep wires, $f_0$ flag wires, and $g_0$ garbage wires.
	Assume $E_1$ is implemented by $U_1$ with $p_1$ prep wires, $f_1$ flag wires, and $g_1$ garbage wires.
	The following qubit circuit $U$ then implements $E_0 \otimes E_1$ with $p_0 + p_1$ prep wires, $f_0 + f_1$ flag wires, and $g_0 + g_1$ garbage wires:
	\[
		\begin{quantikz}
			\lstick{$\size(\cH_0)$} & \gate[3]{U_0} & \qwbundle{\size(\cH')} & \qw & \qw & \gate[3]{U_1} & \qw & \rstick{$\size(\cH_1)$} \qw \\
			\setwiretype{n} & & \setwiretype{q} \qw & \permute{2,1} & \qw & & \qw & \rstick{$f_1$} \qw \\
			\lstick{$p_0$} & & \permute{2,1} & & \gate[2,swap]{} & \setwiretype{n} & \permute{2,1} & \rstick{$f_0$} \qw \\
			\lstick{$p_1$} & \qw & & \gate[2,swap]{} & \setwiretype{n} & \qw & & \rstick{$g_1$} \qw \\
			\setwiretype{n} & & & & \setwiretype{q} & & & \rstick{$g_0$} \qw
		\end{quantikz}
	\]
	As with the other compositions, we're grouping the flags from the two together and grouping the garbage from the two together. This diagram uses a couple of \textsc{swap} gates that do not correspond to physical gates but are just there to ensure wires in the diagram don't collide.
	This is necessary because of the way that we are using single wires to represent different numbers of qubits, for example the unitaries $U_0$ and $U_1$ would have the same number of input qubits as output qubits.
\end{proof}

\begin{lemma}[Tree Leveling Operator]\label{lem:leveling_operator}
Referring to the notation in \Cref{app:binary_trees}, let $\cR$ be a binary tree of size $n$. Let $h \geq \height(\cR)$, and let $\cH_1, \dots, \cH_n$ be Hilbert spaces corresponding to Qunity types or contexts. It is possible to implement an operator
\[
\LEVEL(\cR; h; \cH_1, \dots, \cH_n) : \bigoplus_{j : \cR} \cH_j \rarr \bigoplus_{j : L(\cR)} \cH_j,
\]
such that
\[
\LEVEL(\cR; h; \cH_1, \dots, \cH_n) \inj_j^\cR \ket{v} = \inj_j^{L(\cR, h)} \ket{v}
\]
\end{lemma}

\begin{proof}
For the case where $\cR = \Leaf$, we can implement this using $h$ prep wires and no flag wires with the following circuit $U$:
\[
\begin{quantikz}
\lstick{$\size(\cH_1)$} & \permute{2,1} & \rstick{$h$} \\
\lstick{$h$} & & \rstick{$\size(\cH_1)$}
\end{quantikz}
\]

\begin{align*}
& \bra{0, \enc(v')} U \ket{\enc(v), 0} = \\
={}& \braket{\enc(v')}{\enc{v}} = \\
={}& \bra{v'} \inj_1^{\cR} \inj_1^{L(\cR, h) \dagger} \inj_1^{L(\cR, h)} \inj_1^{\cR \dagger} \ket{v} = \\
={}& \bra{v'} \LEVEL(\Leaf; h; \cH_1)\adj \LEVEL(\Leaf; h; \cH_1) \ket{v}
\end{align*}

Now, when $\cR = (\cR_0, \cR_1)$, we can simply write
\[
\LEVEL((\cR_0, \cR_1); h; \cH_1, \dots, \cH_n) =
\LEVEL(\cR_0; h - 1; \cH_1, \dots, \cH_{\size(\cR_0)}) \oplus
\LEVEL(\cR_1; h - 1; \cH_{\size(\cR_0) + 1}, \dots, \cH_n),
\]
where we take the direct sum of the operators using the construction in \Cref{app:dirsum_proof}. It is clear that this is correct from \Cref{def:leveled_tree_sum}.
\end{proof}

The tree leveling operator allows us to take advantage of the way the encoding is structured and split a direct sum encoding into separate ``index'' and ``data'' registers, which will be used for constructing the orthogonality circuit in \Cref{app:ortho_compilation}.

\section{Proofs of Correctness for High-Level Compilation}\label{app:high_level_compilation}

In this section, we justify the correctness of the high-level stage of the compilation procedure, reproducing the original proofs from \citeauthor{Voichick_2023}, with modifications from the addition of the new constructs and the changes described in \Cref{sec:high_level_opt}.

\subsection{Erasure Compilation}\label{app:erasure_compilation}

\begin{lemma}[erasure compilation]
	\label{lem:erasure}
	Suppose $(\Gamma, \Gamma_j \partition \Delta, \Delta' \vdash e_j' : T')$ for all $j \in \{1, \ldots, n\}$ and $\erases{T'}(x; e_1', \ldots, e_n')$ is true for all $x \in \dom(\Delta)$.
	Then, one can implement an operator
    \[
    \msem{\erases{T'}(\Delta; e_1', \ldots, e_n')} : \cH(\Delta) \otimes \cH(T') \to \cH(T')
    \]
    with the following behavior for all $\sigma \in \VV(\Gamma), \sigma_j \in \VV(\Gamma_j), \tau \in \VV(\Delta), \tau' \in \VV(\Delta')$:
	\[
		\ket{\tau} \otimes \msem{\sigma, \sigma_j : \Gamma, \Gamma_j \partition \Delta, \Delta' \vdash e_j' : T'} \ket{\tau, \tau'}
		\mapsto \msem{\sigma, \sigma_j : \Gamma, \Gamma_j \partition \Delta, \Delta' \vdash e_j' : T'} \ket{\tau, \tau'}
	\]
\end{lemma}
\begin{proof}
	We construct the circuit by recursing on $\Delta$.
	In the base case, $\Delta = \varnothing$ and an identity operator (empty circuit) suffices.
	Thus, we focus on the inductive case where our context is $(x : T_x, \Delta)$, assuming the inductive hypothesis that $\msem{\erases{T'}(\Delta; e_1', \ldots, e_n')}$ is implementable with the behavior:
	\[
		\ket{\tau} \otimes \msem{\sigma, \sigma_j : \Gamma, \Gamma_j \partition x : T_x, \Delta, \Delta' \vdash e_j' : T'} \ket{x \mapsto v, \tau, \tau'}
		\mapsto \msem{\sigma, \sigma_j : \Gamma, \Gamma_j \partition x : T_x, \Delta, \Delta' \vdash e_j' : T'} \ket{x \mapsto v, \tau, \tau'}
	\]
	(Note that $x$ appears here even though this is the inductive hypothesis, which would normally be free of $x$. This is still a valid induction principle; we are effectively inducting on the number of variables that must be erased.)
	The problem is then reduced to implementing an operator $\msem{\erases{T'}(x; e_1', \ldots, e_n')} : \cH(T_x) \otimes \cH(T') \to \cH(T')$ with the following behavior:
	\[
		\ket{v} \otimes \msem{\sigma, \sigma_j : \Gamma, \Gamma_j \partition x : T_x, \Delta, \Delta' \vdash e_j' : T'} \ket{x \mapsto v, \tau, \tau'}
		\mapsto \msem{\sigma, \sigma_j : \Gamma, \Gamma_j \partition x : T_x, \Delta, \Delta' \vdash e_j' : T'} \ket{x \mapsto v, \tau, \tau'}
	\]
	\begin{quantikz}
		\lstick{$\cH(T_x)$} & \qw & \hphantomgate{} & \gate[3]{\erases{T'}(x; e_1', \ldots, e_n')} \\
		\lstick{$\cH(\Delta)$} & \gate[2]{\erases{T'}(\Delta; e_1', \ldots, e_n')} & & \\
		\lstick{$\cH(T')$} & & \qwbundle{\cH(T')} & & \rstick{$\cH(T')$} \qw
	\end{quantikz}
	\begin{alignat*}{2}
		&&\;& \msem{\erases{T'}(x : T_x, \Delta; e_1', \ldots, e_n')} \\
		&:&& \ket{x \mapsto v, \tau} \otimes \msem{\sigma, \sigma_j : \Gamma, \Gamma_j \partition x : T_x, \Delta, \Delta' \vdash e_j' : T'} \ket{x \mapsto v, \tau, \tau'} \\
		&\mapsto&& \ket{v} \otimes \msem{\sigma, \sigma_j : \Gamma, \Gamma_j \partition x : T_x, \Delta, \Delta' \vdash e_j' : T'} \ket{x \mapsto v, \tau, \tau'} \\
		&\mapsto&& \msem{\sigma, \sigma_j : \Gamma, \Gamma_j \partition x : T_x, \Delta, \Delta' \vdash e_j' : T'} \ket{x \mapsto v, \tau, \tau'} \\
	\end{alignat*}

	The rest of this proof constructs this gate $\msem{\erases{T'}(x; e_1', \ldots, e_n')}$ by induction on the rule used to prove the erasure judgment.

	\textsc{E-Var}:
	In this case, $e_1' = \cdots = e_n' = x$ and $T' = T_x$.
	We know that $\Delta = \Delta' = \varnothing$ because these contexts must be relevant.

	\begin{quantikz}
		\lstick{$\cH(T')$} & \gate[style={cloud}]{} \\
		\lstick{$\cH(T')$} & \ctrl{-1} & \rstick{$\cH(T')$} \qw
	\end{quantikz}
	\begin{alignat*}{2}
		&&\;& \msem{\erases{T'}(x; e_1', \ldots, e_n')} \\
		&:&& \ket{v} \otimes \msem{\sigma, \sigma_j : \Gamma, \Gamma_j \partition x : T' \vdash x : T'} \ket{x \mapsto v} \\
		&=&& \ket{v} \otimes \ket{v} \\
		&\mapsto&& \ket{v} \\
		&=&& \msem{\sigma, \sigma_j : \Gamma, \Gamma_j \partition x : T' \vdash x : T'} \ket{x \mapsto v}
	\end{alignat*}

	\textsc{E-Gphase}:
	In this case, the circuit produced by the inductive hypothesis already has the needed behavior.
	\begin{alignat*}{2}
		&&\;& \msem{\erases{T'}(x; e_1', \ldots, e_{j-1}', e_j' \triangleright \gphase{T}{r}, e_{j+1}', \ldots, e_n')} \\
		&:&& \ket{v} \otimes \msem{\sigma, \sigma_j : \Gamma, \Gamma_j \partition x : T_x, \Delta, \Delta' \vdash e_j' \triangleright \gphase{T}{r} : T'} \ket{x \mapsto v, \tau, \tau'} \\
		&=&& \ket{v} \otimes e^{ir} \msem{\sigma, \sigma_j : \Gamma, \Gamma_j \partition x : T_x, \Delta, \Delta' \vdash e_j' : T'} \ket{x \mapsto v, \tau, \tau'} \\
		&\mapsto&& e^{ir} \msem{\sigma, \sigma_j : \Gamma, \Gamma_j \partition x : T_x, \Delta, \Delta' \vdash e_j' : T'} \ket{x \mapsto v, \tau, \tau'} \\
		&=&& \msem{\sigma, \sigma_j : \Gamma, \Gamma_j \partition x : T_x, \Delta, \Delta' \vdash e_j' \triangleright \gphase{T}{r} : T'} \ket{x \mapsto v, \tau, \tau'}
	\end{alignat*}

	\textsc{E-Ctrl}:
	This rule allows us effectively to ``inline'' the right side of the \texttt{ctrl} expressions for the purpose of the erases judgment.
	Assume one of the expressions is of the following form:
	\[
\cntrl{e}{T}{e_{j,1} &\mapsto e_{j,1}' \\ &\cdots \\ e_{j,m} &\mapsto e_{j,m}'}{T'}
	\]

	Then, we can use the fact that the semantics of \texttt{ctrl} is a linear combination of the semantics of its subexpressions:
	\begin{alignat*}{2}
		&&\;& \msem{\erases{T'}(x; e_1', \ldots, e_{j-1}', \cntrl{e}{T}{\cdots}{T'}, e_{j+1}', \ldots, e_n')} \\
		&:&& \ket{v} \otimes \msem{\sigma, \sigma_j : \Gamma, \Gamma_j \partition x : T_x, \Delta, \Delta' \vdash \cntrl{e}{T}{\cdots}{T'} : T'} \ket{x \mapsto v, \tau, \tau'} \\
		&=&& \sum \cdots \ket{v} \otimes \msem{\sigma, \sigma_{j,k} : \Gamma, \Gamma_{j,k} \partition x : T_x, \Delta, \Delta' \vdash e_{j,k}' : T'} \ket{x \mapsto v, \tau, \tau'} \\
		&\mapsto&& \sum \cdots \msem{\sigma, \sigma_{j,k} : \Gamma, \Gamma_{j,k} \partition x : T_x, \Delta, \Delta' \vdash e_{j,k}' : T'} \ket{x \mapsto v, \tau, \tau'} \\
		&=&& \msem{\sigma, \sigma_j : \Gamma, \Gamma_j \partition x : T_x, \Delta, \Delta' \vdash \cntrl{e}{T}{\cdots}{T'} : T'} \ket{x \mapsto v, \tau, \tau'}
	\end{alignat*}

	\textsc{E-Pair0}:
	\[
	\begin{quantikz}
		\lstick{$\cH(T_x)$} & \gate[2]{\erases{T_0}(x; e_{0,1}, \ldots, e_{0,n})} \\
		\lstick{$\cH(T_0)$} & & \rstick{$\cH(T_0)$} \qw \\
		\lstick{$\cH(T_1)$} & \qw & \rstick{$\cH(T_1)$} \qw
	\end{quantikz}
\]
	\begin{alignat*}{2}
		&&\;& \msem{\erases{T_0 \otimes T_1}(x; \pair{e_{0,1}}{e_{1,1}}, \ldots, \pair{e_{0,n}}{e_{1,n}})} \\
		&:&& \ket{v} \otimes \msem{\sigma, \sigma_j : \Gamma, \Gamma_j \partition x : T_x, \Delta, \Delta' \vdash \pair{e_{0,j}}{e_{1,j}} : T_0 \otimes T_1} \ket{x \mapsto v, \tau, \tau'} \\
		&=&& \ket{v} \otimes \msem{\sigma, \sigma_j : \Gamma, \Gamma_j \partition x : T_x, \Delta_{*}, \Delta_0, \Delta_1, \Delta_{*}', \Delta_0', \Delta_1' \vdash \pair{e_{0,j}}{e_{1,j}} : T_0 \otimes T_1} \ket{x \mapsto v, \tau_{*}, \tau_0, \tau_1, \tau_{*}', \tau_0', \tau_1'} \\
		&=&& \ket{v} \otimes \msem{\sigma, \sigma_j : \Gamma, \Gamma_j \partition x : T_x, \Delta_{*}, \Delta_0, \Delta_{*}', \Delta_0' \vdash e_{0,j} : T_0} \ket{x \mapsto v, \tau_{*}, \tau_0, \tau_{*}', \tau_0'} \\ &&&\otimes \msem{\sigma, \sigma_j : \Gamma, \Gamma_j \partition x : T_x, \Delta_{*}, \Delta_1, \Delta_{*}', \Delta_1' \vdash e_{1,j} : T_1} \ket{x \mapsto v, \tau_{*}, \tau_1, \tau_{*}', \tau_1'} \\
		&\mapsto&& \msem{\sigma, \sigma_j : \Gamma, \Gamma_j \partition x : T_x, \Delta_{*}, \Delta_0, \Delta_{*}', \Delta_0' \vdash e_{0,j} : T_0} \ket{x \mapsto v, \tau_{*}, \tau_0, \tau_{*}', \tau_0'} \\ &&&\otimes \msem{\sigma, \sigma_j : \Gamma, \Gamma_j \partition x : T_x, \Delta_{*}, \Delta_1, \Delta_{*}', \Delta_1' \vdash e_{1,j} : T_1} \ket{x \mapsto v, \tau_{*}, \tau_1, \tau_{*}', \tau_1'} \\
		&=&& \msem{\sigma, \sigma_j : \Gamma, \Gamma_j \partition x : T_x, \Delta, \Delta' \vdash \pair{e_{0,j}}{e_{1,j}} : T_0 \otimes T_1} \ket{x \mapsto v, \tau, \tau'}
	\end{alignat*}

	\textsc{E-Pair1}:
	\begin{quantikz}
		\lstick{$\cH(T_x)$} & \qw & \gate[2]{\erases{T_1}(x; e_{1,1}, \ldots, e_{1,n})} \\
		\lstick{$\cH(T_0)$} & \permute{2,1} & & \permute{2,1} & \rstick{$\cH(T_0)$} \qw \\
		\lstick{$\cH(T_1)$} & & \qw & & \rstick{$\cH(T_1)$} \qw
	\end{quantikz}
	\begin{alignat*}{2}
		&&\;& \msem{\erases{T_0 \otimes T_1}(x; \pair{e_{0,1}}{e_{1,1}}, \ldots, \pair{e_{0,n}}{e_{1,n}})} \\
		&:&& \ket{v} \otimes \msem{\sigma, \sigma_j : \Gamma, \Gamma_j \partition x : T_x, \Delta, \Delta' \vdash \pair{e_{0,j}}{e_{1,j}} : T_0 \otimes T_1} \ket{x \mapsto v, \tau, \tau'} \\
		&=&& \ket{v} \otimes \msem{\sigma, \sigma_j : \Gamma, \Gamma_j \partition x : T_x, \Delta_{*}, \Delta_0, \Delta_1, \Delta_{*}', \Delta_0', \Delta_1' \vdash \pair{e_{0,j}}{e_{1,j}} : T_0 \otimes T_1} \ket{x \mapsto v, \tau_{*}, \tau_0, \tau_1, \tau_{*}', \tau_0', \tau_1'} \\
		&=&& \ket{v} \otimes \msem{\sigma, \sigma_j : \Gamma, \Gamma_j \partition x : T_x, \Delta_{*}, \Delta_0, \Delta_{*}', \Delta_0' \vdash e_{0,j} : T_0} \ket{x \mapsto v, \tau_{*}, \tau_0, \tau_{*}', \tau_0'} \\ &&&\otimes \msem{\sigma, \sigma_j : \Gamma, \Gamma_j \partition x : T_x, \Delta_{*}, \Delta_1, \Delta_{*}', \Delta_1' \vdash e_{1,j} : T_1} \ket{x \mapsto v, \tau_{*}, \tau_1, \tau_{*}', \tau_1'} \\
		&\mapsto&& \msem{\sigma, \sigma_j : \Gamma, \Gamma_j \partition x : T_x, \Delta_{*}, \Delta_0, \Delta_{*}', \Delta_0' \vdash e_{0,j} : T_0} \ket{x \mapsto v, \tau_{*}, \tau_0, \tau_{*}', \tau_0'} \\ &&&\otimes \msem{\sigma, \sigma_j : \Gamma, \Gamma_j \partition x : T_x, \Delta_{*}, \Delta_1, \Delta_{*}', \Delta_1' \vdash e_{1,j} : T_1} \ket{x \mapsto v, \tau_{*}, \tau_1, \tau_{*}', \tau_1'} \\
		&=&& \msem{\sigma, \sigma_j : \Gamma, \Gamma_j \partition x : T_x, \Delta, \Delta' \vdash \pair{e_{0,j}}{e_{1,j}} : T_0 \otimes T_1} \ket{x \mapsto v, \tau, \tau'}
	\end{alignat*}

	We have thus demonstrated that a circuit with this semantics can always be constructed.
\end{proof}

\subsection{The New Orthogonality Circuit}\label{app:ortho_compilation}

\begin{lemma}
Suppose that $\ortho{T}{e_1, \dots, e_n}$ holds, with tree structure $\cR$ (\Cref{def:ortho_tree}) and $n > 0$. Suppose that each $e_j$ is typed using the pure expression typing judgment with no classical context and quantum context $\Delta_j$. Then, it is possible to construct
\[
\msem{\ortho{T}{e_1, \dots, e_n}} : \cH(T) \rarr \bigoplus_{j : \cR} \cH(\Delta_j),
\]
such that
\[
\msem{\ortho{T}{e_1, \dots, e_n}}\msem{e_j}\ket{\tau_j} = \inj_j^\cR \ket{\tau_j}.
\]
\end{lemma}

\begin{proof}
For \textsc{O-Void} and \textsc{O-Unit}, this operator is simply the identity circuit on an empty register. For \textsc{O-Var}, it is also the identity circuit, since it maps $\ket{v} \mapsto \ket{x \mapsto v}$.

For \textsc{O-IsoApp}, we define
\[
\msem{\ortho{T'}{f e_1, \dots, f e_n}} =
\msem{\ortho{T}{e_1, \dots, e_n}} \msem{f}\adj.
\]
Then, we have that
\begin{align*}
& \msem{\ortho{T'}{f e_1, \dots, f e_n}}\msem{f e_j}\ket{\tau_j} =
\msem{\ortho{T}{e_1, \dots, e_n}} \msem{f}\adj \msem{f} \msem{e_j} \ket{\tau_j} = \\
={}& \msem{\ortho{T}{e_1, \dots, e_n}} \msem{e_j} \ket{\tau_j} =
\inj_j^\cR \ket{\tau_j}.
\end{align*}
Here, we used the fact that $f$ is an isometry to say that $\msem{f}\adj \msem{f}$ is the identity.

For \textsc{O-Sum}, we define
\[
\msem{\ortho{T_0 \oplus T_1}{\lef{T_0}{T_1}e_1, &\ldots, \lef{T_0}{T_1}e_n, \\ \rit{T_0}{T_1}e_1', &\ldots, \rit{T_0}{T_1}e_{n'}'}} =
\msem{\ortho{T_0}{e_1, \dots, e_n}} \oplus \msem{\ortho{T_1}{e'_1, \dots, e'_{n'}}}.
\]
Suppose that the tree structures corresponding to $\ortho{T_0}{e_1, \dots, e_n}$ and $\ortho{T_1}{e'_1, \dots, e'_{n'}}$ are $\cR_0$ and $\cR_1$. This circuit acts as:
\begin{align*}
& \msem{\lef{T_0}{T_1} e_j}\ket{\tau_j}
= \msem{e_j}\ket{\tau_j} \oplus 0
\mapsto \inj_j^{\cR_0} \ket{\tau_j} \oplus 0^{\oplus \cR_1}
= \inj_j^{(\cR_0, \cR_1)} \ket{\tau_j} \\
& \msem{\rit{T_0}{T_1} e_j'}\ket{\tau_j'}
= 0 \oplus \msem{e_j'}\ket{\tau_j'}
\mapsto 0^{\oplus \cR_1} \oplus \inj_j^{\cR_0} \ket{\tau_j}
= \inj_{\size(\cR_0) + j}^{(\cR_0, \cR_1)} \ket{\tau_j}.
\end{align*}

Now, for \textsc{O-Pair}: suppose that we have $\ortho{T_0}{e_1, \dots, e_n}$ holds with tree structure $\cR_0$ where each $e_j$ is typed with quantum context $\Delta_j$, and for each $j$, we have $\ortho{T_1}{e_{j,1}', \dots, e_{j,n_j}'}$ with tree structure $\cR_1^j$, where each $e_{j,k}$ is typed with quantum context $\Delta_{j,k}'$. The tree structure of the combined sequence of pairs will be $\cR$, which is obtained by replacing the $j$th leaf of $\cR_0$ with a copy of $\cR_1^j$. Now, the operator is defined by the following circuit:

\resizebox{\textwidth}{!}{
\shortstack{
\begin{quantikz}
\setwiretype{n} &&&&& \gate[2]{\LEVEL(\cR_0)} & \setwiretype{q} \qwbundle{\CC^{\oplus L(\cR_0)}} & \permute{2,1} &&& \rstick{$\cdots$} \\
\lstick{$\cH(T_0)$} & \gate{\msem{\ortho{T_0}{e_1, \dots, e_n}}} & \qwbundle{\bigoplus_{j : \cR_0} \cH(\Delta_j)} &&&&&& \gate[2]{\bigoplus_{j:L(\cR_0)} \msem{\ortho{T_1}{e'_{j,1}, \dots, e'_{j, n_j}}}} & \qwbundle{\bigoplus_{j:L(\cR_0)}\bigoplus_{k:L(\cR_1^j)}\cH(\Delta_{j,k})} & \rstick{$\cdots$} \\
\lstick{$\cH(T_1)$} &&&&&&&& 
\end{quantikz}
\\ \vspace{20pt} \\
\begin{quantikz}
\lstick{$\cdots$} && \permute{2,1} & \gate[3]{\bigoplus_{j:L(\cR_0)} \LEVEL\adj(\cR_1^j, h)} & \gate[2]{\LEVEL\adj(\cR_0)} & \gate{\bigoplus_{j:\cR_0}\bigoplus_{k:\cR_1^j}\FINALMERGE_{j,k}} & \rstick{$\bigoplus_{j:\cR} \cH(\Delta_j)$} \\
\lstick{$\cdots$} & \gate[2]{\bigoplus_{j:L(\cR_0)} \LEVEL(\cR_1^j, h)} &&& \\
\setwiretype{n} && \setwiretype{q} &
\end{quantikz}
}}
where we define $h = \max_j\{\height(\cR_1^j)\}$, we define $\FINALMERGE_{j,k}$ (with a diagram drawn in the style of low-level circuits) as:

\begin{quantikz}
\lstick{$\size(\Delta_j)$} && \gate[2]{\MERGE(\Delta_j, \Delta_{j,k}')} \\
\lstick{$\max_{j'} \parens{\size(\Delta_{j'}) - \size(\Delta_j)}$} & \permute{2,1} && \\
\lstick{$\size(\Delta_{j,k}')$} &&& \\
\lstick{$\max_{j',k'} \parens{\size(\Delta_{j',k'}') - \size(\Delta_{j,k}')}$} &&&
\end{quantikz}

where the bottom two wires are marked as flag registers, since they contain zero-padding regions from direct sum encodings.

The $\LEVEL$ operator (\Cref{lem:leveling_operator}) is designed to separate a potentially uneven direct sum encoding structure into a separate \emph{index register} and \emph{data register}. This allows us to concatenate the tree path from $\cR_0$ and $\cR_1^j$ into a single path in the combined tree, then merge the pieces of data from the two sides of the pair and put them in the correct place in $\cR$. The action of the circuit proceeds as follows:

\begin{enumerate}
\item First, applying $\ortho{T_0}{e_1, \dots, e_n}$, we obtain the direct sum of the $\Delta_j$ over $\cR_0$.
\item After applying the leveling operator, we are able to partition our quantum register into two: the index register contains only the information about which path was taken on the tree (with trailing zeros where appropriate due to different-sized branches), and the data register corresponds to each of the $\Delta_j$ (with trailing zeros where appropriate due to different-sized contexts).
\item We take a direct sum of the operators $\msem{\ortho{T_1}{e'_{j,1}, \dots, e'_{j, n_j}}}$ over the leveled tree $L(\cR_0)$ and then take the direct sum of the leveling operators for each $\cR_1^j$, where the leveling is done to the maximum height of all the $\cR_1^j$. This creates a ``stack'' of the index registers of the trees: the first block of qubits encodes a branch in $\cR_0$, possibly followed by some zeros, then the next block encodes a branch in $\cR_1^j$ for the $j$ corresponding to the path taken in the first block, possibly followed by more trailing zeros. The data register then contains all the $\Delta_{j,k}'$ starting at the same position.
\item We put the stacked index registers on top and stack the data registers below them. Now, we need to eliminate the zeros interspersed between encodings and undo the separation into blocks. We apply the direct sum of the adjoints of the leveling operators for the $\cR_1^j$, which removes the zeros between the index encodings and the combined data of $\Delta_j$ and $\Delta_{j,k}'$. Note that these leveling operators are not exactly the same as the ones before, since the data registers are now larger (we omit this in the notation in the circuit diagram).
\item We then apply the adjoint of the leveling operator for $\cR_0$. We now have a sum over the correct structure $\cR$, but each leaf corresponds to two blocks of contexts, possibly separated by zeros.
\item Applying the direct sum of the $\FINALMERGE$ operators, we combine all the contexts in the correct way, obtaining the desired result.
\end{enumerate}

Finally, for \textsc{O-Sub}: Since we assume $n > 0$ (we will treat $n = 0$ as a special case for the control flow constructs), at least one of the expressions before the application of \textsc{O-Sub} must be kept. If $\cR = \Leaf$, the circuit then must be the identity. Now, we define the circuit $U_\cR$ recursively in terms of subtrees and selected indices. Suppose that $\cR = (\cR_0, \cR_1)$. We can associate these subtrees with types $T_0, T_1$ that are formed as sum types taken along the binary trees, with the contexts at the leaves treated as product types (the encoding is identical). In the case where $\cR_0$ has no selected indices, we define $U_\cR = \msem{\rit{T_0}{T_1}}\adj$, and if $\cR_1$ has no selected indices, we define $U_\cR = \msem{\lef{T_0}{T_1}}\adj$. Otherwise, we just define $U_\cR = U_{\cR_0} \oplus U_{\cR_1}$. It is clear that this operator specifically removes those subtrees that only contain discarded expressions.

\end{proof}

While the above circuit appears complicated, it is more efficient than the existing solution in many cases. In fact, if the expressions contain no variables in the first elements of the pairs, this entire circuit actually simplifies to the identity. This is because in that case, we do not need to do anything to combine the indices of the two trees, as the encodings already correspond to the desired structure.

\subsection{Qunity Typing Judgment Compilation}\label{app:typing_judgment_compilation}

Here we give the compiled circuits for all of the typing judgment cases and demonstrate algebraically that they are correct. In the following circuit diagrams, the ``controlled cloud'' is a ``share'' gate, implemented simply as a series of CNOT gates between the given register and an ancilla register. Additionally, we put control symbols on wires associated with classical contexts, to indicate that any interaction with these wires only uses such share gates.

\textsc{T-Gate}: We assume that our low-level circuits contain these gates as primitives.
\[
\begin{quantikz}
	\lstick{$\cH(\Bit)$} & \gate{\begin{matrix} \cos(r_\theta / 2) & -e^{i r_\lambda} \sin(r_\theta / 2) \\ e^{i r_\phi} \sin(r_\theta / 2) & e^{i(r_\phi + r_\lambda)} \cos(r_\theta / 2) \end{matrix}} & \rstick{$\cH(\Bit)$} \qw
\end{quantikz}
\]
\begin{alignat*}{3}
	&&\;& \ket{\zero} &&\in \cH(\Bit) \\
	&\mapsto&& \cos(r_\theta / 2) \ket{\zero} + e^{i r_\phi} \sin(r_\theta / 2) \ket{\one} &&\in \cH(\Bit) \\
	&&\;& \ket{\one} &&\in \cH(\Bit) \\
	&\mapsto&& -e^{i r_\lambda} \sin(r_\theta / 2) \ket{\zero} + e^{i(r_\phi + r_\lambda)} \cos(r_\theta / 2) \ket{\one} &&\in \cH(\Bit)
\end{alignat*}

\textsc{T-Left}: This was already implemented in \Cref{app:low_level_compilation}.
\[
\begin{quantikz}
	\lstick{$\cH(T_0)$} & \gate{\lef{T_0}{T_1}} & \rstick{$\cH(T_0 \oplus T_1)$} \qw
\end{quantikz}
\]
\begin{alignat*}{3}
	&&\;& \ket{v} &&\in \cH(T_0) \\
	&\mapsto&& \ket{v} \oplus 0 \\
\end{alignat*}

\textsc{T-Right}: This was already implemented in \Cref{app:low_level_compilation}.
\[
\begin{quantikz}
	\lstick{$\cH(T_1)$} & \gate{\rit{T_0}{T_1}} & \rstick{$\cH(T_0 \oplus T_1)$} \qw
\end{quantikz}
\]
\begin{alignat*}{3}
	&&\;& \ket{v} &&\in \cH(T_1) \\
	&\mapsto&& 0 \oplus \ket{v} \\
\end{alignat*}

\textsc{T-PureAbs}:
\begin{quantikz}
	\lstick{$\cH(T)$} & \gate{e^\dagger} & \qwbundle{\cH(\Delta)} & \gate{e'} &\rstick{$\cH(T')$} \qw
\end{quantikz}
\begin{alignat*}{3}
	&&\;& \ket{v} &&\in \cH(T) \\
	&\mapsto&& \msem{e}^\dagger \ket{v} &&\in \cH(\Delta) \\
	&\mapsto&& \msem{e'} \msem{e}^\dagger \ket{v} &&\in \cH(T') \\
\end{alignat*}

\textsc{T-Rphase}:
Let $E\subcap{f} : \cH(T) \to \cH(\Delta) \oplus \cH\subcap{f}$ be the norm-preserving operator constructed from $\msem{\varnothing \partition \Delta \vdash e : T}^\dagger$ using Lemma~\ref{lem:pure-error}.
The compiled \textsc{T-Rphase} circuit then looks like this:
\[
\begin{quantikz}
	\lstick{$\cH(T)$} & \gate{E\subcap{f}} & \qwbundle{\cH(\Delta) \oplus \cH\subcap{f}} & \hphantomgate{ex} & \gate{e^{i r'} \mathbb{I}_\Delta \oplus e^{i r} \mathbb{I}\subcap{f}} & \qwbundle{\cH(\Delta) \oplus \cH\subcap{f}} & \hphantomgate{ex} & \gate{E\subcap{f}^\dagger} &\rstick{$\cH(T)$} \qw
\end{quantikz}
\]
Using the fact that $\msem{\varnothing \partition \Delta \vdash e : T}^\dagger = \msem{\lef{\Delta}{\textsc{G}}}^\dagger E\subcap{f}$,
\begin{align*}
	& E\subcap{f}^\dagger \left(e^{i r} \mathbb{I}_\Delta \oplus e^{i r'} \mathbb{I}\subcap{f}\right) E\subcap{f} \\
	=&\; E\subcap{f}^\dagger \left(e^{i r} \msem{\lef{\Delta}{\textsc{G}}} \msem{\lef{\Delta}{\textsc{G}}}^\dagger + e^{i r'} \msem{\rit{\Delta}{\textsc{G}}}\msem{\rit{\Delta}{\textsc{G}}}^\dagger\right) E\subcap{f} \\
	=&\; E\subcap{f}^\dagger \left(e^{i r} \msem{\lef{\Delta}{\textsc{G}}} \msem{\lef{\Delta}{\textsc{G}}}^\dagger + e^{i r'} (\mathbb{I} - \msem{\lef{\Delta}{\textsc{G}}} \msem{\lef{\Delta}{\textsc{G}}}^\dagger)\right) E\subcap{f} \\
	=&\; e^{i r} E\subcap{f}^\dagger \msem{\lef{\Delta}{\textsc{G}}} \msem{\lef{\Delta}{\textsc{G}}}^\dagger E\subcap{f} + e^{i r'} \left(\mathbb{I} - E\subcap{f}^\dagger \msem{\lef{\Delta}{\textsc{G}}} \msem{\lef{\Delta}{\textsc{G}}}^\dagger E\subcap{f}\right) \\
	=&\; e^{i r} \msem{\varnothing \partition \Delta \vdash e : T} \msem{\varnothing \partition \Delta \vdash e : T}^\dagger + e^{i r'} \left(\mathbb{I} - \msem{\varnothing \partition \Delta \vdash e : T} \msem{\varnothing \partition \Delta \vdash e : T}^\dagger\right) \\
\end{align*}

\textsc{T-Pmatch}:

\resizebox{\textwidth}{!}{
\begin{quantikz}
\lstick{$\cH(T)$} & \gate{\msem{\ortho{T}{e_1, \dots, e_n}}} & \qwbundle{\bigoplus_{j:\cR_0} \cH(\Delta_j)} &&& \gate{\textsc{TreeRearrange}(\cR_0, \cR_1)} & \qwbundle{\bigoplus_{j:\cR_1} \cH(\Delta_j)} &&& \gate{\msem{\ortho{T}{e_1, \dots, e_n}}\adj} & \rstick{$\cH(T')$}
\end{quantikz}
}

This is explained in \Cref{sec:pmatch_compilation}.

\textsc{T-Channel}:
\begin{quantikz}
	\lstick{$\cH(T)$} & \gate{f} & \rstick{$\cH(T')$} \qw
\end{quantikz}
\begin{alignat*}{3}
	&&\;& \op{v}{v'} &&\in \cL(\cH(T)) \\
	&\mapsto&& \msem{f} \op{v}{v'} \msem{f}^\dagger &&\in \cL(\cH(T')) \\
\end{alignat*}

\textsc{T-MixedAbs}:
\begin{quantikz}
	\lstick{$\cH(T)$} & \gate{e^\dagger} & \qwbundle{\cH(\Delta)} & \gate{e'} & \rstick{$\cH(T')$} \qw
\end{quantikz}
\begin{alignat*}{3}
	&&\;& \op{v}{v'} &&\in \cL(\cH(T)) \\
	&\mapsto&& \msem{e}^\dagger \op{v}{v'} \msem{e} &&\in \cL(\cH(\Delta)) \\
	&\mapsto&& \msem{e'} \left(\msem{e}^\dagger \op{v}{v'} \msem{e}\right) &&\in \cL(\cH(T'))
\end{alignat*}

\textsc{T-Unit}:
\begin{quantikz}
	\lstick{$\cH(\Gamma)$} & \rstick{$\cH(\Gamma)$} \qw \\
	\lstick{$\cH(\varnothing)$} & \rstick{$\cH(\Unit)$} \\
\end{quantikz}
\begin{alignat*}{2}
	&&\;& \ket{\sigma, \varnothing} \in \cH(\Gamma) \otimes \cH(\varnothing) \\
	&=&& \ket{\sigma, \unit} \in \cH(\Gamma) \otimes \cH(\Unit)
\end{alignat*}

\textsc{T-Cvar}:
\begin{quantikz}
	\lstick{$\cH(\Gamma)$} & \qw & \rstick{$\cH(\Gamma)$} \qw \\
	\lstick{$\cH(x : T)$} & \ctrl{2} & \rstick{$\cH(x : T)$} \qw \\
	\lstick{$\cH(\Gamma')$} & \qw & \rstick{$\cH(\Gamma')$} \qw \\
	\lstick{$\cH(\varnothing)$} \setwiretype{n} & \gate[style={cloud}]{} & \setwiretype{q} \rstick{$\cH(T)$} \qw
\end{quantikz}
\begin{alignat*}{2}
	&&\;& \ket{\sigma, x \mapsto v, \sigma'} \in \cH(\Gamma, x : T, \Gamma') \\
	&\mapsto&& \ket{\sigma, x \mapsto v, \sigma', v} \in \cH(\Gamma, x : T, \Gamma') \otimes \cH(T)
\end{alignat*}

\textsc{T-Qvar}:
\begin{quantikz}
	\lstick{$\cH(\Gamma)$} & \rstick{$\cH(\Gamma)$} \qw \\
	\lstick{$\cH(x : T)$} & \rstick{$\cH(T)$} \qw
\end{quantikz}
\begin{alignat*}{2}
	&&\;& \ket{\sigma, x \mapsto v} \in \cH(\Gamma, x : T) \\
	&=&& \ket{\sigma, v} \in \cH(\Gamma) \otimes \cH(T)
\end{alignat*}

\textsc{T-PurePair}:
\begin{quantikz}
	\lstick{$\cH(\Gamma)$} & \qw & \qw & \ctrl{1} & \ctrl{3} & \qw \rstick{$\cH(\Gamma)$} \\
	\lstick{$\cH(\Delta)$} & \ctrl{2} & \qwbundle{\cH(\Delta)} & \gate[2]{{\Gamma \partition \Delta, \Delta_0 \vdash e_0 : T_0}} & \qw & \qw \rstick{$\cH(T_0)$} \\
	\lstick{$\cH(\Delta_0)$} & \qw & \qw & & \setwiretype{n} \\
     \setwiretype{n} & \gate[style={cloud}]{} & \setwiretype{q} \qwbundle{\cH(\Delta)} & \qw & \gate[2]{{\Gamma \partition \Delta, \Delta_1 \vdash e_1 : T_1}} & \qw \rstick{$\cH(T_1)$} \\
	\lstick{$\cH(\Delta_1)$} & \qw & \qw & \qw & & \setwiretype{n}
\end{quantikz}
\begin{alignat*}{3}
	&&& \ket{\sigma, \tau, \tau_0, \tau_1} &&\in \cH(\Gamma, \Delta, \Delta_0, \Delta_1) \\
	&\mapsto&\;& \ket{\sigma, \tau, \tau_0, \tau, \tau_1} &&\in \cH(\Gamma, \Delta, \Delta_0, \Delta, \Delta_1) \\
	&\mapsto&\;& \ket{\sigma} \otimes \msem{\sigma : \Gamma \partition \Delta, \Delta_0 \vdash e_0 : T_0} \ket{\tau, \tau_0} \otimes \ket{\tau, \tau_1} &&\in \cH(\Gamma) \otimes \cH(T_0) \otimes \cH(\Delta, \Delta_1) \\
	&\mapsto&\;& \ket{\sigma} \otimes \msem{\sigma : \Gamma \partition \Delta, \Delta_0 \vdash e_0 : T_0} \ket{\tau, \tau_0} \otimes \msem{\sigma : \Gamma \partition \Delta, \Delta_1 \vdash e_1 : T_1} \ket{\tau, \tau_1} &&\in \cH(\Gamma) \otimes \cH(T_0 \otimes T_1)
\end{alignat*}

\textsc{T-Ctrl}:

In the special case where $n = 0$, the semantics corresponds to $0 \in \cL(\cH(\Delta, \Delta')) \rarr \cH(\Void)$, so this may be implemented by a circuit that sends all input qubits to the flag register.

Now, consider $n > 0$. The compiled circuit for \textsc{T-Ctrl} uses modified versions of its subcircuits. We use Lemma~\ref{lem:purify} to get a purified version of the circuit for $e$ with semantics $\msem{e} : \cH(\Gamma, \Delta) \to \cH(T) \otimes \cH_{\textsc{g}}$.
Here, $\cH_{\textsc{g}}$ is some ``garbage'' Hilbert space containing vectors
\[
\braces{\ket{g_{\sigma, \tau, v}} : \sigma \in \VV(\Gamma), \tau \in \VV(\Delta), v \in \VV(T)}
\]
such that
\[
\msem{e}\ket{\sigma, \tau} = \sum_{v \in \VV(T)} \bra{g_{\sigma,\tau,v}, v} \msem{e}\ket{\sigma, \tau} \cdot \ket{g_{\sigma, \tau,v}, v}
\]
for all $\sigma, \tau, v$.

The circuit below is too large to fit on a single page, so the dots denote where the two pieces must fit together. All direct sums below are to be understood as being taken over the tree $\cR$ associated with the orthogonality judgment (\Cref{def:ortho_tree}). We use the orthogonality circuit from \Cref{app:ortho_compilation}.

\[
\begin{quantikz}
	\lstick{$\cH(\Gamma)$} & \qw & \qwbundle{\cH(\Gamma)} & \ctrl{1} & \qw & \qw & \qw & \qw & \qw & \ctrl{3} & \qwbundle{\cH(\Gamma)} \rstick{$\cdots$} \\
    \setwiretype{n} & & & \gate[2]{e} & \setwiretype{q} \qwbundle{\cH_{\textsc{g}}} & \qw & \qw & \qw & \qw & \qw & \qwbundle{\cH_{\textsc{g}}} \rstick{$\cdots$} \\
    \setwiretype{n} & \gate[style={cloud}]{} & \setwiretype{q} \qwbundle{\cH(\Delta)} & & \qwbundle{\cH(T)} & \gate{\msem{\ortho{T}{e_1, \dots, e_n}}} & \qwbundle{\substack{\bigoplus_j \cH(\Gamma_j) \\ \hfill}} & \gate[3]{\textsc{distr}} \\
	\lstick{$\cH(\Delta)$} & \ctrl{-1} & \qwbundle{\cH(\Delta)} & \qw & \qw & \qw & \qw & & \qwbundle{\substack{\bigoplus_j \cH(\Gamma_j, \Delta, \Delta') \\ \\ \hfill}} & \gate{\bigoplus_j \msem{e_j'}} & \qwbundle{\bigoplus_j (\cH(\Gamma_j) \otimes \cH(T'))} \rstick{$\cdots$} \\
	\lstick{$\cH(\Delta')$} & \qw & \qw & \qw & \qw & \qw & \qw & \\
\end{quantikz}
\]
\[
\begin{quantikz}
	\lstick{$\cdots$} & \qwbundle{\cH(\Gamma)} & \qw & \qw & \qw & \qw & \ctrl{2} & \qw & \qw & \qw \rstick{$\cH(\Gamma)$} \\
	\lstick{$\cdots$} & \qwbundle{\cH_{\textsc{g}}} & \qw & \qw & \qw & \qw & \gate[2]{e^\dagger} \\
    \setwiretype{n} & \hphantom{ext} & \gate[2]{\textsc{distr}} & \setwiretype{q} \qwbundle{\substack{\bigoplus_j \cH(\Gamma_j) \\ \\ \hfill}} & \gate{\msem{\ortho{T}{e_1, \dots, e_n}}\adj} & \qwbundle{\cH(T)} & & \qwbundle{\cH(\Delta)} & \gate[2]{\textsc{erase}} \\
	\lstick{$\cdots$} & \qwbundle{\bigoplus_j \cdots} & & \qwbundle{\cH(T')} & \qw & \qw & \qw & \qw & & \qw \rstick{$\cH(T')$} \\
\end{quantikz}
\]

\begin{alignat*}{2}
	&&& \ket{\sigma, \tau, \tau'} \\ &&&\in \cH(\Gamma, \Delta, \Delta') \\
	&\mapsto&\;& \ket{\sigma, \tau, \tau, \tau'} \\ &&&\in \cH(\Gamma, \Delta, \Delta, \Delta') \\
	&\mapsto&\;& \ket{\sigma} \otimes \msem{e} \ket{\sigma, \tau} \otimes \ket{\tau, \tau'} \\
	&=&& \ket{\sigma} \otimes \sum_{v \in \VV(T)} \bra{g_{\sigma,\tau,v}, v} \msem{e} \ket{\sigma, \tau} \cdot \ket{g_{\sigma,\tau,v}, v} \otimes \ket{\tau, \tau'} \\ &&&\in \cH(\Gamma) \otimes \cH_{\textsc{g}} \otimes \cH(T) \otimes \cH(\Delta, \Delta') \\
	&\mapsto&& \ket{\sigma} \otimes \sum_{v \in \VV(T)} \bra{g_{\sigma,\tau,v}, v} \msem{e} \ket{\sigma, \tau} \cdot \ket{g_{\sigma,\tau,v}} \otimes \left(\bigoplus_{j:\cR} \sum_{\sigma_j \in \VV(\Gamma_j)} \bra{\sigma_j} \msem{e_j}^\dagger \ket{v} \cdot \msem{e_j}^\dagger \ket{v} \right) \otimes \ket{\tau, \tau'} \\
	&=&& \ket{\sigma} \otimes \sum_{v \in \VV(T)} \bra{g_{\sigma,\tau,v}, v} \msem{e} \ket{\sigma, \tau} \cdot \ket{g_{\sigma,\tau,v}} \otimes \left(\bigoplus_{j:\cR} \sum_{\sigma_j \in \VV(\Gamma_j)} \bra{\sigma_j} \msem{e_j}^\dagger \ket{v} \cdot \bra{\sigma_j} \msem{e_j}^\dagger \ket{v} \cdot \ket{\sigma_j} \right) \otimes \ket{\tau, \tau'} \\
	&=&& \ket{\sigma} \otimes \sum_{v \in \VV(T)} \bra{g_{\sigma,\tau,v}, v} \msem{e} \ket{\sigma, \tau} \cdot \ket{g_{\sigma,\tau,v}} \otimes \left(\bigoplus_{j:\cR} \sum_{\sigma_j \in \VV(\Gamma_j)} \bra{\sigma_j} \msem{e_j}^\dagger \ket{v} \cdot \ket{\sigma_j} \right) \otimes \ket{\tau, \tau'} \\ &&&\in \cH(\Gamma) \otimes \cH_{\textsc{g}} \otimes \bigoplus_{j:\cR} \cH(\Gamma_j) \otimes \cH(\Delta, \Delta') \\
	&\mapsto&& \ket{\sigma} \otimes \sum_{v \in \VV(T)} \bra{g_{\sigma,\tau,v}, v} \msem{e} \ket{\sigma, \tau} \cdot \ket{g_{\sigma,\tau,v}} \otimes \left(\bigoplus_{j:\cR} \sum_{\sigma_j \in \VV(\Gamma_j)} \bra{\sigma_j} \msem{e_j}^\dagger \ket{v} \cdot \ket{\sigma_j, \tau, \tau'} \right) \\ &&&\in \cH(\Gamma) \otimes \cH_{\textsc{g}} \otimes \bigoplus_{j:\cR} \left(\cH(\Gamma_j, \Delta, \Delta')\right) \\
	&\mapsto&& \ket{\sigma} \otimes \sum_{v \in \VV(T)} \bra{g_{\sigma,\tau,v}, v} \msem{e} \ket{\sigma, \tau} \cdot \ket{g_{\sigma,\tau,v}} \otimes \left(\bigoplus_{j:\cR} \sum_{\sigma_j \in \VV(\Gamma_j)} \bra{\sigma_j} \msem{e_j}^\dagger \ket{v} \cdot \msem{e_j'} \ket{\sigma_j, \tau, \tau'} \right)
\end{alignat*}
\begin{alignat*}{3}
	&&&\ket{\sigma} \otimes \sum_{v \in \VV(T)} \bra{g_{\sigma,\tau,v}, v} \msem{e} \ket{\sigma, \tau} \cdot \ket{g_{\sigma,\tau,v}} \\ &&&\qquad\otimes \left(\bigoplus_{j:\cR} \sum_{\sigma_j \in \VV(\Gamma_j)} \bra{\sigma_j} \msem{e_j}^\dagger \ket{v} \cdot \sum_{v' \in \VV(T')} \bra{\sigma_j, v'} \msem{e_j'} \ket{\sigma_j, \tau, \tau'} \cdot \ket{\sigma_j, v'} \right) \\
	&=&&\ket{\sigma} \otimes \sum_{v \in \VV(T)} \bra{g_{\sigma,\tau,v}, v} \msem{e} \ket{\sigma, \tau} \cdot \ket{g_{\sigma,\tau,v}} \\ &&&\qquad\otimes \left(\sum_{j=1}^n \sum_{\sigma_j \in \VV(\Gamma_j)} \bra{\sigma_j} \msem{e_j}^\dagger \ket{v} \cdot \sum_{v' \in \VV(T')} \bra{\sigma_j, v'} \msem{e_j'} \ket{\sigma_j, \tau, \tau'} \cdot \inj_j \ket{\sigma_j, v'} \right) \\ &&&\in \cH(\Gamma) \otimes \cH_{\textsc{g}} \otimes \bigoplus_{j:\cR} \left(\cH(\Gamma_j) \otimes \cH(T')\right) \\
	&\mapsto&&\ket{\sigma} \otimes \sum_{v \in \VV(T)} \bra{g_{\sigma,\tau,v}, v} \msem{e} \ket{\sigma, \tau} \cdot \ket{g_{\sigma,\tau,v}} \\ &&&\qquad\otimes \left(\sum_{j=1}^n \sum_{\sigma_j \in \VV(\Gamma_j)} \bra{\sigma_j} \msem{e_j}^\dagger \ket{v} \cdot \inj_j^\cR \ket{\sigma_j} \otimes \sum_{v' \in \VV(T')} \bra{\sigma_j, v'} \msem{e_j'} \ket{\sigma_j, \tau, \tau'} \cdot \ket{v'} \right) \\ &&&\in \cH(\Gamma) \otimes \cH_{\textsc{g}} \otimes \bigoplus_{j:\cR} \cH(\Gamma_j) \otimes \cH(T') \\
	&\mapsto&&\ket{\sigma} \otimes \sum_{v \in \VV(T)} \bra{g_{\sigma,\tau,v}, v} \msem{e} \ket{\sigma, \tau} \cdot \ket{g_{\sigma,\tau,v}} \\ &&&\qquad\otimes \left(\sum_{j=1}^n \sum_{\sigma_j \in \VV(\Gamma_j)} \bra{\sigma_j} \msem{e_j}^\dagger \ket{v} \cdot \ket{v} \otimes \sum_{v' \in \VV(T')} \bra{\sigma_j, v'} \msem{e_j'} \ket{\sigma_j, \tau, \tau'} \cdot \ket{v'} \right) \\ &&&\in \cH(\Gamma) \otimes \cH_{\textsc{g}} \otimes \cH(T) \otimes \cH(T') \\
	&\mapsto&&\ket{\sigma} \otimes \sum_{v \in \VV(T)} \bra{g_{\sigma,\tau,v}, v} \msem{e} \ket{\sigma, \tau} \\ &&&\qquad\otimes \left(\sum_{j=1}^n \sum_{\sigma_j \in \VV(\Gamma_j)} \bra{\sigma_j} \msem{e_j}^\dagger \ket{v} \cdot \msem{e}^\dagger \ket{g_{\sigma,\tau,v},v} \otimes \sum_{v' \in \VV(T')} \bra{\sigma_j, v'} \msem{e_j'} \ket{\sigma_j, \tau, \tau'} \cdot \ket{v'} \right) \\ &&&\in \cH(\Gamma, \Delta) \otimes \cH(T') \\
\end{alignat*}
\begin{alignat*}{3}
	&&&\ket{\sigma} \otimes \sum_{v \in \VV(T)} \bra{g_{\sigma,\tau,v}, v} \msem{e} \ket{\sigma, \tau} \\ &&&\qquad\otimes \left(\sum_{j=1}^n \sum_{\sigma_j \in \VV(\Gamma_j)} \bra{\sigma_j} \msem{e_j}^\dagger \ket{v} \cdot \sum_{\tau_\star} \bra{\sigma,\tau_\star}\msem{e}^\dagger \ket{g_{\sigma,\tau,v},v} \cdot \ket{\tau_\star} \otimes \sum_{v' \in \VV(T')} \bra{\sigma_j, v'} \msem{e_j'} \ket{\sigma_j, \tau, \tau'} \cdot \ket{v'} \right) \\ &&&\in \cH(\Gamma, \Delta) \otimes \cH(T') \\
	&=&&\ket{\sigma} \otimes \sum_{v \in \VV(T)} \bra{g_{\sigma,\tau,v}, v} \msem{e} \ket{\sigma, \tau} \bra{\sigma,\tau}\msem{e}^\dagger \ket{g_{\sigma,\tau,v},v}\\ &&&\qquad\otimes \left(\sum_{j=1}^n \sum_{\sigma_j \in \VV(\Gamma_j)} \bra{\sigma_j} \msem{e_j}^\dagger \ket{v} \cdot \ket{\tau} \otimes \sum_{v' \in \VV(T')} \bra{\sigma_j, v'} \msem{e_j'} \ket{\sigma_j, \tau, \tau'} \cdot \ket{v'} \right) \\ &&&\in \cH(\Gamma, \Delta) \otimes \cH(T') \\
	&\mapsto&&\ket{\sigma} \otimes \sum_{v \in \VV(T)} \bra{g_{\sigma,\tau,v}, v} \msem{e} \ket{\sigma, \tau} \bra{\sigma,\tau}\msem{e}^\dagger \ket{g_{\sigma,\tau,v},v}\\ &&&\qquad\otimes \left(\sum_{j=1}^n \sum_{\sigma_j \in \VV(\Gamma_j)} \bra{\sigma_j} \msem{e_j}^\dagger \ket{v} \otimes \sum_{v' \in \VV(T')} \bra{\sigma_j, v'} \msem{e_j'} \ket{\sigma_j, \tau, \tau'} \cdot \ket{v'} \right) \\ &&&\in \cH(\Gamma) \otimes \cH(T') \\
\end{alignat*}

\textsc{T-PureApp}:
\begin{quantikz}
	\lstick{$\cH(\Gamma)$} & \ctrl{1} & \qw & \qw & \rstick{$\cH(\Gamma)$} \qw \\
	\lstick{$\cH(\Delta)$} & \gate{e} & \qwbundle{\cH(T)} & \gate{f} & \rstick{$\cH(T')$} \qw
\end{quantikz}
\begin{alignat*}{2}
	&&\;& \ket{\sigma, \tau} \in \cH(\Gamma, \Delta) \\
	&\mapsto&& \ket{\sigma} \otimes \msem{e} \ket{\tau} \in \cH(\Gamma) \otimes \cH(T) \\
	&\mapsto&& \ket{\sigma} \otimes \msem{f} \msem{e} \ket{\tau} \in \cH(\Gamma) \otimes \cH(T') \\
\end{alignat*}

\textsc{T-Mix}:
\begin{quantikz}
    \lstick{$\cH(\Gamma)$} & \ctrl{1} & \rstick{$\cH(\Gamma)$} \\
	\lstick{$\cH(\Delta)$} & \gate{e} & \rstick{$\cH(T)$} \qw
\end{quantikz}
\begin{alignat*}{2}
	&&\;& \op{\sigma, \tau}{\sigma, \tau'} \in \cL(\cH(\Gamma, \Delta)) \\
	&\mapsto&& \op{\sigma}{\sigma} \otimes \msem{e} \op{\tau}{\tau'} \msem{e}^\dagger \in \cL(\cH(\Gamma)) \otimes \cL(\cH(T)) \\
\end{alignat*}

\textsc{T-Discard}:
\begin{quantikz}
\lstick{$\cH(\Gamma)$} & \ctrl{1} & \rstick{$\Gamma$} \\
\lstick{$\cH(\Delta)$} & \gate{e} & \rstick{$\cH(T)$} \\
\lstick{$\cH(\Delta_0)$} & \trash{\cH(\Delta_0)}
\end{quantikz}
\begin{alignat*}{2}
	&&\;& \op{\sigma, \tau, \tau_0}{\sigma, \tau', \tau_0'} \in \cL(\cH(\Gamma, \Delta, \Delta_0)) \\
	&\mapsto&& \op{\sigma}{\sigma} \otimes \tr(\op{\tau_0}{\tau_0'}) \msem{e}(\op{\tau}{\tau'}) \in \cL(\cH(\Gamma)) \otimes \cL(\cH(T)) \\
\end{alignat*}

\textsc{T-MixedPair}:
\begin{quantikz}
	\lstick{$\cH(\Gamma)$} & \qw & \qw & \ctrl{1} & \ctrl{3} & \qw \rstick{$\cH(\Gamma)$} \\
	\lstick{$\cH(\Delta)$} & \ctrl{2} & \qwbundle{\cH(\Delta)} & \gate[2]{{\Gamma \partition \Delta, \Delta_0 \Vdash e_0 : T_0}} & \qw & \qw \rstick{$\cH(T_0)$} \\
	\lstick{$\cH(\Delta_0)$} & \qw & \qw & & \setwiretype{n} \\
     \setwiretype{n} & \gate[style={cloud}]{} & \setwiretype{q} \qwbundle{\cH(\Delta)} & \qw & \gate[2]{{\Gamma \partition \Delta, \Delta_1 \Vdash e_1 : T_1}} & \qw \rstick{$\cH(T_1)$} \\
	\lstick{$\cH(\Delta_1)$} & \qw & \qw & \qw & & \setwiretype{n}
\end{quantikz}
\begin{alignat*}{3}
	&&\;& \op{\sigma, \tau, \tau_0, \tau_1}{\sigma, \tau', \tau_0', \tau_1'} &&\in \cL(\cH(\Gamma, \Delta, \Delta_0, \Delta_1)) \\
	&\mapsto&& \op{\sigma, \tau, \tau_0, \tau, \tau_1}{\sigma, \tau', \tau_0', \tau', \tau_1'} \\
	&=&& \op{\sigma}{\sigma} \otimes \op{\tau,\tau_0}{\tau',\tau_0'} \otimes \op{\tau,\tau_1}{\tau',\tau_1'} &&\in \cL(\cH(\Gamma)) \otimes \cL(\cH(\Delta, \Delta_0)) \otimes \cL(\cH(\Delta, \Delta_1)) \\
	&\mapsto&& \op{\sigma}{\sigma} \otimes \msem{e_0} \left( \op{\tau,\tau_0}{\tau',\tau_0'} \right) \otimes \msem{e_1} \left( \op{\tau,\tau_1}{\tau',\tau_1'} \right) &&\in \cL(\cH(\Gamma) \otimes \cH(T_0 \otimes T_1)) \\
\end{alignat*}

\textsc{T-Try}:
	\[
\begin{quantikz}
    \lstick{$\cH(\Gamma)$} & \ctrl{2} &&&&&&& \rstick{$\cH(\Gamma)$} \\
	\lstick{$\cH(\Delta_0)$} & \gate{{\cptp(\Delta_0 \Vdash e_0 : T)}} & \qwbundle{\substack{\cH(T) \oplus \CC \\ \hfill}} &[1cm] \gate[2]{\textsc{distr}} & \qwbundle{\substack{\cH(T) \otimes (\cH(T) \oplus \CC) \oplus \cH(T) \oplus \CC \\ \\ \hfill}} & \gate{\texttt{left}^\dagger} & \qwbundle{\substack{\cH(T) \otimes (\cH(T) \oplus \CC) \oplus \cH(T) \\ \hfill}} & \gate[2]{\textsc{distr}} & \rstick{$\cH(T)$} \qw \\
	\lstick{$\cH(\Delta_1)$} & \gate{{\cptp(\Delta_1 \Vdash e_1 : T)}} & \qwbundle{\cH(T) \oplus \CC} & &&&&& \trash{\cH(T) \oplus \CC \oplus \CC}
\end{quantikz}
\]
For brevity, define:
\begin{align*}
	\rho_0' &\defeq \msem{\sigma : \Gamma \partition \Delta_0 \Vdash e_0 : T}(\rho_0) \\
	\rho_1' &\defeq \msem{\sigma : \Gamma \partition \Delta_1 \Vdash e_1 : T}(\rho_1)
\end{align*}

Then, the circuit acts as follows:
\begin{alignat*}{3}
	&&\;& \rho_0 \otimes \rho_1 \\
	&\mapsto&& \left(\rho_0' \oplus (\tr\left(\rho_0\right) - \tr(\rho_0'))\right) \otimes \left(\rho_1' \oplus (\tr\left(\rho_1\right) - \tr(\rho_1')) \right) \\
	&\mapsto&& \rho_0' \otimes \left( \rho_1' \oplus (\tr\left(\rho_1\right) - \tr(\rho_1'))\right) \oplus (\tr\left(\rho_0\right) - \tr(\rho_0')) \left( \rho_1' \oplus (\tr\left(\rho_1\right) - \tr(\rho_1')) \right) \\
	&\mapsto&& \rho_0' \otimes \left( \rho_1' \oplus (\tr\left(\rho_1\right) - \tr(\rho_1'))\right) \oplus (\tr\left(\rho_0\right) - \tr(\rho_0')) \rho_1' \\
	&\mapsto&& \rho_0' \otimes \left( \rho_1' \oplus (\tr\left(\rho_1\right) - \tr(\rho_1')) \oplus 0 \right) + \rho_1' \left(0 \oplus 0 \oplus (\tr\left(\rho_0\right) - \tr(\rho_0'))\right) \\
	&\mapsto&& \tr\left(\rho_1\right) \rho_0' + (\tr\left(\rho_0\right) - \tr(\rho_0'))\rho_1'
\end{alignat*}

\textsc{T-MixedApp}:
\begin{quantikz}
	\lstick{$\cH(\Gamma)$} & \ctrl{1} & \qw & \qw & \rstick{$\cH(\Gamma)$} \qw \\
	\lstick{$\cH(\Delta)$} & \gate{e} & \qwbundle{\cH(T)} & \gate{f} & \rstick{$\cH(T')$} \qw
\end{quantikz}
\begin{alignat*}{2}
	&&\;& \op{\sigma, \tau}{\sigma, \tau'} \in \cL(\cH(\Gamma, \Delta)) \\
	&\mapsto&& \op{\sigma}{\sigma} \otimes \msem{e} \left( \op{\tau}{\tau'} \right) \in \cL(\cH(\Gamma)) \otimes \cL(\cH(T)) \\
	&\mapsto&& \op{\sigma}{\sigma} \otimes \msem{f} \left( \msem{e} ( \op{\tau}{\tau'} ) \right) \in \cL(\cH(\Gamma)) \otimes \cL(\cH(T'))
\end{alignat*}

\textsc{T-Match}:

This construction is similar to that of \textsc{T-Ctrl}, but it does not perform uncomputation and instead discards quantum data. We use purified versions of the circuits for the $e_j'$, obtaining $\textsc{purify}(e_j') : \cH(\Gamma,\Gamma_j, \Delta, \Delta_1) \rarr \cH(\Gamma, \Gamma_j) \otimes \cH(T') \otimes \cH_j$. We will still write it as $\msem{e_j'}$ where it is clear from context we are referring to the purification. Here, $\cH_j$ is a ``garbage Hilbert space'' containing vectors
\[
\braces{\ket{g_{j, \sigma, \sigma_j, \tau, \tau_1, v'}} : \sigma \in \VV(\Gamma), \sigma_j \in \VV(\Gamma_j), \tau \in \VV(\Delta), \tau_1 \in \VV(\Delta_1), v' \in \VV(T')},
\]
such that
\[
\msem{e_j'}\ket{\sigma, \sigma_j, \tau, \tau_1} = \sum_{v' \in \VV(T')} \bra{v', g_{j,\sigma,\sigma_j,\tau,\tau_1,v'}} \msem{e_j'}\ket{\sigma,\sigma_j,\tau,\tau_1} \cdot \ket{v', g_{j,\sigma,\sigma_j,\tau,\tau_1,v'}}.
\]
Note that unlike in \textsc{T-Ctrl}, we do \emph{not} purify $e$. As for \textsc{T-Ctrl}, all direct sums are to be understood as being taken over the tree $\cR$ associated with the orthogonality judgment.

\begin{quantikz}
\lstick{$\cH(\Gamma)$} & & & \ctrl{1} &&&&& \gate[5]{\textsc{distr}} \\
\lstick{$\cH(\Delta)$} & \ctrl{2} & \qwbundle{\cH(\Delta)} & \gate[2]{e} & \qwbundle{\cH(T)} & \gate{\msem{\ortho{T}{e_1, \dots, e_n}}} & \qwbundle{\bigoplus_j \cH(\Gamma_j)} &&& \qwbundle{\bigoplus_j \cH(\Gamma, \Gamma_j, \Delta, \Delta_1)} &&& \rstick{$\cdots$} \\
\lstick{$\cH(\Delta_0)$} &&&& \setwiretype{n} \\
\setwiretype{n} & \gate[style={cloud}]{} & \setwiretype{q} \qwbundle{\cH(\Delta)} &&&&&& \\
\lstick{$\cH(\Delta_1)$} &&&&&&&&
\end{quantikz}

\begin{quantikz}
\setwiretype{n} &&&&&&&& \gate[3]{\textsc{distr}} & \setwiretype{q} \qw & \rstick{$\cH(\Gamma)$} \\
\lstick{$\cdots$} & \qwbundle{\substack{\bigoplus_j \cH(\Gamma, \Gamma_j, \Delta, \Delta_1) \\ \\ \hfill}} & \gate{\bigoplus_j \textsc{purify}(e_j')} & \qwbundle{\hspace{-12pt} \bigoplus_j \parens{\cH(\Gamma, \Gamma_j) \otimes \cH(T') \otimes \cH_j}} &&&&&&& \rstick{$\cH(T')$} \\
\setwiretype{n} &&&&&&&&& \setwiretype{q} \qw & \trash{\bigoplus_j \parens{\cH(\Gamma_j) \otimes \cH_j}}
\end{quantikz}

\begin{alignat*}{3}
&&& \op{\sigma, \tau, \tau_0, \tau_1}{\sigma, \tau', \tau_0', \tau_1'} \\
&&&\in \cL(\cH(\Gamma, \Delta, \Delta_0, \Delta_1)) \\
&\mapsto&\;& \op{\sigma, \tau, \tau_0, \tau, \tau_1}{\sigma, \tau', \tau_0', \tau', \tau_1'} \\
&&&\in \cL(\cH(\Gamma, \Delta, \Delta_0, \Delta, \Delta_1)) \\
&\mapsto&\;& \op{\sigma}{\sigma} \otimes \msem{e}\parens{\op{\sigma, \tau, \tau_0}{\sigma, \tau', \tau_0'}} \otimes \op{\tau_0, \tau, \tau_1}{\tau_0', \tau', \tau_1'} \\
&=&\;& \op{\sigma}{\sigma} \otimes \sum_{v, w \in \VV(T)} \bra{v} \parens{\msem{e}\parens{\op{\sigma, \tau, \tau_0}{\sigma, \tau', \tau_0'}}}\ket{w} \cdot \op{v}{w} \otimes \op{\tau, \tau_1}{\tau', \tau_1'} \\
&&&\in \cL(\cH(\Gamma)) \otimes \cL(\cH(T)) \otimes \cL(\cH(T)) \otimes \cL(\cH(\Delta, \Delta_1)) \\
&\mapsto&\;& \op{\sigma}{\sigma} \otimes \sum_{v, w \in \VV(T)} \bra{v} \parens{\msem{e}\parens{\op{\sigma, \tau, \tau_0}{\sigma, \tau', \tau_0'}}}\ket{w} \cdot \\
&&& \qquad \parens{\sum_{j,k} \sum_{\substack{\sigma_j \in \VV(\Gamma_j) \\ \sigma_k \in \VV(\Gamma_k)}} \bra{\sigma_j}\msem{e_j}\adj\ket{v} \bra{w}\msem{e_j}\ket{\sigma_k} \cdot \inj_j^\cR \msem{e_j}\adj\ket{v}\bra{w}\msem{e_k} \inj_k^{\cR\dagger}}
\otimes \op{\tau, \tau_1}{\tau', \tau_1'} \\
&=&\;& \op{\sigma}{\sigma} \otimes \sum_{v, w \in \VV(T)} \bra{v} \parens{\msem{e}\parens{\op{\sigma, \tau, \tau_0}{\sigma, \tau', \tau_0'}}}\ket{w} \cdot \\
&&& \qquad \parens{\sum_{j,k} \sum_{\substack{\sigma_j \in \VV(\Gamma_j) \\ \sigma_k \in \VV(\Gamma_k)}} \bra{\sigma_j}\msem{e_j}\adj\ket{v} \bra{w}\msem{e_j}\ket{\sigma_k} \cdot \inj_j^\cR \op{\sigma_j}{\sigma_k} \inj_k^{\cR\dagger}}
\otimes \op{\tau, \tau_1}{\tau', \tau_1'} \\
&&&\in \cL(\cH(\Gamma)) \otimes \bigoplus_{j:\cR} \cL(\cH(\Gamma_j)) \otimes \cL(\cH(\Delta, \Delta_1)) \\
&\mapsto&\;& \sum_{v, w \in \VV(T)} \bra{v} \parens{\msem{e}\parens{\op{\sigma, \tau, \tau_0}{\sigma, \tau', \tau_0'}}}\ket{w} \cdot \\
&&& \qquad \parens{\sum_{j,k} \sum_{\substack{\sigma_j \in \VV(\Gamma_j) \\ \sigma_k \in \VV(\Gamma_k)}} \bra{\sigma_j}\msem{e_j}\adj\ket{v} \bra{w}\msem{e_j}\ket{\sigma_k} \cdot \inj_j^\cR \op{\sigma, \sigma_j, \tau, \tau_1}{\sigma, \sigma_k, \tau', \tau_1'} \inj_k^{\cR\dagger}} \\
&&&\in \bigoplus_{j:\cR} \cL(\cH(\Gamma,\Gamma_j,\Delta,\Delta_1)) \\
\end{alignat*}
\begin{alignat*}{3}
&\mapsto&\;& \sum_{v, w \in \VV(T)} \sum_{v', w' \in \VV(T')} \bra{v} \parens{\msem{e}\parens{\op{\sigma, \tau, \tau_0}{\sigma, \tau', \tau_0'}}}\ket{w} \cdot \\
&&& \qquad \cdot \sum_{j,k} \sum_{\substack{\sigma_j \in \VV(\Gamma_j) \\ \sigma_k \in \VV(\Gamma_k)}} \bra{\sigma_j}\msem{e_j}\adj\ket{v} \bra{w}\msem{e_j}\ket{\sigma_k} \cdot
\bra{v', g_{j,\sigma,\sigma_j,\tau,\tau_1,v'}} \msem{e_j'} \ket{\sigma,\sigma_j,\tau,\tau_1} \cdot \\
&&& \qquad \cdot \bra{\sigma,\sigma_k,\tau',\tau_1'}\msem{e_k'}\adj \ket{w', g_{k,\sigma,\sigma_k,\tau',\tau_1',w'}} \cdot
\inj_j^\cR \op{\sigma, \sigma_j, v', g_{j,\sigma,\sigma_j,\tau,\tau_1,v'}}{\sigma, \sigma_k, w', g_{k,\sigma,\sigma_k,\tau',\tau_1',w'}} \inj_k^{\cR\dagger} \\
&&&\in \bigoplus_{j:\cR} \cL(\cH(\Gamma,\Gamma_j) \otimes \cH(T') \otimes \cH_j) \\
&\mapsto&\;& \op{\sigma}{\sigma} \otimes \sum_{v, w \in \VV(T)} \sum_{v', w' \in \VV(T')} \bra{v} \parens{\msem{e}\parens{\op{\sigma, \tau, \tau_0}{\sigma, \tau', \tau_0'}}}\ket{w} \cdot \\
&&& \qquad \cdot \sum_{j,k} \sum_{\substack{\sigma_j \in \VV(\Gamma_j) \\ \sigma_k \in \VV(\Gamma_k)}} \bra{\sigma_j}\msem{e_j}\adj\ket{v} \bra{w}\msem{e_j}\ket{\sigma_k} \cdot
\bra{v', g_{j,\sigma,\sigma_j,\tau,\tau_1,v'}} \msem{e_j'} \ket{\sigma,\sigma_j,\tau,\tau_1} \cdot \\
&&& \qquad \cdot \bra{\sigma,\sigma_k,\tau',\tau_1'}\msem{e_k'}\adj \ket{w', g_{k,\sigma,\sigma_k,\tau',\tau_1',w'}} \cdot
\op{v'}{w'} \otimes \inj_j^\cR \op{\sigma_j, g_{j,\sigma,\sigma_j,\tau,\tau_1,v'}}{\sigma_k, g_{k,\sigma,\sigma_k,\tau',\tau_1',w'}} \inj_k^{\cR\dagger} \\
&&&\in \cL(\cH(\Gamma)) \otimes \cL(\cH(T')) \otimes \bigoplus_{j:\cR} \cL(\cH(\Gamma_j) \otimes \cH_j) \\
&\mapsto&\;& \op{\sigma}{\sigma} \otimes \sum_{v, w \in \VV(T)} \sum_{v', w' \in \VV(T')} \bra{v} \parens{\msem{e}\parens{\op{\sigma, \tau, \tau_0}{\sigma, \tau', \tau_0'}}}\ket{w} \cdot \\
&&& \qquad \cdot \sum_{j=1}^n \sum_{\sigma_j \in \VV(\Gamma_j)} \bra{\sigma_j}\msem{e_j}\adj\ket{v} \bra{w}\msem{e_j}\ket{\sigma_j} \cdot
\bra{v', g_{j,\sigma,\sigma_j,\tau,\tau_1,v'}} \msem{e_j'} \ket{\sigma,\sigma_j,\tau,\tau_1} \cdot \\
&&& \qquad \cdot \bra{\sigma,\sigma_j,\tau',\tau_1'}\msem{e_j'}\adj \ket{w', g_{j,\sigma,\sigma_j,\tau',\tau_1',w'}} \cdot
\op{v'}{w'} \\
&=&\;& \op{\sigma}{\sigma} \otimes \sum_{v \in \VV(T)} \bra{v} \parens{\msem{\sigma : \Gamma \partition \Delta, \Delta_0 \Vdash e : T}\parens{\ket{\tau, \tau_0}\bra{\tau', \tau_0'}}} \ket{v} \cdot
\sum_{j=1}^n \sum_{\sigma_j \in \VV(\Gamma_j)} \bra{\sigma_j}\msem{\varnothing : \varnothing \partition \Gamma_j \vdash e_j : T}\adj \ket{v} \cdot \\
&&& \qquad \cdot \msem{\sigma, \sigma_j : \Gamma, \Gamma_j \partition \Delta, \Delta_1 \Vdash e_j'}\parens{\op{\tau, \tau_1}{\tau', \tau_1'}} \\
&&&\in \cL(\cH(\Gamma)) \otimes \cL(\cH(T'))
\end{alignat*}

\fi

\end{document}